\documentclass{article}
\usepackage{amsthm,amssymb, amsmath,dsfont} 
\usepackage{graphicx}
\usepackage{todonotes}
\usepackage{float}
\theoremstyle{plain}
\usepackage{lipsum, array, booktabs, caption}
\usepackage[ruled,vlined,linesnumbered]{algorithm2e}
\usepackage[noend]{algpseudocode}
\usepackage{xcolor,colortbl}
\usepackage{subcaption}
\usepackage{multicol}

\newcommand\unif{\operatorname{{unif}}}
\definecolor{LightCyan}{rgb}{0.88,1,1}
\newcolumntype{a}{>{\color{red}}r}
\usepackage{pdflscape}

\usepackage{color}
\usepackage{ifthen}
\usepackage{xspace}
\usepackage[utf8]{inputenc}
\usepackage{dsfont}
\usepackage{url}
\usepackage{braket}
\usepackage{tikz}
\usetikzlibrary{arrows}

\usepackage{geometry}
\usepackage{todonotes}

\theoremstyle{plain}
\newtheorem{theorem}{Theorem}[section]
\newtheorem{lemma}[theorem]{Lemma}
\newtheorem{assumption}[theorem]{Assumption}
\newtheorem{corollary}[theorem]{Corollary}
\newtheorem{proposition}[theorem]{Proposition}

\theoremstyle{definition}

\theoremstyle{remark}

\newcommand{\proofstring}{Proof}
\renewenvironment{proof}[1][]{\noindent\textbf{\proofstring\ifthenelse{\equal{#1}{}}{:}{~(#1) :}}\xspace}{\hfill$\blacksquare$\medskip\par}

\newcommand{\R}{\mathds{R}}

\newcommand{\E}{\mathds{E}}
\newcommand{\var}{\mathds{V}ar}
\newcommand{\Var}{\var}

\newcommand{\Prob}{\mathds{P}}
\newcommand{\Ind}{\mathds{1}}
\newcommand{\ind}{\Ind}

\usepackage{xcolor,colortbl}
\definecolor{LightCyan}{rgb}{0.88,1,1}
\newcolumntype{a}{>{\color{red}}r}
\newcolumntype{b}{>{\color{orange}}r}

\begin{document}

\title{A maximum-mean-discrepancy goodness-of-fit test for censored data}

\author{Tamara Fern\'andez\\
  Gatsby Computational Neuroscience Unit\\ 
  University College London\\
  \texttt{t.a.fernandez@ucl.ac.uk} \\
   \and
   Arthur Gretton \\
   Gatsby Computational Neuroscience Unit \\
  University College London\\
   \texttt{gretton@gatsby.ucl.ac.uk} \\
}
\maketitle
\begin{abstract}
  We introduce a kernel-based goodness-of-fit test for censored data, where observations may be missing in random time intervals: a common occurrence in clinical trials and industrial life-testing. The test statistic is straightforward to compute, as is the test threshold, and we establish consistency under the null.
Unlike earlier approaches such as the Log-rank test, we make no assumptions as to how the data distribution might differ from the null, and our test has power against a very rich class of alternatives. 
In experiments, our test outperforms competing approaches for periodic and Weibull hazard functions (where risks are time dependent), and
does not show the failure modes of tests that rely on user-defined features.
Moreover, in cases where classical tests are provably most powerful, our test performs almost as well, while being more general.
\end{abstract}

\section{Introduction}
Survival analysis is a branch of statistics focused on the study of time-to-event data, usually called survival times. This type of data usually appears in applications such as industrial life-testing, death times of patients in clinical trials or duration of unemployment in a population. An important characteristic of this type of data is that survival times may be censored, meaning that we do not observe the actual value of a survival time but a bound for it. For instance, it is not uncommon that in a clinical trial, the actual death time of a patient is only known to be within an interval of time. 

Arguably, the most common type of censoring is independent right-censoring which occurs when, instead of observing the actual survival time, say $X$, we observe a lower bound $T$ for it, i.e., we observe $T$ and we know that $X>T$. Other less common types of censoring mechanisms are independent left and interval censoring. Respectively, left and interval censoring arise when we observe either an upper bound $T$ instead of the failure time $X$ or an interval $(T_l,T_u)\subseteq \R_+$ in which the failure time $X$ falls. A reasonable assumption we make is that the censoring mechanisms are non-informative about the distribution of the survival times $X$.   We will provide a more extensive description of our setting and terminology in Section \ref{sec:problemSetting}.

While in most statistical/machine-learning applications the cornerstone for analysing data is the distribution function, in survival analysis the main objects of study are the hazard function and the survival function. For a survival time $X$ with density $f$ and distribution $F$, its survival function $S$ is defined by $1-F$, and its hazard function $\lambda$ by $f/S$. While the survival $S(t)$ function gives us the probability a patient survives up to time $t$, the hazard $\lambda(t)$ function is the instant risk of death given that she has survived until time $t$. Additionally, we define the cumulative hazard function by $\Lambda(t) = \int_0^t \lambda(x)dx$. It can be check that $S(t) = e^{-\Lambda(t)}$ and thus $S$ and $\lambda$ are in a 1-1 correspondence. 

The hazard function is extremely important in applications, and different families of hazard functions give rise to different problems in the area. Examples of important families are proportional hazards, which in a clinical trial may represent treatments with constant benefit/dis-benefit over time (when compared with the baseline), and crossing hazards, representing treatments that may have a negative impact at the beginning but long-term benefits, e.g. chemotherapy, among other behaviours. Distinguishing between different hazard functions is a fundamental problem in survival analysis.

In this paper, we study goodness of fit, i.e. the problem of testing the null hypothesis $H_0:\lambda=\lambda_0$, or alternatively, $\mathbb{H}_0:S=S_0$, where $\lambda_0$ and $S_0$ denote some specified  hazard and survival functions in the setting of independent right censoring. 

A few methodologies have been proposed to attack this problem. The  Log-rank test is the most popular among practitioners. This test is based on the simple idea of comparing the cumulative hazard function under the null against the empirical cumulative hazard function. Among the good properties of this test, we have that the Log-rank test is the most powerful test for proportional hazard alternatives. Unfortunately, when the true relationship of the hazards is time-dependent it may lead to wrong decisions, i.e. low power \cite{aalenbook}. An option to increase the power of Log-rank tests against time-dependent alternatives is to consider weighted Log-rank tests. These tests have been extensively studied \cite{breslow1975analysis,hollander1979testing,harrington1982class}; see \cite{aalenbook,klein2006survival} for details. By choosing an appropriate weight function, the weighted Log-rank test can be tailored to be optimal under specific alternatives, at the expense of reduced performance against other alternatives. Modern approaches attempt to increase overall test power by considering the combination of several weighted Log-rank tests into a single test-statistic, e.g. \cite{bathke2009combined,BrendelJanssen,ditzhaus2018more}. Nonetheless, weight-based approaches require us to hand-design the  weight functions in advance, in anticipation of a particular set of alternatives. Moreover, the amount of data required by the test grows with the number of weights chosen.

As an alternative to log-rank tests, there exist a number of  Chi-squared tests under censoring \cite{mihalko1980chi}, \cite{akritas1988} and \cite{hollander1992chi}, where the space is first partitioned, and  the empirical probability of uncensored events is then compared in chi squared distance with its expectation (the latter requiring an estimate of the censoring distribution).  See \cite{aalenbook} and \cite{FlemingBook} for more detail. 

Yet another approach, described in \cite{BagIoaKal13}
 is based on defining a {\em kernel density estimate} for the survival function, i.e. for $S=1-F$ (obtained using a slightly modified Kaplan Meier procedure), with the test statistic then defined as the squared difference between this density estimate and the model density.  Since this procedure relies on density estimation as an intermediate step, it has been found to be relatively data-inefficient, compared with more direct tests (see e.g. the recent discussion in  \cite{ditzhaus2018more}). We have independently
 confirmed this issue in our own experiments.

 In the present work, we propose a new goodness of fit test for censored data, based on distances between probability distribution embeddings in a reproducing kernel Hilbert space (RKHS)
 \cite[Chapter 4]{BerTho04}, \cite{SmoGreSonSch07}.  A particular challenge arises due to the unknown censoring distribution: correcting for this directly using e.g. a Kaplan-Meier estimate of the survival time \cite{kaplan1958nonparametric} leads
to a more complex estimator when compared to the
uncensored case, making standard tasks as bootstrapping or computing limiting distributions very difficult. Indeed, in this setting, naively applying the Kaplan-Meier estimator together with standard kernel test \cite{BaringhausHenze88,SzeRiz05} would lead to a potential incorrectly calibrated test.

Instead, we construct a sample mapping that requires no such correction, which we describe in Section \ref{sec:constructionNullDistribution}.
 We emphasise that our approach does not require evaluation nor integration of the hazard function under the null, which can be challenging.
 In  Section \ref{sec:kernelBasedTest} we define 
a test statistic for these transformed data
based on the maximum mean discrepancy \cite{GreBorRasSchetal12}, which is the RKHS distance between two distribution embeddings. The resulting test statistic is a simple V-statistic, and hence the test threshold can  readily be obtained using a wild bootstrap procedure.

In Section \ref{sec:experiments}, we illustrate the performance of our model for a number of use-cases, both for proportional hazards (where the classical Log-rank test is provably most powerful) and for time-varying hazards, including periodic and Weibull hazard functions.  In the case of proportional hazards, we perform almost as well as the most powerful model, despite our test being more general; in the case of periodic hazards, we greatly outperform alternative approaches and in the case of Weibull hazards we also outperform alternative approaches.
\section{Problem setting}\label{sec:problemSetting}

We briefly explain the setting of censored data, and associated challenges. Let $X_1,\ldots,X_n$ be a random sample from a continuous distribution of interest $F$ on $\R_+$ and, independent from this sample, consider $C_1,\ldots,C_n\overset{i.i.d.}{\sim}G$ from a nuisance distribution $G$. The data we observe correspond to the pairs $(T_1,\Delta_1),\ldots (T_n,\Delta_n)$, where $T=\min\{X,C\}$ and $\Delta=\ind\{T=X\}$. This data-framework corresponds to \textit{independent right censoring}.

Given this data, the task of estimating the cumulative distribution function $F=1-S$ of the failure times can not be solved by using the empirical distribution. Indeed, the empirical estimator given by $H_n(x)=\sum_{T_i\leq x}\frac{1}{n}$ approximates the distribution of the minimum of $X$ and $C$, i.e., $H=1-(1-F)(1-G)$. Alternatively, if we drop all the observations that we know are censored, that is, where $\Delta_i=0$ and thus $T_i=\min\{X_i,C_i\}=C_i$, the empirical distribution $H^1_n(x)=\sum_{T_i\leq x,\Delta_i=1}\frac{1}{n}$ approximates the  function $H^1(x)=\int_0^x (1-G(t))dF(t)$. The lack of natural empirical-type estimators for $F$ takes out of competition all the testing approaches which heavily rely on them. 

The non-parametric maximum likelihood estimator of $F$ under independent right censoring is the Kaplan-Meier estimator \cite{kaplan1958nonparametric}, defined as $\bar{F}_n(x)=\sum_{T_{[i:n]}\leq x}W_i$, where $W_i=\frac{\Delta_{[i:n]}}{n}\prod_{j=1}^{i-1}\left(1+\frac{1-\Delta_{[j:n]}}{n-j}\right)$,  $T_{[i:n]}$ denotes the $i$-th order statistic of the sample $\{T_i\}_{i=1}^n$, and $\Delta_{[i:n]}$ is its corresponding censoring indicator. In the particular case in which all the observations are uncensored, the Kaplan-Meier estimator simplifies to the empirical estimator $\bar{F}_n(x)=\hat{F}_n(x)=\sum_{X_i\leq x}\frac{1}{n}$. An immediate drawback of the Kaplan-Meier estimator is  that it can not be written as the sum of independent random variables (note that $W_i$ depends explicitly on the first $i$ data points), and all statistical approaches based on this estimator must address this. An example of this is that not even the standard Central Limit Theorem (CLT) can be applied naively as it is defined for independent random variables (a CLT result for the Kaplan-Meier estimator was proved much later than the standard CLT by Stute  \cite{stute1995} and Akritas \cite{akritas2000}).
\section{Construction of a null distribution}\label{sec:constructionNullDistribution}
Consider the independent right-censoring scheme. Under the null hypothesis $\mathbb{H}_0:F= F_0$, it holds $F_0(X_i)\sim\mathcal{U}(0,1)$, then testing the null hypothesis is equivalent to test for $\mathbb{H}_0:FF_0^{-1}=F_{\unif}$, where $F_{\unif}$ denotes the uniform distribution function on $(0,1)$. Notice that since we have right censored data we do not observe the failure time $X_i$ but instead we observe $T_i=\min(X_i,C_i)$. Nevertheless, since $F_0$ is increasing (it is a distribution function), it holds $F_0(T_i)=F_0(\min\{X_i,C_i\})=\min\{F_0(X_i),F_0(C_i)\}$, and thus the indicator function $\Delta_i$ is consistent with the order of $F_0(X_i)$ and $F_0(C_i)$. Then, we transform our initial problem into testing whether $\{F_0(X_i)\}_{i=1}^n$ follows a uniform distribution based on the right censored data $\{U_i,\Delta_i\}_{i=1}^n$, with $U_i=F_0(T_i)$. For left and interval censoring the same argument applies. 

Up to this point, we have just transformed our problem to test for uniformity, but we still need to deal with the censored data. We overcome this problem by introducing an estimator of the distribution function $FF_0^{-1}$, based on the censored data $\{U_i,\Delta_i\}_{i=1}^n$ which can be written as the sum of independent random variables, and which is unbiased under the null hypothesis. The form of this estimator allows us to use the classical theory of U-statistics to derive the asymptotic distributions related to our test approach.

For right censored data, whenever $\Delta_i=1$, we know that $U_i=F_0(X_i)$ which is exactly the random variable we are interested in. Then, by following the approach of the empirical distribution, we put a point mass of size $1/n$ on $U_i$. Otherwise, if $\Delta_i=0$ then $U_i=F_0(C_i)$. From this event, the only information we can deduce is $F_0(X_i)>F_0(C_i)$. To reflect our lack of information, we distribute the weight $1/n$ associated to $U_i$ uniformly on $(U_i,1)$ . The estimate we just described corresponds to
\begin{eqnarray}
\tilde{F}(x)=\frac{1}{n}\sum_{U_i\leq x}\Delta_i+(1-\Delta_i)\frac{x-U_i}{1-U_i}.\label{eqn:alternative estimator}
\end{eqnarray}
It is clear that this estimator can be generalized to deal with left and interval censoring by  distributing the weights associated to these censored random variables uniformly over their corresponding intervals: in the case of left censoring, uniformly on the interval $(0,U_i)$; and for interval censoring, uniformly on $(U_{i,l},U_{i,u})$, where $U_{i,l}$ and $U_{i,u}$ denote a lower an upper limit for the true value $F_0(X_i)$. The next proposition, whose proof is deferred to the supplementary material, establishes that our estimator is unbiased.

\begin{proposition}\label{Prop: Unbiased}
Under the null hypothesis, the estimator $\tilde{F}$, based on the data $\{U_i,\Delta_i\}_{i=1}^n$, is an unbiased estimator of the uniform distribution function $F_{\unif}$.
\end{proposition}

\section{A kernel-based test}\label{sec:kernelBasedTest}

In defining a statistic for testing goodness of fit for censored data, we make use of kernel distribution embeddings: that is, embeddings
of probability measures to a reproducing kernel Hilbert space (RKHS) \cite[Chapter 4]{BerTho04}, \cite{SmoGreSonSch07}.
The distance between distribution embeddings is denoted the maximum mean discepancy (MMD) \cite{GreBorRasSchetal12}.

Goodness of fit can be tested using  the MMD between the model and sample:
this is the approach in \cite{BaringhausHenze88,SzeRiz05}, bearing in mind the equivalence
 of distance and kernel-based measures of divergence \cite{SejSriGreFuk13}.\footnote{We note that an alternative approach is to modify the RKHS using a Stein operator,
 yielding a class of functions with zero expectation under the model:
 \cite{Chwialkowski2016,LiuLeeJor2016}. Yet another approach would be to define a Fisher statistic
 between the model and sample in an RKHS
  \cite{HarBacMou08,BalLiYua17}.}
 This approach is inappropriate for our setting, given the censoring.
Instead, we will use as our statistic the MMD between a uniform
 distribution and the sample-based distribution introduced in eq. (\ref{eqn:alternative estimator}).

We begin with the reproducing kernel Hilbert space $(\mathcal{H},\langle\cdot,\cdot\rangle_{\mathcal{H}})$ of functions from $[0,1]\to\R$, with kernel $K:[0,1]\times[0,1]\to\R$ such that $K(\cdot,x)\in\mathcal{H}$ for all $x\in[0,1]$, and $f(x)=\langle f(\cdot),K(\cdot,x)\rangle_{\mathcal{H}}$ for each $f\in\mathcal{H}$ and $x\in[0,1]$. 
The mean embedding of a probability measure $P$ on the RKHS $\mathcal{H}$  is 
$\mu_P(\cdot)=\E_P(K(\cdot,X)),$
where the expectation is to be understood in terms of the Bochner integral.
In particular, $\mu_P$ is well-defined whenever $\E_P(\sqrt{K(X,X)})<\infty$, which is guaranteed under the following assumption:
\begin{assumption}\label{Assumption: bounded}
There  exists a constant $M \geq 1$ such that $|K(x,y)|\leq M$ for all $x,y \in [0,1]$.
\end{assumption}
For a sufficiently rich RKHS, mean embeddings are injective, and uniquely represent their
respective probability measures \cite{SriGreFukSchetal10}: such RKHS are called {\em characterisitc}.
The exponentiated quadratic kernel used in the present work satisfies this property.

We base our test on the maximum mean discrepancy between the uniform distribution $F_{\unif}$ and $FF_0^{-1}$, i.e., $MMD(F_{\unif},FF_0^{-1}):=\|\mu_{FF_0^{-1}}-\mu_{\unif}\|_{\mathcal{H}}$, where we use the estimator $\tilde F$ of equation \eqref{eqn:alternative estimator} for $FF_0^{-1}$.

We now proceed to study the asymptotics of  $MMD(F_{\unif},\tilde{F}(x))$.
As we will see, the main advantage of our goodness-of-fit estimator is that it can be expressed as the sum of independent random variables, and thus it allows us to use standard machinery from the theory of U-statistics
\cite{Serfling80} in deriving the asymptotic properties of our statistic.

From the definition of the kernel mean embedding and by the reproducing kernel property, it holds that
\begin{eqnarray}
MMD(F_{\unif},\tilde{F})^2=\|\mu_{\unif}-\mu_{\tilde{F}}\|^2_{\mathcal{H}}=\int_0^1\int_0^1K(x,y)(dx-d\tilde{F}(x))(dy-d\tilde{F}(y)).
\end{eqnarray}
Recall the definition of $\tilde{F}(x)=\sum_{i=1}^n h_{U_i,\Delta_i}(x)$, where $h_{U_i,\Delta_i}(x)=\ind_{\{U_i\leq x\}}\left(\Delta_i+(1-\Delta_i)\frac{x-U_i}{1-U_i}\right)$, then $d\tilde{F}(x)$ is given by
\begin{eqnarray}
d\tilde{F}(x)
&=&\frac{1}{n}\sum_{i=1}^n\ind_{\{U_i<x\}}\frac{1-\Delta_i}{1-U_i}dx+\Delta_i\delta_{U_i}(x),\label{eqn:diff}
\end{eqnarray}
where $\delta_U(x)$ denotes a delta measure on $U$. Based on the reproducing kernel $K$, we define the $U$-statistic kernel $J:([0,1]\times\{0,1\})^2\to\R$ as
\begin{eqnarray}
J((u,\delta),(u',\delta'))
&=&\int_0^1\int_0^1K(x,y)(dx-dh_{u,\delta}(x))(dy-dh_{u',\delta'}(y)).\label{eqn:V-stat kernel} 
\end{eqnarray}

Note that the U-statistic kernel $J:\mathbb{R}\times\{0,1\}\to\mathbb{R}$  depends neither on $F_0$ nor on the censoring distribution $G$; i.e., its implementation is distribution free, therefore it need only be computed once. By contrast, as we will see in Section \ref{sec:competing}, Log-rank and Pearson tests require us to evaluate and integrate the hazard function $\lambda_0=f_0/(1-F_0)$, which may be not trivial as $\lambda_0$ may not be easily computable. Moreover, since the test statistic depends explicitly on the function $\lambda_0$, it needs to be recomputed  for each different hypothesis.

\begin{proposition}\label{Prop: Consistency}
Denote by $Q_0$ and $Q_a$ the distribution of the pair $(U,\Delta)$ under the null and alternative hypotheses, respectively. Then, under assumption \ref{Assumption: bounded} and $\epsilon>0$, it holds 
\footnotesize{
\begin{eqnarray}
\Prob_{Q_0}\left(\left|MMD^2(\tilde{F},F_{\unif}))-\E_{Q_0}MMD^2(\tilde{F},F_{\unif}))\right|\geq \epsilon\right)
&\leq& \exp\left\{-\frac{\epsilon^2n}{32M}\right\}\nonumber
\end{eqnarray}
}
and $\E_{Q_0}MMD^2(\tilde{F},F_{\unif}))\leq 4M/n$. The same result holds when replacing $Q_0$ by $Q_a$, and $F_{\unif}$ by the mean measure $\E(\tilde F)$, respectively.
\end{proposition}

\begin{proposition}\label{Prop: MMD U}
The $MMD(F_{\unif},\tilde{F})$ can be written as a $V$-statistic of degree 2 in the product space of survival-times and censoring-indicators. In particular
\begin{eqnarray}
MMD(F_{\unif},\tilde{F})=\frac{1}{n^2}\sum_{i=1}^n\sum_{j=1}^nJ((U_i,\Delta_i),(U_j,\Delta_j)),\label{eqn: V-statistic}
\end{eqnarray}

where $J$ is the $U$-statistic kernel defined in equation \eqref{eqn:V-stat kernel}. Moreover, under the null hypothesis, an unbiased estimate of $MMD(F_{\unif},FF_0^{-1})$ corresponds to the $U$-statistic
\begin{eqnarray}
MMD_U(F_{\unif},\tilde{F})={{n}\choose{2}}^{-1}\sum_{i< j}J((U_i,\Delta_i),(U_j,\Delta_j)).\label{eqn: Ustatsic}
\end{eqnarray}
\end{proposition}
\begin{lemma}\label{Lemma: nMMD(x,tilde F)}
Under the null hypothesis and under assumption \ref{Assumption: bounded}, it holds
\begin{eqnarray}
nMMD(F_{\unif},\tilde{F})\overset{\mathcal{D}}{\to}Y+\E_{Q_0}(J(U,\Delta),(U,\Delta)),
\end{eqnarray}
where $Y=\sum_{j=1}^\infty\lambda_j(\xi_j^2-1)$ and $\xi_j$ are independent normal random variables with zero mean and unit variance. Moreover, $\lambda_j$ are the eigenvalues of the linear transformation $T:([0,1]\times\{0,1\},Q_0)\to  L_2([0,1]\times\{0,1\},Q_0)$ given by 
$$(Tg)(u,\delta)=\E_{Q_0}(J((U,\Delta),(u,\delta))g(U,\Delta)),$$ 
where $Q_0$ denotes the joint distribution of the pair $(U,\Delta)$ under the null hypothesis.
\end{lemma}
\begin{corollary}Under the null hypothesis
$nMMD_U(F_{\unif},\tilde{F})\overset{\mathcal{D}}{\to}Y$,
with $Y$ defined as in Lemma \ref{Lemma: nMMD(x,tilde F)}. Under the alternative, when $ E_{Q_a}(\tilde{F})\neq F_{\unif}$, $\sqrt{n}MMD_U(F_{\unif},\tilde{F})$ is asymptotically normal.
\end{corollary}
The proof of our results are given in the supplementary material. 

Although we have an expression for the asymptotic distribution of our test statistic, the parameters of this distribution are hard to compute. Instead, we propose to use the wild bootstrap to approximate the rejection regions of our test statistic under the null hypothesis.

Specifically,  we can re-sample from the distribution of our statistic $MMD(F_{\unif},\tilde{F}_n)$ by repeatedly sampling
\begin{eqnarray}
\frac{1}{n^2}\sum_{i=1}^n\sum_{j=1}^n\mathcal{W}_i\mathcal{W}_jJ((U_i,\Delta_i),(U_j,\Delta_j)),\label{eqn: Wild boot version}
\end{eqnarray} 
where $\{\mathcal{W}_i\}_{i=1}^n$ is a sequence of independent random variables with zero mean (to preserve the degeneracy property) and variance one. It was proved in \cite{dehling1994random} that \eqref{eqn: Wild boot version} has the same limit distribution as our test-statistic $MMD(F_{\unif},\tilde{F}_n)$. 

\begin{algorithm}
  \KwIn{data $\{U_i,\Delta_i\}_{i=1}^n$}
  Consider $\mathcal{W}_1,\ldots,\mathcal{W}_n$ a random sample with $\E(W_i)=0 $ and $\Var(\mathcal{W}_i^2)=1$\\
Return $MMD^B_k=\frac{1}{n^2}\sum_{i=1}^n\sum_{j=1}^n\mathcal{W}_i\mathcal{W}_jJ((U_i,\Delta_i),(U_j,\Delta_j))$
\caption{Wild bootstrap.}\label{Algorithm}
\end{algorithm}

Given a nominal $\alpha$ value, the rejection region can be approximated by using a bootstrap sample $\{MMD^B_k\}_{k=1}^N$, as shown in Algorithm \ref{Algorithm}. 

\section{Competing approaches}\label{sec:competing}
\subsection{Pearson-type goodness-of-fit} 
This approach, due to Akritas  \cite{akritas1988},  considers the partition of the sample space of the random variables $(T_i,\Delta_i)$, for $\Delta\in\{0,1\}$, into $k$ cells, and studies the distribution of the $k$-dimensional vector of observed minus expected frequencies. Let $A_j=[a_{j-1},a_j)$ with $j=1,\ldots,k$ and  $0=a_0<a_1<\ldots<a_{k-1}<a_k=\infty$. We define the observed frequencies as
$N_{1,j}=\sum_{i=1}^n\ind\{T_i\in A_j,\Delta_i=1\}$ and the expected frequencies as $p_{1,j}=\int_{A_j}(1-G)dF_0$.  Observe that the expected frequencies depend on the unknown censoring distribution $G$, thus they are estimated by replacing the distribution $G$ by the estimate $\hat{G}=1-(1-\hat{H})/(1-F_0)$, where $\hat{H}$ is the empirical distribution of the minimum $T=\min\{X,C\}$. Estimators for the expected frequencies are then given by $\hat{p}_{1,j}=\int_{A_j}(1-\hat{G})dF_0=\int_{A_j}(1-\hat{H})d\Lambda_{0}$.
Under the null hypothesis $\mathbb{H}_0: F=F_0$, where $F_0$ is a specified continuous function, the Pearson-type statistic corresponds to
$\sum_{j=1}^k(N_{1j}-n\hat{p}_{1j})^2/(n\hat{p}_{1j})$, whose asymptotic distribution is $\chi^2_k$. 

\subsubsection{Implementation} For the Pearson test, ``Pearson" in the tables, we consider $k=4$ cells, each of them accumulating $0.25$ probability under the null hypothesis. The reason for  using $k=4$ cells is due to the trade-off between the number of data points required and the number of cells used: as the number of cells increases more data is needed to achieve the correct level of the test (i.e. Type-I error). We found that for $k=4$, the test is competitive even though the Type-I error is a little bit overestimated for small sample sizes ($30,50$ data points).

\subsection{Log-rank test and Weighted Log-rank tests}
Arguably, the Log-rank test is the most commonly-used statistical test for comparing survival curves. The test is performed by comparing differences of area as $Z=\int_0^\infty Y(t)d(\Lambda(t)-\Lambda_0(t))$,
where $Y(t)$ denotes the so-called risk function, given by $Y(t)=\sum_{i=1}^n\ind\{T_i\geq t\}$. The true cumulative hazard function, $\Lambda$, is estimated by the non-parametric Nelson-Aalen estimator. From the theory of counting processes \cite{FlemingBook}, $Z$ has an asymptotically normal distribution under appropriate scaling.

Weighted Log-rank tests generalize the classical Log-rank test, by considering an extra weight function $W$ when comparing areas, giving the more general statistic $Z(W)=\int_0^\infty W(t)Y(t)d(\Lambda(t)-\Lambda_0(t))$.

\subsubsection{Implementation} We consider weighted Log-rank tests, ``LR1" and ``LR2" in the tables, with weight functions $W_1(t)=1$ and $W_2(t)=Y(t)$, respectively ($Y(t)$ risk function). Observe that $Z(W_1)$ is the classical Log-rank test. The second statistic $Z(W_2)$ corresponds to the so-called Gehan-Breslow test \cite{aalenbook}.

\subsection{Combined Log-rank tests}
These tests   automate the procedure of choosing weight functions to create a more flexible test with broader power.
We describe here the approach of \cite{ditzhaus2018more}. A vector of weighted Log-rank tests $\boldsymbol{Z}=(Z(\omega_1\circ\hat{F}_n),\ldots,Z(\omega_k\circ\hat{F}_n))$ is defined, with weight functions $\omega_1,\ldots,\omega_k$.  The $\omega_i'$s are continuous and of bounded variation, and $\hat{F}_n$ is the Kaplan-Meier estimator. A combined-Log-rank test-statistic is computed as $S_n=\boldsymbol{Z}^\intercal\hat\Sigma^{+}\boldsymbol{Z}$, where $\hat{\Sigma}$ is the empirical covariance matrix of $\boldsymbol{Z}$, and $\Sigma^+$ represents the pseudo-inverse of $\Sigma$. Under some regularity conditions it is shown that $S_n\to\chi_k^2$ as $n$ grows to infinity.

\subsubsection{Implementation} We consider four different functions for the combined Log-rank test, denoted as ``WLR" in the tables. These functions correspond to i) the constant weight function, $\omega_1(t)=1$, which weights all time points equally, ii) an early weight function, $\omega_2(t)=t(1-t)^3$, which gives more weights to departures of the null at early times, iii) a central weight function, $\omega_3(t)=(1-t)t$, giving more weight at central times, iv) and a crossing weight function $\omega_4(t)=1-2t$, which has a sign switch at $t=1/2$. 

\subsection{Kernel test}
In the approach of \cite{BagIoaKal13},  a modified Kaplan-Meier estimate of the density was used for goodness of fit testing.
As noted already in \cite{ditzhaus2018more}, however, the procedure suffers from low test power at reasonable sample sizes,
which we have also confirmed independently (see Appendix, Section 2). For this reason, our experiments in the next section
will focus on the chi-squared and log-rank approaches.

\subsubsection{Implementation} For the implementation of this test, we use the code made available by the authors.

\section{Experiments}\label{sec:experiments}

For our MMD-based test, we choose the kernel to be $K_l(x,y)=\exp\{(x-y)^2/l^2\}$. For the length-scale parameter, we either use  $l=1$ (for simple settings),  or (for more complex settings) $\hat{l}_n$ computed by taking the median of the pair-wise differences of  survival times, without making a distinction between censored or uncensored time points. For the estimation of the rejection regions, we implement wild bootstrap as described in Algorithm \ref{Algorithm}. In particular, we consider three different sets of random variables $\{\mathcal{W}_i\}_{i=1}^n$: i)$\mathcal{W}_i\overset{i.i.d.}{\sim}N(0,1)$, ii) $\mathcal{W}_i\overset{ }{\sim}Multinomial(n,1/n,\ldots,1/n)-1$ and iii) $\mathcal{W}_i\overset{i.i.d.}{\sim}Rademacher$. Each test is denoted by ``MW1", ``MW2" and ``MW3" in the tables, respectively.

In all our experiments, we consider the null hypothesis to be $H_0:\lambda(t)=1$, or alternatively $H_0:S(t)=e^{-t}$. 
Then, as the data for different experiments is very similar under the null (we just vary the censoring distribution), we only show results in the main body for the estimated Type-I error for our first experiment: parallel hazards. These results are presented in Tables \ref{Table 1:null30}. (See sections 3,4 and 5 of the supplementary material for the remaining tables, including more sample sizes and censoring parameters.)

\begin{multicols}{2}
\begin{table}[H]
\centering
\resizebox{1\columnwidth}{!}{
\begin{tabular}{lrrrrrrr}
\hline
  \multicolumn{8}{c}{Type-I error; Censoring $30\%$; Fixed length-scale 1}\\
  \hline
 $\boldsymbol{\alpha}$& MW1 & MW2 & MW3 & Pearson & LR1 & LR2 &WLR\\ 
  \hline
  \multicolumn{8}{l}{Sample size $n=30$} \\
10 \% & 10.10 & 11.70 & 10.20 &\color{red}{14.50} & 11.15 & 11.10 &\color{red}{19.65}\\ 
  5 \% & 5.10 & 6.55 &4.85& \color{red}{8.70} & 5.90 & 6.60 &\color{red}{13.65}\\ 
  1 \% &1.15 & 2.05 & 1.10 &\color{red}{3.05} & 1.60 & 1.30 &\color{red}{6.30}\\ 
  \hline
  \multicolumn{8}{l}{Sample size $n=200$} \\
  10 \% & 9.35 & 9.55 & 9.45 & 9.80 & 10.05 & 9.80 & 11.40 \\ 
  5 \% & 4.90 & 4.85 & 4.75 & 5.65 & 5.10 & 5.10 & 6.85 \\ 
  1 \% & 1.15 & 1.35 & 1.20 & 1.30 & 1.25 & 1.30 & 1.80 \\ 
  \hline
\end{tabular}
}
\caption{{Estimated Type-I error, $\alpha\in\{10\%,5\%,1\%\}$. Sample size $n=30,200$. Censoring percentage $30\%$. Fixed-length-scale 1}}\label{Table 1:null30}
\end{table}

\begin{table}[H]
\centering
\resizebox{\columnwidth}{!}{
\begin{tabular}{lrrrrrrr}
\hline
  \multicolumn{8}{c}{Type-I error; Censoring $50\%$; Fixed length-scale 1}\\
  \hline
 $\boldsymbol{\alpha}$& MW1 & MW2 & MW3 & Pearson & LR1 & LR2 &WLR\\ 
  \hline
  \multicolumn{8}{l}{Sample Size $n=30$} \\
  
  10 \% & 10.50 & 11.80 & 10.50 & \color{red}{15.30} & 10.50 & 10.25 &\color{red}{19.25}\\ 
  5 \% & 4.45 & 5.90 & 4.50 & \color{red}{10.15} & 5.75 & 5.70 &\color{red}{14.05}\\ 
  1 \% & 1.05 & 1.65 & 1.05 &\color{red}{5.30} & 1.30 & 1.65 &\color{red}{6.35} \\\hline
  \multicolumn{8}{l}{Sample Size $n=200$} \\
  
10 \% & 10.05 & 10.10 & 9.80 & 10.05 & 11.10 & 11.20 & 12.65 \\ 
  5 \% & 5.05 & 5.10 & 5.00 & 5.50 & 4.75 & 5.50 & 7.30 \\ 
  1 \% & 1.00 & 1.15 & 0.90 & 1.55 & 1.20 & 1.10 & 2.40 \\ 
  \hline
\end{tabular}
}
\caption{{Estimated Type-I error, $\alpha\in\{10\%,5\%,1\%\}$. Sample size $n=30,200$. Censoring percentage $50\%$. Fixed-length-scale 1}}\label{Table:null50}
\end{table}
\end{multicols}

From tables \ref{Table 1:null30} and \ref{Table:null50}, we can observe that all the tests achieve the correct level for larger sample sizes, i.e. $n=200$. For small sample sizes, i.e. $n=30$, the Pearson test and the combined Log-rank test have a wrong level (shown in red in the tables); thus, the measured performance of these tests under the alternative is not meaningful. This may be due to incomplete convergence of the test statistics, since the associated thresholds are based on asymptotic results. In the case of the Pearson test, the test-statistic uses a plug-in estimator $\hat{G}$ of the censoring distribution $G$, making the testing procedure more complex. Similarly, for the combined Log-rank tests  (WLR) the more weight functions we consider, the more complex is our testing procedure and thus, the more data we need.

\subsection{Proportional hazard functions}
In this experiment, we consider testing the family of alternatives $\lambda_{\theta}(t)=\theta$ where $\theta\in\{0.5+0.05k;k=1,\ldots,30\}$ against the null hypothesis $H_0:\lambda(t)=1$. We consider the censoring distribution $G(t)=1-e^{-\gamma t}$, where $\gamma$ is chosen in such a way it generates a censoring percentage of $30\%$ and $50\%$ for each combination of alternative. The level of the test is fixed at $\alpha=0.05$, and we consider a sample size $n\in\{30,50,100,200\}$. Each experiment considers $N=2000$ independent repetitions.

It is a well-known fact that the classical Log-rank test is optimal, in the sense that it is the most powerful against proportional hazards alternatives. The aim of our first experiment is thus to compare our test  performance to a test that we know is optimal. Results are shown in Figure \ref{Fig 1}.

\begin{figure}[H]
\begin{center}
 \begin{tabular}{@{}ccc@{}}
\includegraphics[width=.75\textwidth]{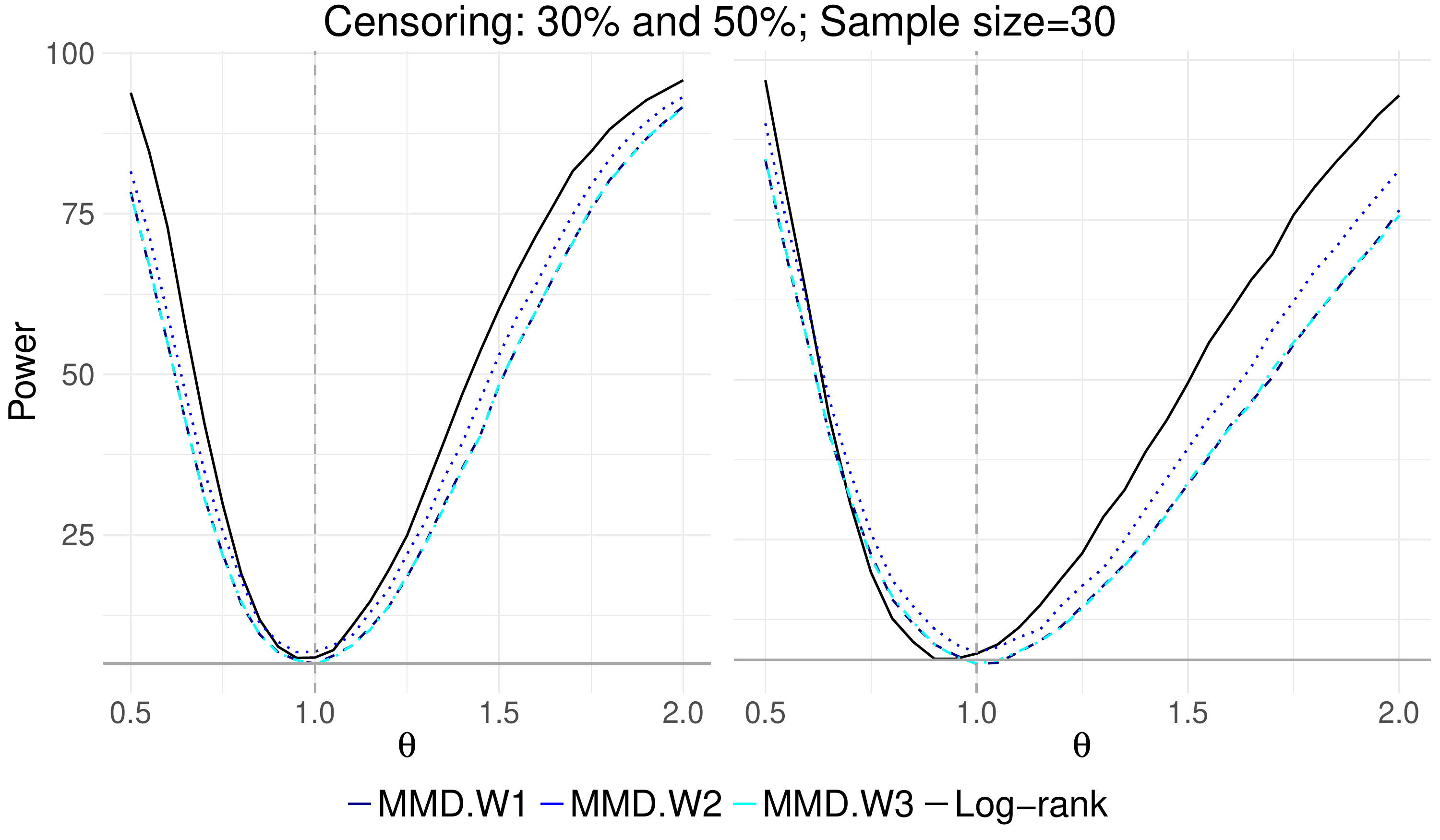}\\
\includegraphics[width=.75\textwidth]{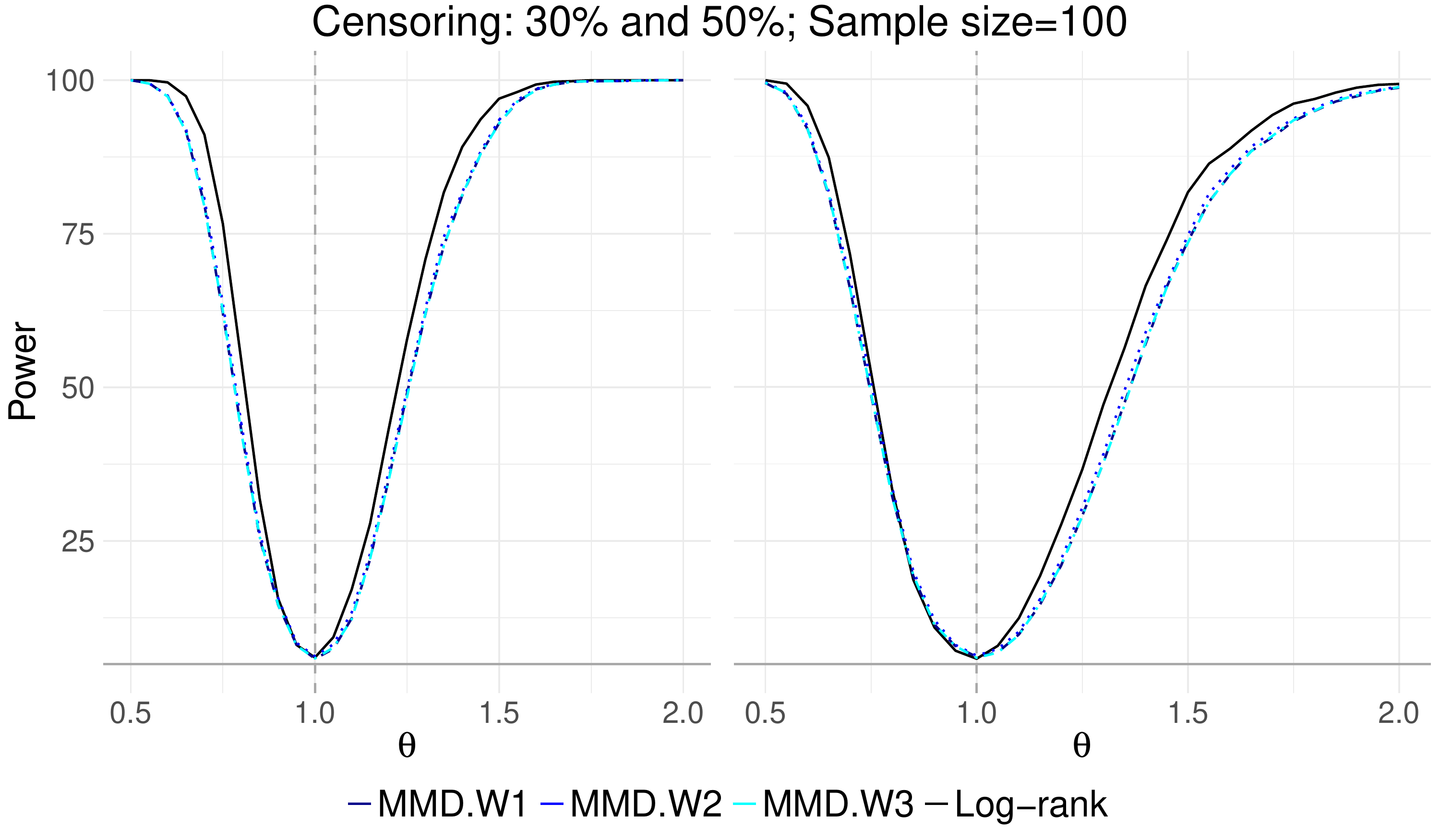}
  \end{tabular}\caption{Estimated power related to the family of parallel hazard alternatives $\lambda_{\theta}(t)=\theta$. The null is recovered, when $\theta=1$.}\label{Fig 1}
\end{center}
\end{figure}

Overall, our test performs strongly (across all wild bootstrap distributions), despite being more general: for the case of 30 data point our test is quite competitive, and when we observe 100 data points, the tests behave almost equally. Recall that $\theta =1$ represents the null hypothesis, so the closer to 1 we are, the harder it is for the test to discriminate from the null.  See Section \ref{sec:PropHazardExperiments} of the supplementary material for a table with the numerical values,  more sample sizes and combination of parameters, and other competitors.

\subsection{Time-dependent hazard functions}
In this section, we consider time-dependent hazard alternatives, which describe a risk that varies over time. Examples include clinical treatment that becomes less effective over time, or seasonal trends in selling patterns. We consider two particular instances of time-dependent hazard functions: i) periodic hazard functions and ii) Weibull hazard functions.

\subsubsection{Periodic hazard functions:}
Periodic hazard functions arise in scenarios exhibiting periodic patterns, including failure of machines in industrial life-testing, consumer behaviour, and labour market participation among other scenarios: see \cite{potter2001covariate} for  discussion.

\begin{figure}[H]
\begin{center}
 \begin{tabular}{@{}ccc@{}}
\includegraphics[width=.63\textwidth]{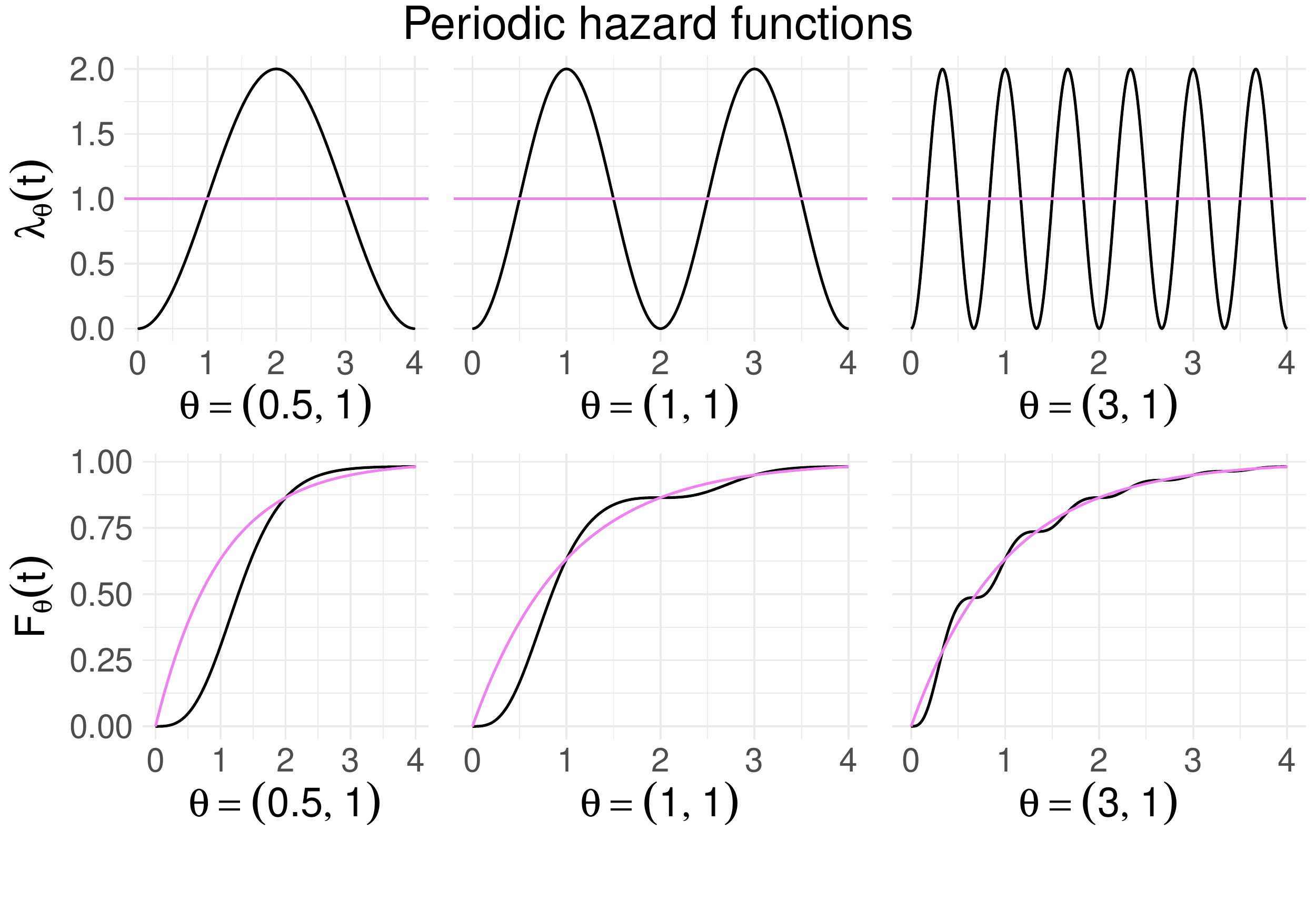}
  \end{tabular}\caption{Top row: hazard function $\lambda_{\theta}(t)=1-\cos(\theta_1\pi t)$ for $\theta_1\in\{0.5,1,3\}$. Bottom row: corresponding cumulative distribution.  In pink: exponential distribution (null).} \label{Fig:2}
\end{center}
\end{figure}
We consider the family of hazards functions given by $\lambda_{\theta}(t)=1-\theta_2\cos(\theta_1\pi t)$, with $\theta=(\theta_1,\theta_2)$ such that $\theta_2<1$ and $\theta_1\in\R$. The null hypothesis is given by $H_0:\lambda(t)=1$ (equivalently, $\theta_2=0$), and we consider the alternatives $\lambda_{\theta}(t)$ with $\theta_2=1$ and $\theta_1\in\{0.5,0.7,0.9,1,1.5,3\}$. 
Instead of fixing a censoring percentage, in this experiment we fix a censoring distribution. In particular we choose the censoring distribution to be $G(t)=1-e^{-\gamma t}$, with $\gamma\in\{1/2,2\}$. Each experiment considers $N=2000$ independent repetitions. 

In Figure \ref{Fig:2}, we show examples of hazards and their corresponding c.d.f. for different parameters of $\theta$. Notice that as the frequency parameter $\theta_1$ increases, the alternative distribution functions $F_{\theta}$ appears closer to the null distribution function $F_0(t)=1-\exp\{-t\}$, shown in pink.

\begin{multicols}{2}
\begin{table}[H]
\centering
\resizebox{1\columnwidth}{!}{
\begin{tabular}{rrrrarra}
  \hline
  \multicolumn{8}{c}{Periodic hazard function}\\
  \hline
  \multicolumn{8}{c}{Sample size n=30; Cens. param. $\gamma=1/2$; Adaptive length-scale}\\
  \hline
 $\boldsymbol{\alpha}$& MW1 & MW2 & MW3 & Pearson & LR1 & LR2 & WLR \\ 
  \hline
  \multicolumn{8}{l}{$\theta=(0.5,1)$}\\
   10 \% & 99.95 & 99.95 & 99.95 & 99.40 & 89.65 & 99.95 & 62.15 \\ 
  5 \% & 99.95 & 99.95 & 99.95 & 96.95 & 74.50 & 99.80 & 44.60 \\ 
  1 \% & 99.75 & 99.80 & 99.75 & 76.80 & 31.25 & 94.20 & 17.30 \\ 
  \hline
     \multicolumn{8}{l}{$\theta=(1,1)$}\\
 10 \% & 98.60 & 99.05 & 98.60 & 87.35 & 16.20 & 62.80 & 31.35 \\ 
  5 \% & 95.50 & 96.25 & 95.30 & 73.60 & 8.05 & 39.55 & 21.55 \\ 
  1 \% & 82.00 & 83.75 & 82.40 & 42.05 & 0.80 & 6.55 & 10.30 \\ 
  \hline
   \multicolumn{8}{l}{$\theta=(3,1)$}\\
   10 \% & 23.40 & 25.80 & 23.40 & 20.00 & 8.95 & 5.35 & 19.35 \\ 
  5 \% & 13.40 & 15.65 & 13.60 & 12.70 & 4.55 & 2.20 & 13.60 \\ 
  1 \% & 4.00 & 5.60 & 4.10 & 4.90 & 0.95 & 0.25 & 6.70 \\ 
   \hline
\end{tabular}}\caption{Power (from 0\% to 100\%) under different alternatives $\theta$ for the Periodic hazards experiment. Sample size 30. Censoring Parameter $1/2$. Adaptive length-scale.}\label{Table: Periodic 1}
\end{table}

\begin{table}[H]
\centering
\resizebox{1\columnwidth}{!}{
\begin{tabular}{rrrrrrrr}
  \hline
  \multicolumn{8}{c}{Periodic hazard functions}\\
  \hline
  \multicolumn{8}{c}{Sample size n=200; Cens. param. $\gamma=1/2$; Adaptive length-scale}\\
  \hline
 & MW1 & MW2 & MW3 & Pearson & LR1 & LR2 & WLR \\ 
  \hline
  \multicolumn{8}{l}{$\theta=(0.5,1)$}\\
10 \% & 100.00 & 100.00 & 100.00 & 100.00 & 100.00 & 100.00 & 100.00 \\ 
  5 \% & 100.00 & 100.00 & 100.00 & 100.00 & 100.00 & 100.00 & 100.00 \\ 
  1 \% & 100.00 & 100.00 & 100.00 & 100.00 & 100.00 & 100.00 & 100.00 \\ 
   \hline
     \multicolumn{8}{l}{$\theta=(1,1)$}\\
10 \% & 100.00 & 100.00 & 100.00 & 100.00 & 83.00 & 100.00 & 100.00 \\ 
  5 \% & 100.00 & 100.00 & 100.00 & 100.00 & 71.00 & 100.00 & 99.95 \\ 
  1 \% & 100.00 & 100.00 & 100.00 & 100.00 & 39.95 & 100.00 & 99.30 \\ 
   \hline
   \multicolumn{8}{l}{$\theta=(3,1)$}\\
   10 \% & 98.20 & 98.35 & 98.30 & 44.70 & 10.70 & 21.85 & 44.35 \\ 
  5 \% & 89.75 & 89.95 & 89.85 & 31.70 & 4.95 & 11.40 & 31.45 \\ 
  1 \% & 48.90 & 48.00 & 47.85 & 14.20 & 1.00 & 2.35 & 12.85 \\ 
   \hline
\end{tabular}}\caption{Power (from 0\% to 100\%) under different alternatives $\theta$ for the Periodic hazards experiment. Sample size 200. Censoring Parameter $1/2$. Adaptive length-scale.}\label{Table: Periodic 2}
\end{table}

\end{multicols}

In Tables \ref{Table: Periodic 1} and \ref{Table: Periodic 2} we show the result of our experiments for sample size $n=30$ and $n=200$, respectively. Our test strongly outperforms the other tests, and is able to discriminate the alternative from the null in the first two cases, while in the third case ($\theta = (3,1)$) it has the better overall result for both sample sizes. These results also confirm that periodic hazards with high frequencies present a more challenging task. We note moreover from Table \ref{Table: Periodic 1} that Pearson and WLR tests (in red) do not have the correct level for 30 samples (see Section \ref{sec:PeriodicHazardExperiments} of the supplementary material), and thus their reported power might be over-optimistic. In Section \ref{sec:PeriodicHazardExperiments} of the supplementary material  we include more sample sizes and more parameters, which achieve similar results.

\subsubsection{Weibull hazard functions}
Weibull models are popular in survival analysis and reliability applications where there is either an increasing or decreasing hazard rate. Their popularity is due to their flexibility, despite only being parametrised by two values. Weibull hazard functions has been used to model failure of composite materials, adhesive wear in metals, and fracture strength of glass among other examples, see \cite{Lai2006} for more details.

The Weibull hazard function is given by $\lambda_{\theta_w}(t)=\theta_{w1}/\theta_{w2}(t/\theta_{w2})^{\theta_{w1}-1}$, where $\theta_w=(\theta_{w1},\theta_{w2})\in\R_+^2$. In particular, $\theta_{w2}$ denotes the scale parameter which has an proportional effect on the hazard function, and $\theta_{w1}$ is the shape parameter. For $\theta_{w1}<1$ we have a decreasing hazard, for  $\theta_{w1}>1$ we have an increasing hazard, and for $\theta_{w_1}=1$  we recover a constant hazard representing the exponential distribution. The null hypothesis occurs when $\theta_w=(1,1)$ (shown in pink in Figures). In Figure~\ref{Fig:3} we show three different Weibull hazard functions (and their corresponding survival functions) with three types of behaviour: decreasing hazard, increasing concave, and increasing convex. 

\begin{figure}[H]
\begin{center}
 \begin{tabular}{@{}ccc@{}}
\includegraphics[width=.63\textwidth]{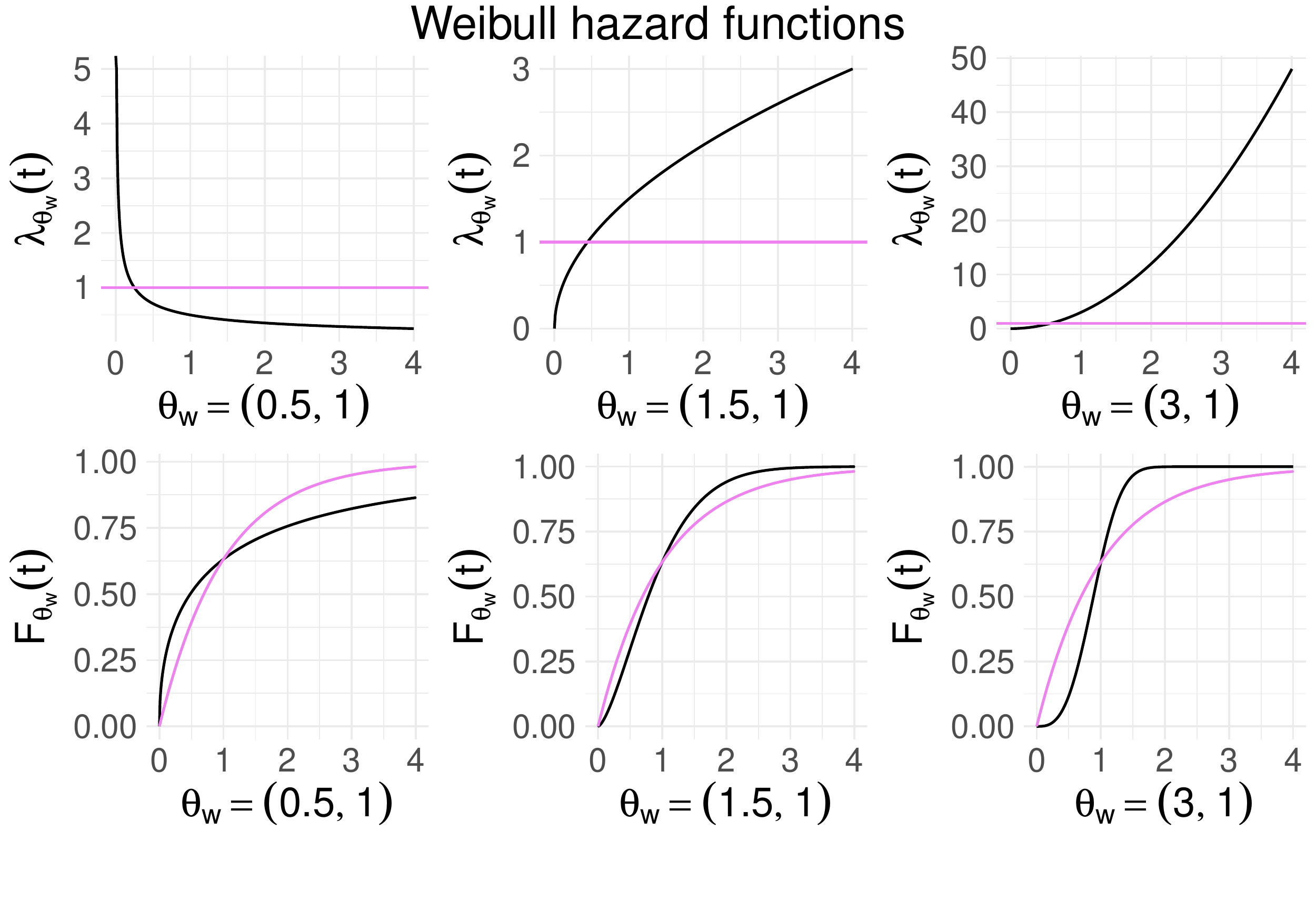}
  \end{tabular}\caption{Top row: hzard function $\lambda_{\theta_w}(t)=\theta_{w1}(t)^{\theta_{w1}-1}$, for $\theta_{w1} \in \{0.5,1.5,3\}$. Bottom row: The corresponding cumulative functions.} \label{Fig:3}
\end{center}
\end{figure}

\begin{multicols}{2}
\begin{table}[H]
\centering
\resizebox{1\columnwidth}{!}{
\begin{tabular}{rrrrarra}
\hline
\multicolumn{8}{c}{Weibull hazard functions}
\\\hline
\multicolumn{8}{c}{Sample size $n= 30 $; Censoring $ 30 \% $; Adaptive length-scale }
\\\hline
& MW1 & MW2 & MW3 & Pearson & LR1 & LR2 &WLR\\ 
\hline
\multicolumn{8}{l}{$\theta_w=(0.5,1)$}\\
10 \%& 67.20 & 70.70 & 67.30 & 70.25 & 28.55 & 71.30 & 70.55 \\ 
5 \%& 52.30 & 56.10 & 52.15 & 60.15 & 20.45 & 65.15 & 63.70 \\ 
1\%& 21.95 & 29.40 & 21.95 & 42.05 & 9.60 & 51.80 & 52.55 \\ 
\hline
\multicolumn{8}{l}{$\theta_w=(1.5,1)$}\\
10 \%& 45.35 & 48.20 & 45.50 & 52.20 & 3.65 & 8.45 & 33.40 \\ 
5 \%& 32.50 & 35.80 & 32.05 & 40.80 & 1.55 & 3.45 & 24.10 \\ 
1\%& 12.70 & 15.05 & 12.35 & 21.10 & 0.05 & 0.25 & 11.55 \\ 
\hline
\multicolumn{8}{l}{$\theta_w=(3,1)$}\\
10 \%& 100.00 & 100.00 & 100.00 & 100.00 & 0.05 & 43.40 & 99.55 \\ 
5 \%& 100.00 & 100.00 & 100.00 & 100.00 & 0.00 & 16.30 & 99.00 \\ 
1\%& 99.80 & 99.95 & 99.85 & 99.90 & 0.00 & 0.45 & 96.20 \\ 
   \hline
\end{tabular}}\caption{Power (from 0\% to 100\%) under different alternatives $\theta_{w1}$ for the Weibull hazards experiment. Sample size 30. Censoring percentage $30\%$. Adaptive length-scale.}\label{Table:Weibull 1}
\end{table}

\begin{table}[H]
\centering
\resizebox{\columnwidth}{!}{
\begin{tabular}{rrrrrrrr}
\hline
\multicolumn{8}{c}{Weibull hazard functions}
\\\hline
\multicolumn{8}{c}{Sample size $n= 200 $; Censoring $ 30 \% $; Adaptive length-scale }
\\\hline
& MW1 & MW2 & MW3 & Pearson & LR1 & LR2 &WLR\\ 
\hline
\multicolumn{8}{l}{$\theta_w=(0.5,1)$}\\
10 \%& 100.00 & 100.00 & 100.00 & 100.00 & 46.05 & 99.95 & 99.95 \\ 
5 \%& 100.00 & 100.00 & 100.00 & 100.00 & 36.80 & 99.90 & 99.95 \\ 
1\%& 100.00 & 100.00 & 100.00 & 100.00 & 20.90 & 99.30 & 99.60 \\ 
\hline
\multicolumn{8}{l}{$\theta_w=(1.5,1)$}\\
10 \%& 99.95 & 100.00 & 99.95 & 100.00 & 6.10 & 72.65 & 98.85 \\ 
5 \%& 99.90 & 99.85 & 99.90 & 100.00 & 2.75 & 56.50 & 96.85 \\ 
1\%& 98.70 & 98.95 & 98.65 & 99.80 & 0.45 & 22.40 & 88.35 \\ 
\hline
\multicolumn{8}{l}{$\theta_w=(3,1)$}\\
10 \%& 100.00 & 100.00 & 100.00 & 100.00 & 2.15 & 100.00 & 100.00 \\ 
5 \%& 100.00 & 100.00 & 100.00 & 100.00 & 0.20 & 100.00 & 100.00 \\ 
1\%& 100.00 & 100.00 & 100.00 & 100.00 & 0.00 & 100.00 & 100.00 \\ 
   \hline
\end{tabular}}\caption{Power (from 0\% to 100\%) under different alternatives $\theta_{w1}$ for the Weibull hazards experiment. Sample size 200. Censoring  percentage $30\%$. Adaptive length-scale.}\label{Table:Weibull 2}
\end{table}
\end{multicols}

In our experiments, we consider alternatives $\lambda_{\theta_w}$ with fixed $\theta_{w2}=1$ and $\theta_{w1}\in\{0.5,1,1.5,2,3\}$.  The censoring distribution is chosen to be $G(t)=1-e^{-\gamma t}$, where $\gamma$ is chosen i such way it produces $30\%$ and $50\%$ of censoring given a fixed alternative distribution. Each experiment considers $N=2000$ independent repetitions. 

In Tables \ref{Table:Weibull 1} and \ref{Table:Weibull 2} we show the result of our experiments for sample sizes $n=30$ and $200$. Our test again yields the best performance.  It also noticeable that Log-rank performs extremely poorly, which is a well-known failure mode for such hazards.

We remark that we should again treat the 
Pearson and WLR results with caution for sample size 30 as they given an  incorrect Type-1 error.
See tables in Section \ref{sec:WeiHazardExperiments} of the supplementary material to see the later and other experiments with more sample sizes and parameters (they achieve similar results).

\section{Discussion}

We have presented a novel testing procedure for goodness-of-fit for right-censored data, based on the MMD distance between a transformation of the observed variables and the uniform distribution.
Being based on kernels, it is not necessary to specify features in advance (as for the weighted log-rank test): rather, we take advantage of the infinite dictionary of features provided implicitly by the kernel.
Our approach has several advantages: First, it is simple to implement, since we only need to be able to evaluate the distribution $F_0$ in the survival times $T_i$ to generate the data $\{F_0(T_i),\Delta_i\}$, and we do not need to know/evaluate $F_0^{-1}$. Second,  the U-statistic kernel $J:\mathbb{R}\times\{0,1\}\to\mathbb{R}$ of equation \eqref{eqn:V-stat kernel} is distribution free, and need be computed/tabulated only once. 
Third, being a U-statistic, the asymptotic analysis is straightforward, as is the bootstrap approach for the test threshold.

We emphasize that extensions to other type of censoring (left and interval) are straightforward, as our test depends on censoring only through the estimate $\tilde F$, in which the mass of a censored interval is distributed uniformly over such an interval. 

Further improvements in the performance of our test might be achieved by a better choice of kernel function for the
problem at hand. In the case of the maximum mean discrepancy on uncensored data, test power is improved by choosing a kernel to optimise the ratio of the statistic to its variance \cite{SutTunStretal17}. Adaptive linear-time test statistics may also be constructed for two-sample \cite{JitSzaChwGre16} and Stein goodness-of-fit \cite{JitXuSzaFuketal17} tests, where the features are again chosen to maximise test power. It would be of interest to extend these ideas to the present setting.

\section{Acknowledgements} Tamara Fern\'andez was supported by Biometrika Trust.

\clearpage
\newpage
\bibliography{ref}
\bibliographystyle{plain}
\newpage
\section{Supplementary material}
\section{Proofs}\label{sec:AppendixProofs}

\subsection{Proof of Proposition \ref{Prop: Unbiased}}
\begin{proof}
Let $T=\min\{X,C\}$ and $\Delta=\ind\{T=X\}$ with $X\sim F$ independent of $C\sim G$. Let $F_0$ be a continuous distribution on $\R$ and define $U=F_0(T)$, then the joint distribution $Q$ of the pair $(U,\Delta)$ is given by
\begin{eqnarray}
\Prob(U\leq u,\Delta=1)&=&\Prob(\min\{F_0(X),F_0(C)\}\leq u, F_0(X)\leq F_0(C))\nonumber\\
&=&\Prob(F_0(X)\leq u, F_0(X)\leq F_0(C))\nonumber\\
&=&\Prob(X\leq F_0^{-1}(u), X\leq C)\nonumber\\
&=&\int_0^{F_0^{-1}(u)}\Prob(x\leq C)dF(x)\nonumber\\
&=&\int_{0}^{F_0^{-1}(u)}[1-G(x)]dF(x)\nonumber\\
&=&\int_{0}^{u}[1-G(F_0^{-1}(x))]dF(F_0^{-1}(x)),
\end{eqnarray}
and
\begin{eqnarray}
\Prob(U\leq u,\Delta=0)&=&\Prob(\min\{F_0(X),F_0(C)\}\leq u, F_0(X)> F_0(C))\nonumber\\
&=&\Prob(F_0(C)\leq u, F_0(X)> F_0(C))\nonumber\\
&=&\Prob(C\leq F_0^{-1}(u), X> C)\nonumber\\
&=&\int_0^{F_0^{-1}(u)}\Prob(X> c)dG(c)\nonumber\\
&=&\int_{0}^{F_0^{-1}(u)}[1-F(c)]dG(c)\nonumber\\
&=&\int_{0}^{u}[1-F(F_0^{-1}(c))]dG(F_0^{-1}(c)),
\end{eqnarray}

Let $u\in[0,1]$ be fixed, we define the random variable
$$Z^u(U,\Delta)=\ind\{U\leq u\}\Delta+\ind\{U\leq u\}(1-\Delta)\frac{u-U}{1-U}.$$
By the strong law of large numbers, it holds
$$\hat{F}(u)=\frac{1}{n}\sum_{i=1}^nZ_i^u\overset{a.s.}{\to}\E(Z^u),$$
where (computed by using the joint distribution $Q$ of the pair $(U,\Delta)$)
\begin{eqnarray}
\E(Z^u)&=&\int_0^u(1-GF_0^{-1}(s))dFF_0^{-1}(s)+\int_0^u\frac{u-s}{1-s}(1-FF_0^{-1}(s))dGF_0^{-1}(s)\nonumber\\
&=&FF_0^{-1}(u)-\left[FF_0^{-1}(u)-\int_0^uFF_0^{-1}(s)dGF_0^{-1}(s)\right]+\int_0^u\frac{u-s}{1-s}(1-FF_0^{-1}(s))dGF_0^{-1}(s)\nonumber\\
&=&FF_0^{-1}(u)+\int_0^u\frac{1}{1-s}[(1-s)(1-FF_0^{-1}(u))-(1-u)(1-FF_0^{-1}(s))]dGF_0^{-1}(s),\nonumber
\end{eqnarray}
(the second equality follows from integration by parts) which is an unbiased estimator of $FF_0^{-1}$ when $F=F_0$, in which case $E(Z^u)=u$. Observe that in the case of extreme censoring, for example $G$ is delta measure on zero, $\E(Z^u)=u$ which reflects our lack of information. 
\end{proof}

\subsection{Proof of Lemma \ref{Lemma: nMMD(x,tilde F)}}
\begin{proof}
By the main theorem of section 5.2.2. of \cite{Serfling80}, it suffices to prove that $\E(J(U,\Delta),(U',\Delta')^2)<\infty$, where $(U,\Delta),(U',\Delta')\overset{i.i.d.}{\sim}Q$ and that the kernel $J$ is degenerated. 
\subparagraph{Degeneracy:} For the degeneracy, note that
\begin{eqnarray}
\E[J((u,\delta),(U,\Delta))]&=&\E\left(\int_{0}^1\int_0^1K(x,y)(dx-dh_{u,\delta}(x))(dy-dh_{U,\Delta}(y))\right)\nonumber\\
&=&\E\left(\int_{0}^1\psi_{u,\delta}(y)(dy-dh_{U,\Delta}(y))\right),\nonumber
\end{eqnarray}
where $\psi_{u,\delta}(y)=\int_0^1K(x,y)(dx-dh_{u,\delta}(x))$. By using equation (3), it holds
\begin{eqnarray}
\E[J((u,\delta),(U,\Delta))]&=&\int_0^1\psi_{u,\delta}(y)dy-\E\left(\int_{U}^1\psi_{u,\delta}(y)\frac{1-\Delta}{1-U}dy+\Delta\psi_{u,\delta}(U)\right)\nonumber\\
&=&\int_0^1\psi_{u,\delta}(y)dy-\int_0^1\int_x^1\psi_{u,\delta}(y)dydG(x)-\int_0^1(1-G(x))\psi_{u,\delta}(x)dx\nonumber\\
&=&\int_0^1\psi_{u,\delta}(y)dy-\int_0^1 G(y)\psi_{u,\delta}(y)dy-\int_0^1(1-G(x))\psi_{u,\delta}(x)dx=0.\nonumber
\end{eqnarray}

\subparagraph{$\boldsymbol{\E(J^2)<\infty:}$}

We continue by checking the finite variance condition, that is $\E(J((U,\Delta),(U',\Delta'))^2)<\infty$. Under assumption 4.1, observe that
\begin{eqnarray}
\E(J((U,\Delta),(U',\Delta'))^2)&=&\E\left(\left[\int_0^1\int_0^1K(x,y)(dx-dh_{U,\Delta}(x))(dy-dh_{U',\Delta'}(y))\right]^2\right)\nonumber\\
&\leq&M^2\left(1+\E\left(\left(\int_0^1dh_{U,\Delta}(x)\right)^4\right)\right)\nonumber\\
&\leq& 2M^2,
\end{eqnarray}
since $h_{U,\Delta}(x)$ is a cumulative distribution function.
\subparagraph{Diagonal:}
We finalize by analysing the asymptotic behaviour of the diagonal term. Under assumption \ref{Assumption: bounded}, it holds
\begin{eqnarray}
\E(J((U,\Delta),(U,\Delta)))&=&\E\left(\int_0^1\int_0^1K(x,y)(dx-dh_{U,\Delta}(x))(dy-dh_{U,\Delta}(y))\right)\nonumber\\
&\leq&M\left(1+\E\left(\left(\int_0^1dh_{U,\Delta}(x)\right)^2\right)\right),\nonumber\\
&\leq&2M
\end{eqnarray}
since $h_{U,\Delta}(x)$ is a distribution function. Then by the strong law of large numbers, the diagonal of the $V$-statistic converges to
\begin{eqnarray}
n\text{Diag}=\frac{1}{n}\sum_{i=1}^nJ((U_i,\Delta_i),(U_i,\Delta_i))\overset{a.s.}{\to}\E(J((U,\Delta),(U,\Delta))).
\end{eqnarray}
\end{proof}

\subsection{Proof of Proposition \ref{Prop: MMD U}}

\begin{proof}
Equation \eqref{eqn: V-statistic} follows easily from equation \eqref{eqn:diff}. The unbiasedness of the U-statistic in equation \eqref{eqn: Ustatsic} follows from the degeneracy property (proved in Lemma \ref{Lemma: nMMD(x,tilde F)}). 
\end{proof}

\subsection{Proof of Proposition \ref{Prop: Consistency}}
In this section, instead of proving Proposition \ref{Prop: Consistency}, we prove a even stronger result. Suppose that each data point $i$ generates a (random) probability measure $\alpha_i$ in $[0,1]$ and suppose all these points are independent and thus the measures they represent are also independent. We define the measure $\alpha$ as the expected measure of $\alpha_i$, i.e. for  $A \subseteq [0,1]$ measurable, we define $\alpha(A) = \E(\alpha_i(A))$. In our setting, under the null it holds $\alpha_i([0,x)) = \ind\{U<x\}\left(\Delta_i+(1-\Delta)\frac{x-U}{1-U}\right)$, and $\alpha([0,x)) = \E(\alpha_i([0,x))) = x$ for $x \in [0,1]$. Our estimator $\tilde F(x)$ corresponds to $n^{-1}\sum_{i=1}\alpha_i([0,x))$. We prove the following theorem.

\begin{theorem}
Let $K:[0,1]^2 \to \R$ be a kernel such that it exist $M \geq 1$ with $|K(x,y)|\leq M$ for all $x,y$. Then
\begin{eqnarray}
\E\left(MMD\left(\frac{1}{n}\sum_i \alpha_i(\cdot), \alpha(\cdot)\right)^2 \right)\leq \frac{4M}{n}
\end{eqnarray}
and, moreover, for all $\varepsilon> 0$ and $n \geq 2$ it holds
\begin{eqnarray}
\Prob\left(\left|MMD\left(\frac{1}{n}\sum_i \alpha_i, \alpha\right)^2-\E\left( MMD\left(\frac{1}{n}\sum_i \alpha_i, \alpha\right)\right)^2\right|> \varepsilon \right) \leq  \exp\left(-\varepsilon^2 n/(32M^2) \right)
\end{eqnarray}
\end{theorem}

\begin{proof}

Denote the signed measure $\beta_i =(\alpha_i-\alpha)$, and for shortness, denote $Z = MMD\left(\frac{1}{n}\sum_i \alpha_i, \alpha\right)$, then
\begin{eqnarray}\label{eqn:MMDZexpand}
Z^2 = \frac{1}{n^2}\int_{[0,1]^2} K(x,y)\sum_{i,j}  \beta_i(dx)\beta_j(dx)= \frac{1}{n^2}\sum_{i,j}\int_{[0,1]^2} K(x,y) \beta_i(dx)\beta_j(dx)
\end{eqnarray}

Using that $\E(Z) \leq \E(Z^2)^{1/2}$, we get
\begin{eqnarray}\label{eqn:EMMDZ}
\E Z \leq \
\frac{n-1}{n}\E\int_{[0,1]^2} K(x,y)\beta_1(dx) \beta_2(dy) + \frac{1}{n}\E\int_{[0,1]^2} K(x,y)\beta_1(dx) \beta_1(dy)
\end{eqnarray}
Note that $|\beta_1(dx)| \leq \alpha_1(dx)+\alpha(dx)$, then

\begin{eqnarray}\label{eqn:EMMDZbound2}
\E\int_{[0,1]^2} K(x,y)\beta_1(dx) \beta_1(dy) \leq M\E\int_{[0,1]^2} (\alpha_1(dx)+\alpha(x))(\alpha_1(dx)+\alpha(x)) = 4M
\end{eqnarray}
Now, we claim that
\begin{eqnarray}\label{eqn:EMMZclaim1}
\E\int_{[0,1]^2} K(x,y)\beta_1(dx) \beta_2(dy) = 0
\end{eqnarray}
as $\beta_1$ and $\beta_2$ are i.i.d measures, and for all measurable sets $A \subseteq [0,1]$ it holds $\E(\beta_1(A)) = 0$.
To check equation~\eqref{eqn:EMMZclaim1} we suppose that $K(x,y) = \sum_{k=1}^N c_k \int_{S_k}(x,y)$ where $S_k$ are rectangles in $[0,1]^2$, i.e. $K$ is a simple function in $[0,1]^2$. For a rectangle $S = [x,x']\times[y,y'] \in [0,1]^2$ we denote $S^1 = [x,x']$ and $S^2 = [y,y']$. Then 
$$\E\int_{[0,1]^2} K(x,y)\beta_1(dx) \beta_2(dy) = \sum_{k=1}^N c_k \E\beta_1(S_k^1)\E\beta_2(S_k^2) = 0$$
then as any arbitrary $K$ can be approximated by simple functions, as our kernel $K$ is bounded by $M$, and $\beta$ is the difference of two probability measures, by the dominated convergence theorem equation~\eqref{eqn:EMMZclaim1} holds for an arbitrary bounded kernel. This proves the first part of the theorem.

For the second part, i.e. concentration. Let $Z'$ be the random variable $Z$ but replacing the data point $j$ by another $j'$ (which is independent of everything). From equation~\eqref{eqn:MMDZexpand} it holds that $Z^2-Z'^2$ equals 

$$  \frac{1}{n^2}\int_{[0,1]^2} K(x,y)[\beta_j(dx)\beta_j(dy)-\beta_{j'}(dx)
\beta_{j'}(dy)]+ \frac{1}{n^2}\sum_{i\neq j} \int_{[0,1]^2} K(x,y) \beta_i(dx)(\beta_j(dy)-\beta_j'(dy))$$

Using the same argument as equation~\eqref{eqn:EMMDZbound2} the the absolute value of the above sum is less or equal than $8M/n$. By the McDiarmid inequality we obtain the result.
\end{proof}

\section{Kernel-based test}\label{Kernel test: results}
In this section, we show results for an extra competitor based on a kernel approach \cite{BagIoaKal13}. This approach considers a kernel density estimate for the survival function, i.e., for  $S=1-F$,  which is obtained by using a slightly modified Kaplan-Meier procedure. Then, the test-statistic is defined as the squared difference between this
density estimate and the model density. The implementation of this test was directly derived from the code made available by the authors.

Since this procedure relies on density estimation as an intermediate step, it has been found to be relatively data-inefficient,
compared with more direct tests, see e.g. the recent discussion in \cite{ditzhaus2018more}. For instance, for the periodic hazard experiment with censoring distribution $G(t)=1-e^{-1/2t}$, we obtain a fairly correct estimation of the Type-I error, see Table \ref{Table KernelType1}, (the null is recovered by considering $\theta_2=0$ in the model for the hazard function given by $\lambda_{\theta}(t)=1-\theta_2\cos(\theta_1\pi t)$).
\begin{multicols}{2}
\begin{table}[H]
\centering
\begin{tabular}{lrrr}
\hline
  \multicolumn{4}{c}{Type-I error}\\
  \hline
  $\boldsymbol{\alpha}$ &n=30 & n=50 & n=100 \\ 
  \hline
   10 \% &11.25&8.65&7\\
    5 \% &5.35&4.65&3.6\\
    1 \% &1.3&1.1&0.9\\
  \hline
\end{tabular}\caption{Estimated Type I error for the Periodic hazard experiments, censoring distribution $G(t)=1-e^{-1/2t}$.}\label{Table KernelType1}
\end{table}

\begin{table}[H]
\centering
\begin{tabular}{lrrr}
\hline
  &\multicolumn{3}{c}{Power}\\
  \hline
  $\boldsymbol{\alpha=5\%}$ &n=30 & n=50 & n=100 \\ 
  \hline
    $\theta=(0.5,1)$ &22.4&43.4&73.45\\
    $\theta=(0.9,1)$ &3.65&4.95&14.1\\
    $\theta=(1.5,1)$ &3.35&5.4&10.4\\
  \hline
\end{tabular}\caption{Estimated Power for the Periodic hazard experiments, censoring distribution $G(t)=1-e^{-1/2t}$.}\label{Table KernelPower}
\end{table}

\end{multicols}
Nevertheless for alternatives that are distinguishable by other tests, as for example for the periodic hazard setting with censoring distribution $G(t)=1-e^{-1/2t}$ and true parameters $\theta=(\theta_1,1)$ and $\theta_1\in\{0.5,0.9,1.5\}$, we obtain discouraging results for small sample sizes, see Table \ref{Table KernelPower}. Indeed, even though we observe an increment of the power with the sample size, the overall power attained by this particular test is clearly inferior when compared to all the other competitors. This behaviour is also observed for the parallel and Weibull experiments. Therefore, we omit this test from our comparisons. Tables were generated by using 2000 independent experiments.

\section{Proportional hazards experiment}\label{sec:PropHazardExperiments}
For this experiment, we consider the family of alternatives $\lambda_{\theta}(t)=\theta$ where $\theta\in\{0.5+0.05k;k=1,\ldots,30\}$. The censoring distribution is chosen to be $G(t)=1-e^{-\gamma t}$, where $\gamma$ is chosen i such way it produces $30\%$ and $50\%$ of censoring given a fixed alternative distribution. 
\subsection{Type-I error}
Estimated Type-I error. In {\color{red}{red}} we observe tests that have an clear incorrect level. In{ \color{orange}{orange}}, we observe tests that have a questionable incorrect level. Tables are based on $2000$ independent experiments. 

\begin{multicols}{2}
\textbf{Fixed length-scale 1; Censoring $30\%$}
\begin{table}[H]
\centering
\resizebox{1\columnwidth}{!}{
\begin{tabular}{rrrrrrrr}
\hline
  \multicolumn{8}{c}{Type-I error}\\
  \hline
 & MW1 & MW2 & MW3 & Pearson & LR1 & LR2 &WLR\\ 
  \hline
  \multicolumn{8}{l}{Sample size n=30} \\
10 \% & 10.10 & 11.70 & 10.20 & \color{red}{14.50} & 11.15 & 11.10 & \color{red}{19.65}\\ 
  5 \% & 5.10 & 6.55 & 4.85 & \color{red}{8.70} & 5.90 & 6.60 & \color{red}{13.65} \\ 
  1 \% &1.15 & 2.05 & 1.10 & \color{red}{3.05} & 1.60 & 1.30 & \color{red}{6.30}\\ 
  \hline
    \multicolumn{8}{l}{Sample size n=50} \\
10 \% & 10.65 & 11.70 & 10.75 & \color{orange}{13.35} & 11.30 & 10.45 & \color{red}{17.80} \\ 
  5 \% & 5.50 & 6.20 & 5.75 & \color{orange}{8.15} & 6.10 & 5.80 & \color{red}{11.85} \\ 
  1 \%& 1.20 & 1.60 & 1.10 & \color{orange}{2.65} & 1.45 & 1.35 & \color{red}{4.60} \\ 
  \hline
  \multicolumn{8}{l}{Sample size n=100} \\
10 \% & 11.10 & 11.60 & 10.90 & 11.00 & 10.95 & 11.10 & \color{orange}{13.00} \\ 
  5 \% & 5.90 & 6.10 & 5.75 & 6.40 & 5.85 & 6.00 & \color{orange}{8.40} \\ 
  1 \%& 1.25 & 1.40 & 1.30 & 2.20 & 1.35 & 1.35 & \color{orange}{2.85} \\ 
  \hline
\end{tabular}
}
\end{table}

\textbf{Adaptive length-scale; Censoring $30\%$}
\begin{table}[H]
\centering
\resizebox{\columnwidth}{!}{
\begin{tabular}{rrrrrrrr}
\hline
  \multicolumn{8}{c}{Type-I error}\\
  \hline
 & MW1 & MW2 & MW3 & Pearson & LR1 & LR2 &WLR\\ 
  \hline
  \multicolumn{8}{l}{Sample size n=30} \\
10 \% & 9.35 & 10.70 & 9.70 & \color{red}{14.50} & 11.15 & 11.10 & \color{red}{19.65}  \\ 
  5 \% & 4.40 & 5.45 & 4.50 & \color{red}{8.70} & 5.90 & 6.60 & \color{red}{13.65} \\ 
  1 \% & 0.95 & 1.45 & 0.95 & \color{red}{3.05} & 1.60 & 1.30 & \color{red}{6.30}\\ 
  \hline
  \multicolumn{8}{l}{Sample size n=50} \\
10 \% & 10.50 & 11.30 & 10.55 & \color{orange}{13.35} & 11.30 & 10.45 & \color{red}{17.80}\\ 
  5 \% & 5.35 & 5.85 & 5.15 & \color{orange}{8.15} & 6.10 & 5.80 & \color{red}{11.85}\\ 
  1 \%& 0.70 & 1.20 & 0.85 & \color{orange}{2.65} & 1.45 & 1.35 & \color{red}{4.60} \\ 
  \hline
  \multicolumn{8}{l}{Sample size n=100} \\
10 \%& 10.65 & 10.90 & 10.70 & 11.00 & 10.95 & 11.10 & \color{orange}{13.00}  \\ 
  5 \%& 5.50 & 5.45 & 5.50 & 6.40 & 5.85 & 6.00 & \color{orange}{8.40} \\ 
  1 \%& 1.20 & 1.40 & 1.15 & 2.20 & 1.35 & 1.35 & \color{orange}{2.85} \\ 
  \hline
\end{tabular}
}
\end{table}
\end{multicols}

\begin{multicols}{2}
\textbf{Fixed length-scale 1; Censoring $50\%$}
\begin{table}[H]
\centering
\resizebox{1\columnwidth}{!}{
\begin{tabular}{rrrrrrrr}
\hline
  \multicolumn{8}{c}{Type-I error}\\
  \hline
 & MW1 & MW2 & MW3 & Pearson & LR1 & LR2 &WLR\\ 
  \hline
  \multicolumn{8}{l}{Sample size n=30} \\
10 \% & 10.50 & 11.80 & 10.50 & \color{red}{15.30} & 10.50 & 10.25 & \color{red}{19.25}\\ 
  5 \% & 4.45 & 5.90 & 4.50 &\color{red}{ 10.15} & 5.75 & 5.70 & \color{red}{14.05} \\ 
  1 \% & 1.05 & 1.65 & 1.05 & \color{red}{5.30} & 1.30 & 1.65 & \color{red}{6.35}\\ 
  \hline
    \multicolumn{8}{l}{Sample size n=50} \\
    10 \% & 10.70 & 11.45 & 10.75 &  \color{red}{14.75} & 10.65 & 10.50 &  \color{red}{16.25}   \\ 
  5 \% & 5.55 & 6.50 & 5.70 &  \color{red}{9.30} & 5.90 & 5.10 &  \color{red}{10.90} \\ 
  1 \%& 1.35 & 1.45 & 1.20 &  \color{red}{5.10} & 1.60 & 1.35 &  \color{red}{4.40} \\ 
  \hline
  \multicolumn{8}{l}{Sample size n=100} \\
10 \% & 12.20 & 12.70 & 12.20 &  \color{orange}{13.25} & 11.45 & 11.65 &  \color{red}{14.10}  \\ 
  5 \% & 6.15 & 6.35 & 6.00 &  \color{red}{8.70} & 5.85 & 5.80 &  \color{orange}{8.25} \\ 
  1 \%& 1.45 & 1.65 & 1.45 & \color{orange}{2.20} & 1.35 & 1.30 &  \color{red}{3.15} \\ 
  \hline
\end{tabular}
}
\end{table}

\textbf{Adaptive length-scale; Censoring $50\%$}
\begin{table}[H]
\centering
\resizebox{1\columnwidth}{!}{
\begin{tabular}{rrrrrrrr}
\hline
  \multicolumn{8}{c}{Type-I error}\\
  \hline
 & MW1 & MW2 & MW3 & Pearson & LR1 & LR2 &WLR\\ 
  \hline
  \multicolumn{8}{l}{Sample size n=30} \\
10 \% & 8.65 & 10.05 & 8.40 &  \color{red}{15.30} & 10.50 & 10.25 & \color{red}{19.25}  \\ 
  5 \% & 4.55 & 5.25 & 4.50 & \color{red}{10.15} & 5.75 & 5.70 & \color{red}{14.05} \\ 
  1 \% & 1.05 & 1.30 & 0.95 & \color{red}{5.30} & 1.30 & 1.65 & \color{red}{6.35} \\ 
  \hline
  \multicolumn{8}{l}{Sample size n=50} \\
10 \% & 9.45 & 10.25 & 9.40 & \color{red}{14.75} & 10.65 & 10.50 & \color{red}{16.25} \\ 
  5 \% & 4.85 & 5.50 & 4.80 & \color{red}{9.30} & 5.90 & 5.10 & \color{red}{10.90}\\ 
  1 \%& 1.20 & 1.40 & 1.20 & \color{red}{5.10} & 1.60 & 1.35 & \color{red}{4.40}\\ 
  \hline
  \multicolumn{8}{l}{Sample size n=100} \\
10 \%& 11.00 & 11.25 & 11.05 & \color{orange}{13.25} & 11.45 & 11.65 & \color{red}{14.10}   \\ 
  5 \%& 5.40 & 5.60 & 5.25 & \color{red}{8.70} & 5.85 & 5.80 & \color{orange}{8.25} \\ 
  1 \%& 1.35 & 1.50 & 1.40 & \color{orange}{2.20} & 1.35 & 1.30 & \color{red}{3.15}  \\ 
  \hline
\end{tabular}
}
\end{table}
\end{multicols}

\subsection{Power tables}
In {\color{red}{red}} we observe tests that have an clear incorrect level and thus estimated power must be looked with caution. In{ \color{orange}{orange}} are tests that have a questionable incorrect level. Tables are based on $2000$ independent experiments. 

\begin{landscape}
\subsubsection{Parallel hazard functions - Sample size $\boldsymbol{n=30}$ - Significance $\boldsymbol{\alpha=0.05}$ - Fixed length-scale 1 }
\begin{table}[H]
\centering
\begin{tabular}{l|rr|br|rr|aa|rr|br|aa}
  \hline
  rate& \multicolumn{2}{c|}{MW1} & \multicolumn{2}{c|}{MW2} & \multicolumn{2}{c|}{MW3} & \multicolumn{2}{c|}{Pearson} & \multicolumn{2}{c|}{LR1} & \multicolumn{2}{c|}{LR2}& \multicolumn{2}{c}{WLR}\\
cens & 30\% & 50\% & 30\% & 50\% & 30\% & 50\% & 30\% & 50\% & 30\% & 50\% & 30\% & 50\% & 30\% & 50\% \\ 
  \hline
0.5 & 74.30 & 67.45 & 78.15 & 72.05 & 74.40 & 67.60 & 64.90 & 45.20 & 88.20 & 77.45 & 78.40 & 65.40 & 62.70 & 42.80 \\ 
  0.55 & 62.55 & 55.40 & 66.55 & 59.75 & 62.60 & 55.75 & 47.70 & 30.80 & 77.10 & 63.25 & 65.40 & 51.95 & 46.15 & 29.35 \\ 
  0.6 & 50.30 & 44.70 & 55.05 & 49.15 & 50.85 & 44.70 & 30.80 & 19.30 & 62.95 & 49.85 & 51.65 & 37.70 & 31.30 & 20.30 \\ 
  0.65 & 39.00 & 33.35 & 43.50 & 37.80 & 39.15 & 33.45 & 19.40 & 12.10 & 48.35 & 35.65 & 38.25 & 26.30 & 22.30 & 13.80 \\ 
  0.7 & 28.20 & 25.20 & 32.55 & 28.55 & 27.90 & 25.35 & 11.85 & 8.90 & 34.40 & 24.65 & 25.20 & 18.25 & 15.85 & 10.50 \\ 
  0.75 & 19.35 & 18.15 & 22.55 & 20.80 & 19.45 & 17.65 & 8.30 & 6.55 & 22.45 & 15.90 & 15.85 & 12.35 & 12.60 & 9.45 \\
\hline
\end{tabular}
\end{table}
\end{landscape}

\begin{landscape}
\subsubsection{Parallel hazard functions - Sample size $\boldsymbol{n=30}$ - Significance $\boldsymbol{\alpha=0.05}$ - Adaptive length-scale  }
\begin{table}[H]
\centering
\begin{tabular}{l|rr|rr|rr|aa|rr|br|aa}
  \hline
  rate& \multicolumn{2}{c|}{MW1} & \multicolumn{2}{c|}{MW2} & \multicolumn{2}{c|}{MW3} & \multicolumn{2}{c|}{Pearson} & \multicolumn{2}{c|}{LR1} & \multicolumn{2}{c|}{LR2}& \multicolumn{2}{c}{WLR}\\
cens & 30\% & 50\% & 30\% & 50\% & 30\% & 50\% & 30\% & 50\% & 30\% & 50\% & 30\% & 50\% & 30\% & 50\% \\ 
  \hline
0.5 & 70.65 & 63.75 & 73.85 & 66.95 & 70.60 & 63.45 & 64.90 & 45.20 & 88.20 & 77.45 & 78.40 & 65.40 & 62.70 & 42.80 \\ 
  0.55 & 57.80 & 51.80 & 61.80 & 55.60 & 57.20 & 51.90 & 47.70 & 30.80 & 77.10 & 63.25 & 65.40 & 51.95 & 46.15 & 29.35 \\ 
  0.6 & 45.05 & 39.40 & 49.15 & 43.45 & 45.10 & 39.15 & 30.80 & 19.30 & 62.95 & 49.85 & 51.65 & 37.70 & 31.30 & 20.30 \\ 
  0.65 & 33.45 & 29.40 & 37.35 & 32.45 & 33.20 & 29.40 & 19.40 & 12.10 & 48.35 & 35.65 & 38.25 & 26.30 & 22.30 & 13.80 \\ 
  0.7 & 23.15 & 21.30 & 25.95 & 24.30 & 23.25 & 21.25 & 11.85 & 8.90 & 34.40 & 24.65 & 25.20 & 18.25 & 15.85 & 10.50 \\ 
  0.75 & 16.05 & 15.90 & 18.65 & 17.60 & 15.90 & 15.25 & 8.30 & 6.55 & 22.45 & 15.90 & 15.85 & 12.35 & 12.60 & 9.45 \\
\hline
\end{tabular}
\end{table}
\end{landscape}

\begin{landscape}
\subsubsection{Parallel hazard functions - Sample size $\boldsymbol{n=50}$ - Significance $\boldsymbol{\alpha=0.05}$ - Fixed length-scale 1 }
\begin{table}[H]
\centering
\begin{tabular}{l|rr|rb|rr|ba|rr|rr|aa}
  \hline
  rate& \multicolumn{2}{c|}{MW1} & \multicolumn{2}{c|}{MW2} & \multicolumn{2}{c|}{MW3} & \multicolumn{2}{c|}{Pearson} & \multicolumn{2}{c|}{LR1} & \multicolumn{2}{c|}{LR2}& \multicolumn{2}{c}{WLR}\\
cens & 30\% & 50\% & 30\% & 50\% & 30\% & 50\% & 30\% & 50\% & 30\% & 50\% & 30\% & 50\% & 30\% & 50\% \\ 
  \hline
  0.5 & 93.65 & 89.70 & 94.15 & 90.80 & 93.80 & 89.70 & 89.70 & 75.80 & 97.70 & 94.45 & 94.35 & 88.20 & 88.90 & 74.95 \\ 
  0.55 & 85.55 & 79.30 & 87.55 & 81.25 & 86.15 & 79.25 & 77.05 & 58.85 & 94.15 & 85.95 & 87.35 & 76.15 & 76.00 & 57.00 \\ 
  0.6 & 73.10 & 65.45 & 76.15 & 67.85 & 73.85 & 65.40 & 60.25 & 41.55 & 84.30 & 71.95 & 74.60 & 59.70 & 59.50 & 39.85 \\ 
  0.65 & 59.40 & 51.45 & 62.30 & 54.00 & 59.90 & 51.30 & 42.75 & 26.70 & 70.05 & 57.70 & 59.60 & 44.45 & 40.75 & 26.40 \\ 
  0.7 & 44.75 & 38.80 & 47.70 & 41.05 & 44.05 & 38.65 & 27.35 & 16.65 & 54.50 & 40.90 & 42.25 & 31.90 & 27.20 & 17.20 \\ 
  0.75 & 32.40 & 28.25 & 34.70 & 30.40 & 32.30 & 28.05 & 16.55 & 10.75 & 37.30 & 27.30 & 28.80 & 21.05 & 17.85 & 12.40 \\
\hline
\end{tabular}
\end{table}
\end{landscape}

\begin{landscape}
\subsubsection{Parallel hazard functions - Sample size $\boldsymbol{n=50}$ - Significance $\boldsymbol{\alpha=0.05}$ - Adaptive length-scale  }
\begin{table}[H]
\centering
\begin{tabular}{l|rr|rr|rr|ba|rr|rr|aa}
  \hline
  rate& \multicolumn{2}{c|}{MW1} & \multicolumn{2}{c|}{MW2} & \multicolumn{2}{c|}{MW3} & \multicolumn{2}{c|}{Pearson} & \multicolumn{2}{c|}{LR1} & \multicolumn{2}{c|}{LR2}& \multicolumn{2}{c}{WLR}\\
cens & 30\% & 50\% & 30\% & 50\% & 30\% & 50\% & 30\% & 50\% & 30\% & 50\% & 30\% & 50\% & 30\% & 50\% \\ 
  \hline
0.5 & 91.45 & 86.75 & 92.50 & 87.80 & 91.55 & 86.55 & 89.70 & 75.80 & 97.70 & 94.45 & 94.35 & 88.20 & 88.90 & 74.95 \\ 
  0.55 & 82.50 & 75.60 & 84.30 & 77.85 & 82.90 & 75.60 & 77.05 & 58.85 & 94.15 & 85.95 & 87.35 & 76.15 & 76.00 & 57.00 \\ 
  0.6 & 69.60 & 61.30 & 71.25 & 63.80 & 69.60 & 61.05 & 60.25 & 41.55 & 84.30 & 71.95 & 74.60 & 59.70 & 59.50 & 39.85 \\ 
  0.65 & 54.35 & 47.15 & 56.65 & 48.85 & 53.95 & 47.25 & 42.75 & 26.70 & 70.05 & 57.70 & 59.60 & 44.45 & 40.75 & 26.40 \\ 
  0.7 & 39.95 & 34.60 & 42.10 & 36.80 & 40.30 & 34.35 & 27.35 & 16.65 & 54.50 & 40.90 & 42.25 & 31.90 & 27.20 & 17.20 \\ 
  0.75 & 28.30 & 23.80 & 30.30 & 25.80 & 28.45 & 24.15 & 16.55 & 10.75 & 37.30 & 27.30 & 28.80 & 21.05 & 17.85 & 12.40 \\
\hline
\end{tabular}
\end{table}
\end{landscape}

\begin{landscape}
\subsubsection{Parallel hazard functions - Sample size $\boldsymbol{n=100}$ - Significance $\boldsymbol{\alpha=0.05}$ - Fixed length-scale 1 }
\begin{table}[H]
\centering
\begin{tabular}{l|rr|rr|rr|ra|rr|rr|aa}
  \hline
  rate& \multicolumn{2}{c|}{MW1} & \multicolumn{2}{c|}{MW2} & \multicolumn{2}{c|}{MW3} & \multicolumn{2}{c|}{Pearson} & \multicolumn{2}{c|}{LR1} & \multicolumn{2}{c|}{LR2}& \multicolumn{2}{c}{WLR}\\
cens & 30\% & 50\% & 30\% & 50\% & 30\% & 50\% & 30\% & 50\% & 30\% & 50\% & 30\% & 50\% & 30\% & 50\% \\ 
  \hline
  0.5 & 99.90 & 99.40 & 99.90 & 99.55 & 99.80 & 99.45 & 99.85 & 98.45 & 100.00 & 99.85 & 99.90 & 99.35 & 99.80 & 98.60 \\ 
  0.55 & 99.00 & 97.65 & 98.95 & 97.85 & 98.85 & 97.75 & 98.35 & 93.70 & 99.80 & 99.30 & 99.15 & 96.85 & 98.65 & 94.40 \\ 
  0.6 & 95.70 & 91.90 & 96.00 & 92.40 & 95.65 & 91.90 & 93.85 & 82.35 & 98.70 & 95.70 & 96.15 & 89.15 & 92.80 & 82.00 \\ 
  0.65 & 87.35 & 81.15 & 88.35 & 81.75 & 87.50 & 81.35 & 81.50 & 64.65 & 94.35 & 87.30 & 87.05 & 76.30 & 79.95 & 62.30 \\ 
  0.7 & 74.60 & 66.10 & 75.55 & 67.80 & 74.35 & 66.10 & 60.65 & 43.10 & 83.90 & 71.65 & 73.70 & 60.40 & 59.10 & 42.70 \\ 
  0.75 & 57.20 & 48.75 & 58.45 & 50.75 & 56.95 & 48.95 & 39.10 & 25.85 & 66.60 & 52.40 & 54.35 & 40.90 & 37.10 & 25.20 \\
\hline
\end{tabular}
\end{table}
\end{landscape}

\begin{landscape}
\subsubsection{Parallel hazard functions - Sample size $\boldsymbol{n=100}$ - Significance $\boldsymbol{\alpha=0.05}$ - Adaptive length-scale  }
\begin{table}[H]
\centering
\begin{tabular}{l|rr|rr|rr|ra|rr|rr|aa}
  \hline
  rate& \multicolumn{2}{c|}{MW1} & \multicolumn{2}{c|}{MW2} & \multicolumn{2}{c|}{MW3} & \multicolumn{2}{c|}{Pearson} & \multicolumn{2}{c|}{LR1} & \multicolumn{2}{c|}{LR2}& \multicolumn{2}{c}{WLR}\\
cens & 30\% & 50\% & 30\% & 50\% & 30\% & 50\% & 30\% & 50\% & 30\% & 50\% & 30\% & 50\% & 30\% & 50\% \\ 
  \hline
0.5 & 99.60 & 99.40 & 99.60 & 99.50 & 99.60 & 99.40 & 99.85 & 98.45 & 100.00 & 99.85 & 99.90 & 99.35 & 99.80 & 98.60 \\ 
  0.55 & 98.60 & 96.45 & 98.90 & 96.70 & 98.65 & 96.25 & 98.35 & 93.70 & 99.80 & 99.30 & 99.15 & 96.85 & 98.65 & 94.40 \\ 
  0.6 & 94.65 & 89.35 & 94.90 & 89.80 & 94.70 & 89.50 & 93.85 & 82.35 & 98.70 & 95.70 & 96.15 & 89.15 & 92.80 & 82.00 \\ 
  0.65 & 85.30 & 77.30 & 85.80 & 77.95 & 85.05 & 77.00 & 81.50 & 64.65 & 94.35 & 87.30 & 87.05 & 76.30 & 79.95 & 62.30 \\ 
  0.7 & 69.75 & 61.00 & 70.70 & 61.90 & 69.75 & 61.25 & 60.65 & 43.10 & 83.90 & 71.65 & 73.70 & 60.40 & 59.10 & 42.70 \\ 
  0.75 & 51.15 & 43.05 & 52.20 & 44.50 & 51.05 & 43.00 & 39.10 & 25.85 & 66.60 & 52.40 & 54.35 & 40.90 & 37.10 & 25.20 \\
\hline
\end{tabular}
\end{table}
\end{landscape}

\begin{landscape}
\subsubsection{Parallel hazard functions - Sample size $\boldsymbol{n=200}$ - Significance $\boldsymbol{\alpha=0.05}$ - Fixed length-scale 1 }
\begin{table}[H]
\centering
\begin{tabular}{l|rr|rr|rr|rr|rr|rr|bb}
  \hline
  rate& \multicolumn{2}{c|}{MW1} & \multicolumn{2}{c|}{MW2} & \multicolumn{2}{c|}{MW3} & \multicolumn{2}{c|}{Pearson} & \multicolumn{2}{c|}{LR1} & \multicolumn{2}{c|}{LR2}& \multicolumn{2}{c}{WLR}\\
cens & 30\% & 50\% & 30\% & 50\% & 30\% & 50\% & 30\% & 50\% & 30\% & 50\% & 30\% & 50\% & 30\% & 50\% \\ 
  \hline
0.5 & 100.00 & 100.00 & 100.00 & 100.00 & 100.00 & 100.00 & 100.00 & 100.00 & 100.00 & 100.00 & 100.00 & 100.00 & 100.00 & 100.00 \\ 
  0.55 & 100.00 & 99.95 & 100.00 & 99.95 & 100.00 & 99.95 & 100.00 & 99.90 & 100.00 & 99.95 & 100.00 & 99.95 & 100.00 & 99.90 \\ 
  0.6 & 99.90 & 99.65 & 99.90 & 99.65 & 99.90 & 99.60 & 99.90 & 99.45 & 99.90 & 99.85 & 99.95 & 99.55 & 99.90 & 99.35 \\ 
  0.65 & 99.30 & 98.10 & 99.35 & 98.30 & 99.35 & 98.25 & 99.05 & 95.20 & 99.85 & 99.40 & 99.30 & 96.80 & 99.00 & 94.35 \\ 
  0.7 & 96.25 & 91.80 & 96.10 & 91.95 & 96.20 & 91.30 & 93.35 & 80.80 & 98.85 & 95.40 & 95.80 & 88.10 & 93.55 & 80.90 \\ 
  0.75 & 85.90 & 77.30 & 86.50 & 78.10 & 86.30 & 77.40 & 74.85 & 54.85 & 93.00 & 82.90 & 84.45 & 70.75 & 73.55 & 55.90 \\
\hline
\end{tabular}
\end{table}
\end{landscape}

\begin{landscape}
\subsubsection{Parallel hazard functions - Sample size $\boldsymbol{n=200}$ - Significance $\boldsymbol{\alpha=0.05}$ - Adaptive length-scale  }
\begin{table}[H]
\centering
\begin{tabular}{l|rr|rr|rr|rr|rr|rr|bb}
  \hline
  rate& \multicolumn{2}{c|}{MW1} & \multicolumn{2}{c|}{MW2} & \multicolumn{2}{c|}{MW3} & \multicolumn{2}{c|}{Pearson} & \multicolumn{2}{c|}{LR1} & \multicolumn{2}{c|}{LR2}& \multicolumn{2}{c}{WLR}\\
cens & 30\% & 50\% & 30\% & 50\% & 30\% & 50\% & 30\% & 50\% & 30\% & 50\% & 30\% & 50\% & 30\% & 50\% \\ 
  \hline
  0.5 & 100.00 & 100.00 & 100.00 & 100.00 & 100.00 & 100.00 & 100.00 & 100.00 & 100.00 & 100.00 & 100.00 & 100.00 & 100.00 & 100.00 \\ 
  0.55 & 100.00 & 99.95 & 100.00 & 99.95 & 99.95 & 99.95 & 100.00 & 99.90 & 100.00 & 99.95 & 100.00 & 99.95 & 100.00 & 99.90 \\ 
  0.6 & 99.90 & 99.65 & 99.90 & 99.65 & 99.90 & 99.55 & 99.90 & 99.45 & 99.90 & 99.85 & 99.95 & 99.55 & 99.90 & 99.35 \\ 
  0.65 & 99.05 & 97.05 & 99.15 & 97.00 & 99.10 & 97.05 & 99.05 & 95.20 & 99.85 & 99.40 & 99.30 & 96.80 & 99.00 & 94.35 \\ 
  0.7 & 94.60 & 89.05 & 94.65 & 89.40 & 94.70 & 89.05 & 93.35 & 80.80 & 98.85 & 95.40 & 95.80 & 88.10 & 93.55 & 80.90 \\ 
  0.75 & 81.60 & 71.85 & 81.65 & 72.70 & 81.70 & 71.75 & 74.85 & 54.85 & 93.00 & 82.90 & 84.45 & 70.75 & 73.55 & 55.90 \\
\hline
\end{tabular}
\end{table}
\end{landscape}

\section{Periodic hazard functions}\label{sec:PeriodicHazardExperiments}
Let $\theta=(\theta_1,\theta_2)$ such that $\theta_2<1$ and $\theta_1\in\R$. We consider the family of hazard functions $\lambda_{\theta}(t)=1-\theta_2\cos(\theta_1\pi t)$ for fixed $\theta_2=1$ and $\theta_1\in\{0.5,0.7,0.9,1,1.5,3\}$. We include plots of the hazard function, distribution function and distribution of the transformed data $U=F_0(X)$ for each combination of alternative and null hypothesis. The null hypothesis, shown in the plots in color pink, is denoted by $F_0(t)=1-e^{-t}$ and $\lambda_0(t)=1$. 

Instead of fixing a censoring percentage, in this experiment we fix a censoring distribution. In particular we choose the censoring distribution to be $G(t)=1-e^{-\gamma t}$, with $\gamma\in\{1/2,2\}$. Notice the censoring percentage varies for different choices of alternative distributions, determined by the parameter $\theta$, and censoring parameter $\gamma$. \\

\begin{tabular}{c}
   \includegraphics[width=0.9\textwidth]{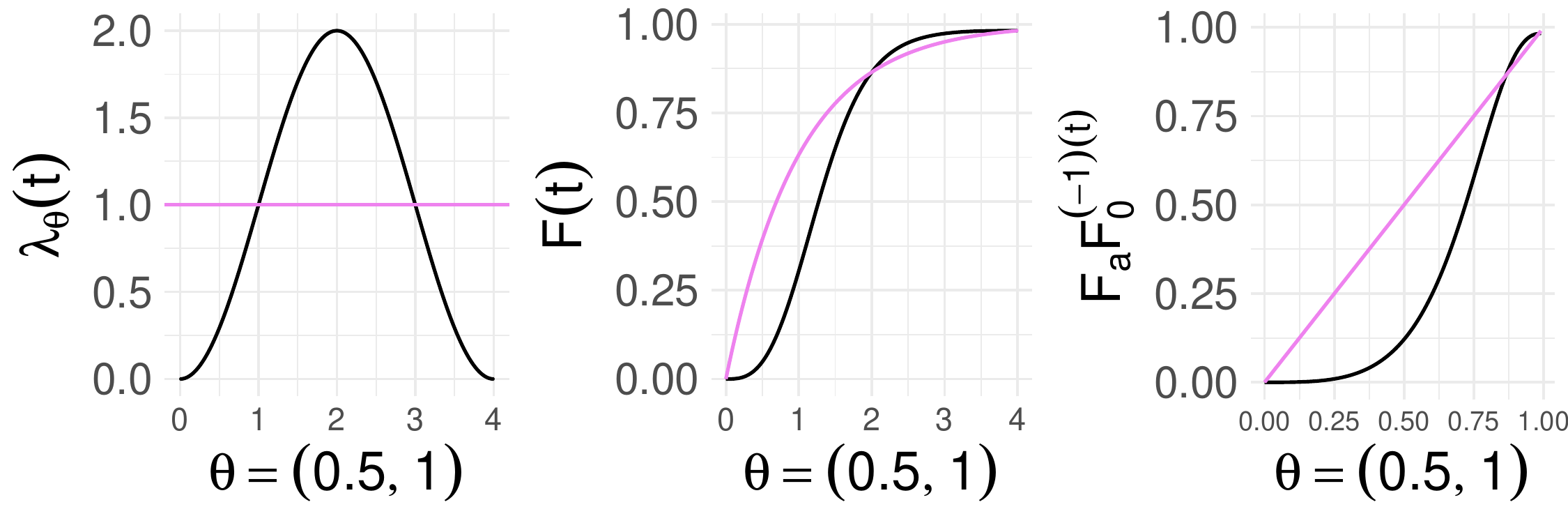}\\
        \includegraphics[width=0.9\textwidth]{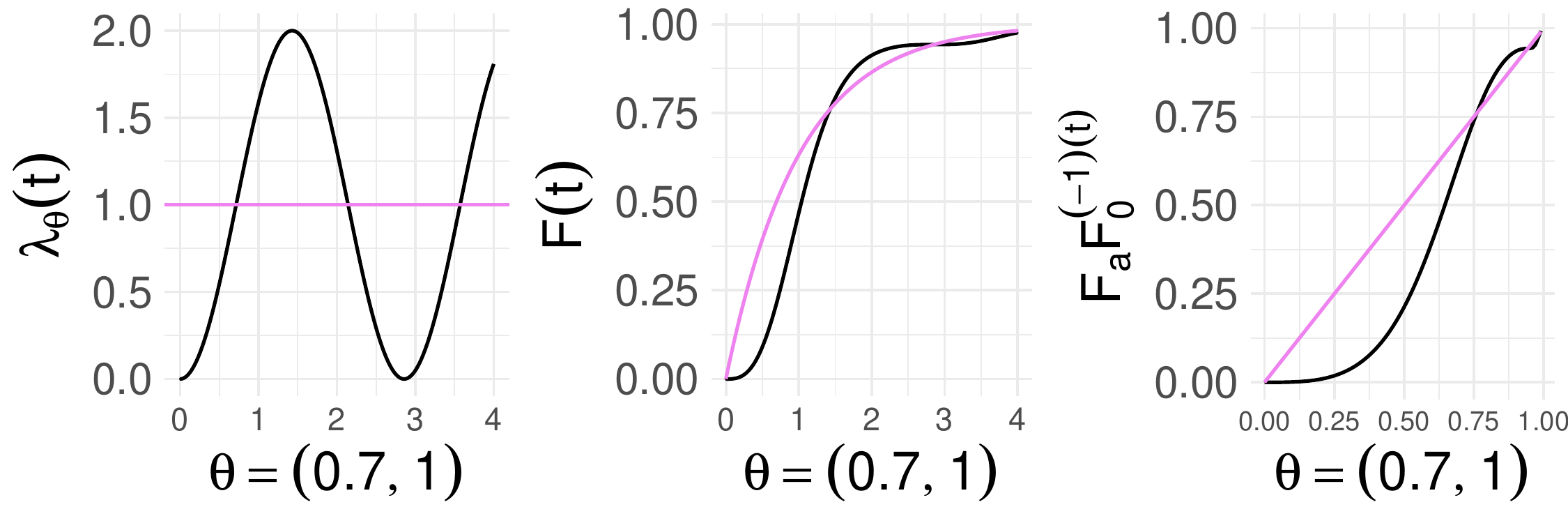}\\
        \includegraphics[width=0.9\textwidth]{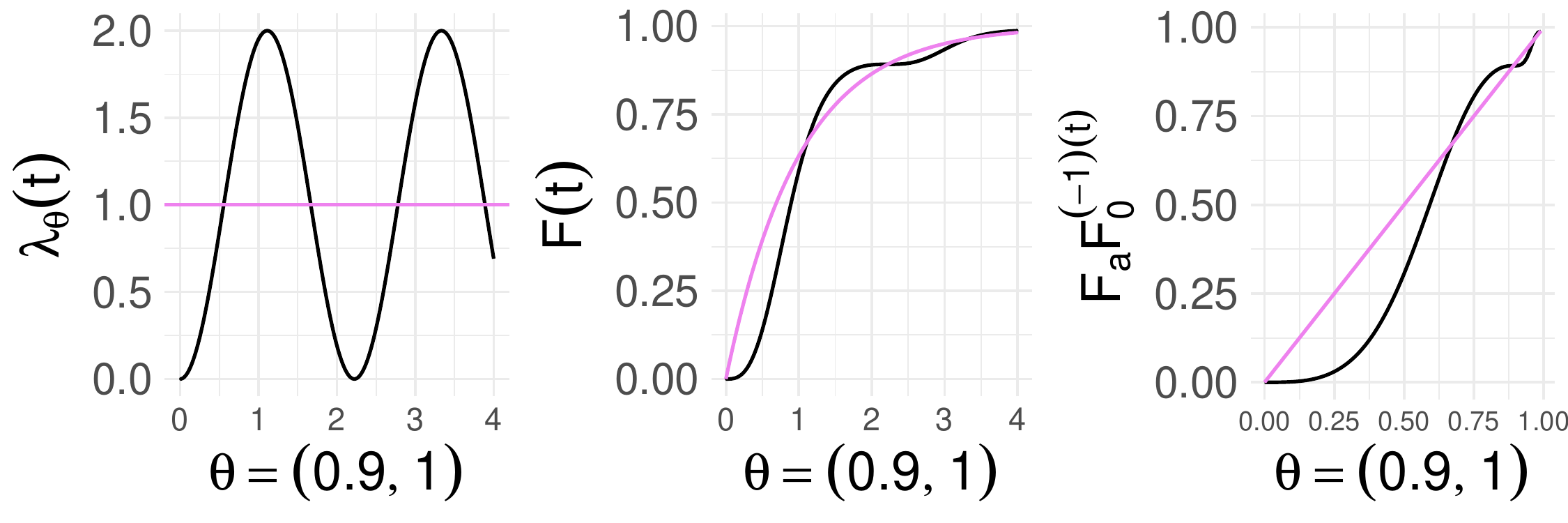}
        \end{tabular}

\begin{tabular}{c}
        \includegraphics[width=0.9\textwidth]{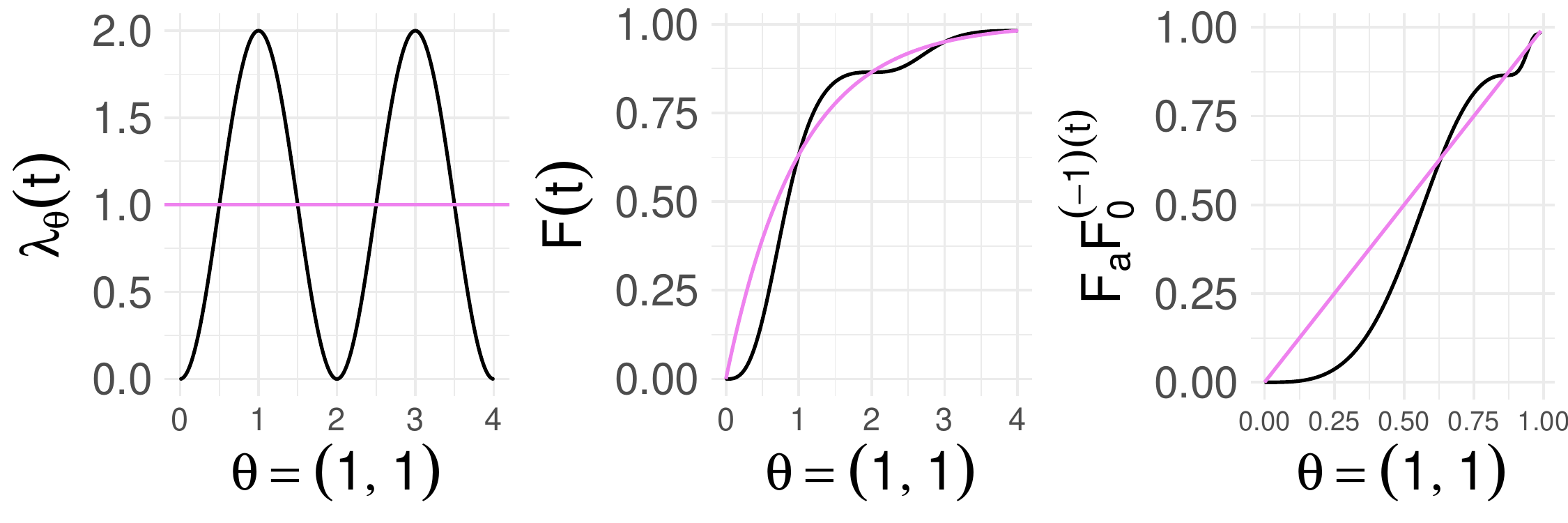}\\        
        \includegraphics[width=0.9\textwidth]{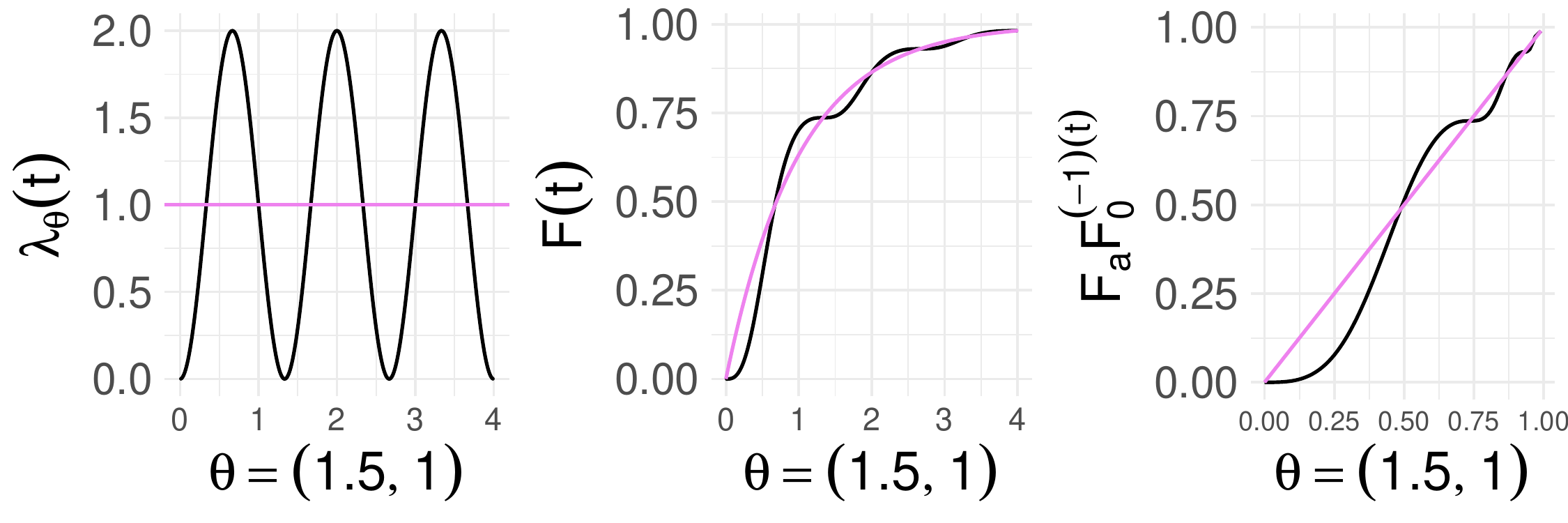}\\
        \includegraphics[width=0.9\textwidth]{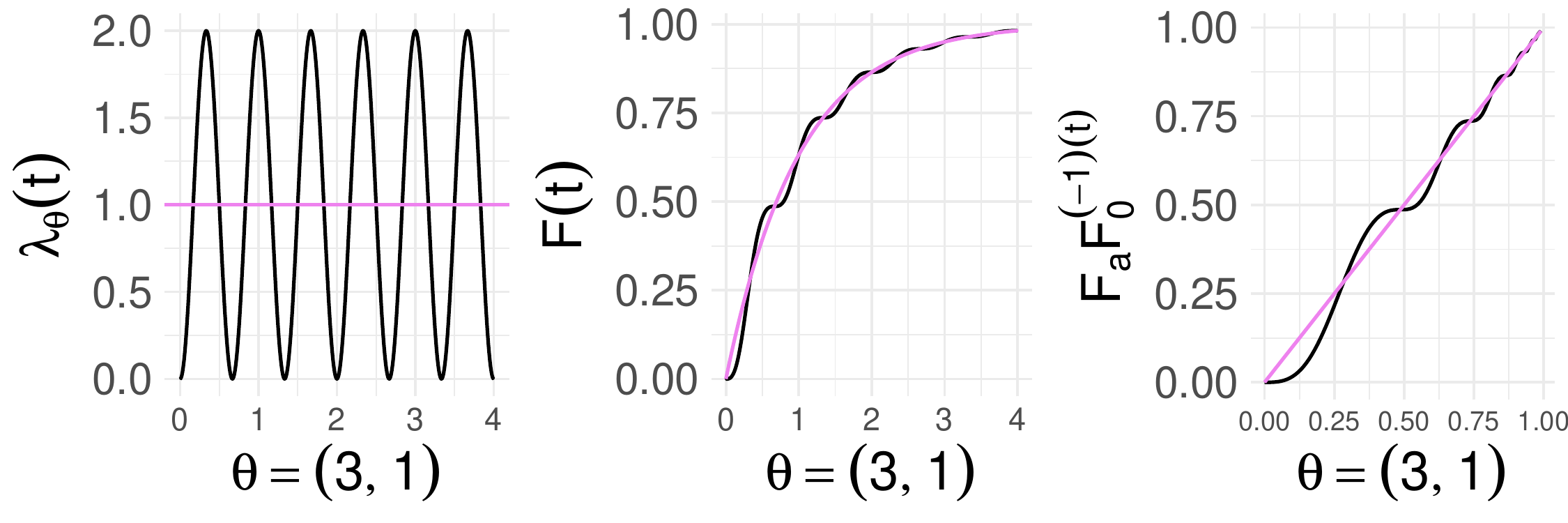}
\end{tabular}
\newpage
\subsection{Type I error}
For this experiment, the null hypothesis is recovered when $\theta_2=0$. In {\color{red}{red}} we observe tests that have an clear incorrect level. In{ \color{orange}{orange}}, we observe tests that have a questionable incorrect level. Tables are based on 2000 independent experiments.
\begin{multicols}{2}
\textbf{Fixed length-scale 1}\\
\textbf{Censoring distribution $G(t)=1-e^{-1/2t}$}
\begin{table}[H]
\centering
\resizebox{1\columnwidth}{!}{
\begin{tabular}{rrrrrrrr}
\hline
  \multicolumn{8}{c}{Type-I error}\\
  \hline
 & MW1 & MW2 & MW3 & Pearson & LR1 & LR2 &WLR\\ 
  \hline
  \multicolumn{8}{l}{Sample size n=30} \\
  10 \% & 9.30 & 10.65 & 9.20 & \color{red}{14.75} & 10.60 & 10.05 & \color{red}{20.95} \\ 
  5 \% & 4.70 & 5.60 & 4.75 & \color{red}{9.50} & 5.45 & 5.60 & \color{red}{14.85} \\ 
  1 \% & 1.25 & 2.05 & 0.95 & \color{red}{4.75} & 1.25 & 1.65 & \color{red}{6.35} \\ 
\hline
    \multicolumn{8}{l}{Sample size n=50} \\
 10 \% & 10.55 & 11.70 & 10.90 & \color{orange}{13.95} & 10.35 & 10.80 & \color{red}{15.50} \\ 
  5 \% & 5.70 & 6.70 & 5.90 & \color{orange}{8.35} & 5.55 & 5.75 & \color{red}{9.70} \\ 
  1 \% & 1.25 & 1.55 & 1.35 & \color{orange}{3.60} & 1.90 & 1.80 & \color{orange}{3.20} \\ 
  \hline
  \multicolumn{8}{l}{Sample size n=100} \\
  10 \% & 9.35 & 9.75 & 9.55 & 12.50 & 10.15 & 9.85 & \color{red}{14.50} \\ 
  5 \% & 4.90 & 5.35 & 4.85 & 7.25 & 5.50 & 5.50 & \color{red}{8.50} \\ 
  1 \% & 0.95 & 1.35 & 1.15 & 1.70 & 1.45 & 1.10 & \color{orange}{3.15} \\\hline
   \multicolumn{8}{l}{Sample size n=200} \\
10 \% & 10.80 & 11.05 & 10.85 & 10.10 & 10.75 & 10.40 & 10.90 \\ 
  5 \% & 5.65 & 5.95 & 5.50 & 4.90 & 5.60 & 5.95 & 6.20 \\ 
  1 \% & 1.00 & 1.25 & 1.10 & 0.90 & 1.35 & 0.85 & 1.95 \\
  \hline
\end{tabular}
}
\end{table}
\textbf{Adaptive length-scale}\\
\textbf{Censoring distribution $G(t)=1-e^{-1/2t}$}
\begin{table}[H]
\centering
\resizebox{1\columnwidth}{!}{
\begin{tabular}{rrrrrrrr}
\hline
  \multicolumn{8}{c}{Type-I error}\\
  \hline
 & MW1 & MW2 & MW3 & Pearson & LR1 & LR2 &WLR\\ 
  \hline
  \multicolumn{8}{l}{Sample size n=30} \\
10 \% & 10.00 & 11.05 & 9.90 & \color{red}{14.75} & 10.60 & 10.05 & \color{red}{20.95} \\ 
  5 \% & 5.25 & 6.15 & 4.85 & \color{red}{9.50} & 5.45 & 5.60 & \color{red}{14.85} \\ 
  1 \% & 1.05 & 1.50 & 0.95 & \color{red}{4.75} & 1.25 & 1.65 & \color{red}{6.35} \\ 
  \hline
  \multicolumn{8}{l}{Sample size n=50} \\
10 \% & 9.85 & 10.70 & 9.85 & \color{orange}{ 13.95} & 10.35 & 10.80 & \color{red}{15.50} \\ 
  5 \% & 5.20 & 5.90 & 4.95 & \color{orange}{8.35} & 5.55 & 5.75 & \color{red}{9.70} \\ 
  1 \% & 1.25 & 1.50 & 1.15 & \color{orange}{3.60} & 1.90 & 1.80 & \color{orange}{3.20} \\ 
  \hline
  \multicolumn{8}{l}{Sample size n=100} \\
 10 \% & 9.80 & 10.30 & 9.40 & 12.50 & 10.15 & 9.85 & \color{red}{14.50} \\ 
  5 \% & 4.35 & 4.55 & 4.10 & 7.25 & 5.50 & 5.50 & \color{red}{8.50} \\ 
  1 \% & 1.10 & 1.20 & 1.05 & 1.70 & 1.45 & 1.10 & \color{orange}{3.15} \\ 
 \hline
   \multicolumn{8}{l}{Sample size n=200} \\
10 \% & 10.55 & 10.90 & 10.55 & 10.10 & 10.75 & 10.40 & 10.90 \\ 
  5 \% & 5.65 & 5.70 & 5.70 & 4.90 & 5.60 & 5.95 & 6.20 \\ 
  1 \% & 0.65 & 0.75 & 0.95 & 0.90 & 1.35 & 0.85 & 1.95 \\ 
  \hline
\end{tabular}
}
\end{table}
\end{multicols}

\begin{multicols}{2}
\textbf{Fixed length-scale 1}\\
\textbf{Censoring distribution $G(t)=1-e^{-2t}$}
\begin{table}[H]
\centering
\resizebox{1\columnwidth}{!}{
\begin{tabular}{rrrrrrrr}
\hline
  \multicolumn{8}{c}{Type-I error}\\
  \hline
 & MW1 & MW2 & MW3 & Pearson & LR1 & LR2 &WLR\\ 
  \hline
  \multicolumn{8}{l}{Sample size n=30} \\
10 \% & 10.05 & 11.15 & 9.80 & \color{orange}{13.60} & 10.00 & 10.55 & \color{orange}{13.85} \\ 
  5 \% & 5.65 & 6.55 & 5.45 & \color{red}{9.65} & 5.05 & 5.55 & \color{red}{9.50} \\ 
  1 \% & 1.25 & 2.30 & 1.05 & \color{red}{5.45} & 1.15 & 1.55 & \color{red}{4.95} \\ 
  \hline
    \multicolumn{8}{l}{Sample size n=50} \\
10 \% & 10.80 & 11.20 & 10.50 & \color{red}{13.95} & 11.20 & 10.15 & \color{orange}{13.75} \\ 
  5 \% & 5.80 & 6.50 & 5.80 & \color{red}{9.25} & 6.00 & 5.85 & \color{red}{9.65} \\ 
  1 \% & 1.35 & 1.75 & 1.40 & \color{red}{4.55} & 1.60 & 1.60 & \color{red}{5.70} \\ 
  \hline
  \multicolumn{8}{l}{Sample size n=100} \\
  10 \% & 10.45 & 10.90 & 10.35 & \color{red}{13.55} & 9.55 & 10.05 & \color{orange}{12.05} \\ 
  5 \% & 5.55 & 5.85 & 5.35 & \color{red}{9.65} & 5.45 & 5.05 & \color{orange}{7.95} \\ 
  1 \% & 1.55 & 1.55 & 1.35 & \color{red}{5.10} & 1.45 & 1.45 & \color{orange}{3.55} \\\hline
  \multicolumn{8}{l}{Sample size n=200} \\
  10 \% & 9.85 & 10.20 & 10.25 &  \color{red}{14.15 }& 9.65 & 10.00 & 11.30 \\ 
  5 \% & 4.80 & 5.25 & 5.10 &  \color{red}{9.50} & 4.75 & 4.90 & 6.65 \\ 
  1 \% & 0.85 & 1.10 & 0.85 &  \color{red}{4.30} & 1.00 & 0.90 & 2.60 \\ 
\hline
\end{tabular}
}
\end{table}

\textbf{Adaptive length-scale}\\
\textbf{Censoring distribution $G(t)=1-e^{-2t}$}

\begin{table}[H]
\centering
\resizebox{1\columnwidth}{!}{
\begin{tabular}{rrrrrrrr}
\hline
  \multicolumn{8}{c}{Type-I error}\\
  \hline
 & MW1 & MW2 & MW3 & Pearson & LR1 & LR2 &WLR\\ 
  \hline
  \multicolumn{8}{l}{Sample size n=30} \\
10 \% & 10.25 & 11.30 & 10.20 & \color{orange}{13.60} & 10.00 & 10.55 & \color{orange}{13.85} \\ 
  5 \% & 5.40 & 6.35 & 5.20 & \color{red}{9.65} & 5.05 & 5.55 & \color{red}{9.50} \\ 
  1 \% & 1.20 & 1.50 & 1.40 & \color{red}{5.45} & 1.15 & 1.55 & \color{red}{4.95} \\ 
  \hline
  \multicolumn{8}{l}{Sample size n=50} \\
10 \% & 10.55 & 10.90 & 10.65 & \color{red}{13.95} & 11.20 & 10.15 & \color{orange}{13.75} \\ 
  5 \% & 5.65 & 5.90 & 5.40 & \color{red}{9.25} & 6.00 & 5.85 & \color{red}{9.65} \\ 
  1 \% & 1.35 & 1.55 & 1.35 & \color{red}{4.55} & 1.60 & 1.60 & \color{red}{5.70} \\ 
  \hline
  \multicolumn{8}{l}{Sample size n=100} \\
 10 \% & 10.30 & 10.45 & 10.30 & \color{red}{13.55} & 9.55 & 10.05 & \color{orange}{12.05} \\ 
  5 \% & 5.45 & 5.45 & 5.30 & \color{red}{9.65} & 5.45 & 5.05 & \color{orange}{7.95} \\ 
  1 \% & 1.65 & 1.75 & 1.50 & \color{red}{5.10} & 1.45 & 1.45 & \color{orange}{3.55} \\ \hline
    \multicolumn{8}{l}{Sample size n=200} \\
10 \% & 10.25 & 10.40 & 10.25 & \color{red}{14.15} & 9.65 & 10.00 & 11.30 \\ 
  5 \% & 4.85 & 4.95 & 5.00 &  \color{red}{9.50} & 4.75 & 4.90 & 6.65 \\ 
  1 \% & 0.90 & 0.80 & 0.85 &  \color{red}{4.30} & 1.00 & 0.90 & 2.60 \\ 
\hline
\end{tabular}
}
\end{table}
\end{multicols}
\newpage
\subsection{Power tables}
In {\color{red}{red}} we observe tests that have an clear incorrect level and thus estimated power must be looked with caution. In{ \color{orange}{orange}} are tests that have a questionable incorrect level. Tables are based on $2000$ independent experiments. 

\subsubsection{Sample size $30$}
\begin{multicols}{2}
\textbf{Fixed length-scale 1}\\
\textbf{Censoring distribution $G(t)=1-e^{-1/2t}$}

\begin{table}[H]
\centering
\resizebox{1\columnwidth}{!}{
\begin{tabular}{rrrrarra}
\hline
  \multicolumn{8}{c}{Fixed length-scale 1}\\
  \hline
  \multicolumn{8}{c}{Sample size $n=30$; $\gamma=1/2$}\\
  \hline
 & MW1 & MW2 & MW3 & Pearson & LR1 & LR2& WLR \\ 
  \hline
  \multicolumn{8}{l}{$\theta=(0.5,1)$}\\
   10 \% & 100.00 & 100.00 & 100.00 & 99.40 & 89.65 & 99.95 & 62.15 \\ 
  5 \% & 99.95 & 99.95 & 99.95 & 96.95 & 74.50 & 99.80 & 44.60 \\ 
  1 \% & 99.90 & 99.95 & 99.90 & 76.80 & 31.25 & 94.20 & 17.30 \\ 
\hline
     \multicolumn{8}{l}{$\theta=(0.7,1)$}\\
 10 \% & 99.80 & 99.90 & 99.85 & 95.85 & 43.95 & 97.95 & 40.85 \\ 
  5 \% & 99.45 & 99.55 & 99.45 & 86.05 & 27.45 & 90.25 & 27.90 \\ 
  1 \% & 96.45 & 97.95 & 96.50 & 56.55 & 6.40 & 49.15 & 10.70 \\ 
  \hline
        \multicolumn{8}{l}{$\theta=(0.9,1)$}\\
10 \% & 97.25 & 97.75 & 97.20 & 89.35 & 20.65 & 80.00 & 31.40 \\ 
  5 \% & 93.40 & 94.45 & 93.25 & 76.35 & 10.80 & 56.55 & 20.55 \\ 
  1 \% & 77.05 & 82.10 & 76.70 & 46.60 & 1.65 & 13.00 & 8.75 \\ 
   \hline
    \multicolumn{8}{l}{$\theta=(1,1)$}\\
10 \% & 93.10 & 94.25 & 93.10 & 87.35 & 16.20 & 62.80 & 31.35 \\ 
  5 \% & 84.70 & 86.60 & 84.55 & 73.60 & 8.05 & 39.55 & 21.55 \\ 
  1 \% & 60.00 & 65.75 & 59.80 & 42.05 & 0.80 & 6.55 & 10.30 \\ 
\hline
       \multicolumn{8}{l}{$\theta=(1.5,1)$}\\
10 \% & 40.05 & 43.00 & 40.15 & 62.95 & 8.40 & 18.20 & 25.20 \\ 
  5 \% & 26.45 & 29.95 & 26.20 & 46.20 & 3.85 & 7.35 & 17.85 \\ 
  1 \% & 8.60 & 11.90 & 8.40 & 21.75 & 0.55 & 1.00 & 8.70 \\ 
 \hline
       \multicolumn{8}{l}{$\theta=(3,1)$}\\
10 \% & 11.60 & 13.20 & 11.70 & 20.00 & 8.95 & 5.35 & 19.35 \\ 
  5 \% & 5.20 & 7.15 & 5.25 & 12.70 & 4.55 & 2.20 & 13.60 \\ 
  1 \% & 1.50 & 2.25 & 1.55 & 4.90 & 0.95 & 0.25 & 6.70 \\ 
  \hline
\end{tabular}
}
\end{table}
\textbf{Adaptive length-scale}\\
\textbf{Censoring distribution $G(t)=1-e^{-1/2t}$}
\begin{table}[H]
\centering
\resizebox{1\columnwidth}{!}{
\begin{tabular}{rrrrarra}
\hline
  \multicolumn{8}{c}{Adaptive length-scale}\\
  \hline
  \multicolumn{8}{c}{Sample size $n=30$; $\gamma=1/2$}\\
  \hline
 & MW1 & MW2 & MW3 & Pearson & LR1 & LR2& WLR \\ 
  \hline
  \multicolumn{8}{l}{$\theta=(0.5,1)$}\\
  10 \% & 99.95 & 99.95 & 99.95 & 99.40 & 89.65 & 99.95 &70.15\\ 
  5 \% & 99.95 & 99.95 & 99.95 & 96.95 & 74.50 & 99.80 &52.40\\ 
  1 \% & 99.75 & 99.80 & 99.75 & 76.80 & 31.25 & 94.20 &21.35\\ 
   \hline
     \multicolumn{8}{l}{$\theta=(0.7,1)$}\\
10 \% & 99.90 & 99.90 & 99.90 & 95.85 & 43.95 & 97.95 &46.50\\ 
  5 \% & 99.75 & 99.85 & 99.75 & 86.05 & 27.45 & 90.25 &29.20\\ 
  1 \% & 96.30 & 97.50 & 96.85 & 56.55 & 6.40 & 49.15 &8.60\\ 
   \hline
        \multicolumn{8}{l}{$\theta=(0.9,1)$}\\
10 \% & 99.25 & 99.35 & 99.30 & 89.35 & 20.65 & 80.00 &33.75\\ 
  5 \% & 97.90 & 98.40 & 98.00 & 76.35 & 10.80 & 56.55 &20.20\\ 
  1 \% & 87.80 & 89.15 & 87.70 & 46.60 & 1.65 & 13.00 &6.00\\ 
   \hline
    \multicolumn{8}{l}{$\theta=(1,1)$}\\
  10 \% & 98.60 & 99.05 & 98.60 & 87.35 & 16.20 & 62.80 &30.55\\ 
  5 \% & 95.50 & 96.25 & 95.30 & 73.60 & 8.05 & 39.55 &18.30\\ 
  1 \% & 82.00 & 83.75 & 82.40 & 42.05 & 0.80 & 6.55 &5.35\\ 
\hline
       \multicolumn{8}{l}{$\theta=(1.5,1)$}\\
   10 \% & 84.85 & 86.65 & 84.50 & 62.95 & 8.40 & 18.20 &25.20\\ 
  5 \% & 71.70 & 73.65 & 71.50 & 46.20 & 3.85 & 7.35 &15.50\\ 
  1 \% & 39.90 & 43.40 & 40.85 & 21.75 & 0.55 & 1.00 &5.50\\ 
\hline
       \multicolumn{8}{l}{$\theta=(3,1)$}\\
  10 \% & 23.40 & 25.80 & 23.40 & 20.00 & 8.95 & 5.35 &18.85\\ 
  5 \% & 13.40 & 15.65 & 13.60 & 12.70 & 4.55 & 2.20 &11.05\\ 
  1 \% & 4.00 & 5.60 & 4.10 & 4.90 & 0.95 & 0.25 &3.40\\ 
\hline
\end{tabular}
}
\end{table}
\end{multicols}

\begin{multicols}{2}
\textbf{Fixed length-scale 1}\\
\textbf{Censoring distribution $G(t)=1-e^{-2t}$}
\begin{table}[H]
\centering
\resizebox{1\columnwidth}{!}{
\begin{tabular}{rrrrarra}
\hline
  \multicolumn{8}{c}{Fixed length-scale 1}\\
  \hline
  \multicolumn{8}{c}{Sample size $n=30$; $\gamma=2$}\\
  \hline
 & MW1 & MW2 & MW3 & Pearson & LR1 & LR2& WLR \\ 
  \hline
  \multicolumn{8}{l}{$\theta=(0.5,1)$}\\
   10 \% & 99.95 & 99.95 & 99.95 & 86.00 & 98.35 & 99.95 & 47.55 \\ 
  5 \% & 99.95 & 99.95 & 99.95 & 62.10 & 93.30 & 99.65 & 24.50 \\ 
  1 \% & 99.80 & 99.90 & 99.85 & 18.85 & 62.05 & 84.00 & 3.25 \\ 
\hline
     \multicolumn{8}{l}{$\theta=(0.7,1)$}\\
10 \% & 99.70 & 99.75 & 99.70 & 69.70 & 82.55 & 99.40 & 23.60 \\ 
  5 \% & 99.45 & 99.50 & 99.45 & 46.10 & 66.35 & 96.15 & 10.80 \\ 
  1 \% & 96.65 & 97.65 & 96.50 & 18.80 & 24.05 & 54.95 & 3.10 \\ 
   \hline
        \multicolumn{8}{l}{$\theta=(0.9,1)$}\\
10 \% & 98.15 & 98.35 & 98.20 & 60.30 & 58.95 & 95.90 & 13.35 \\ 
  5 \% & 96.10 & 96.85 & 96.15 & 39.05 & 36.75 & 85.85 & 7.40 \\ 
  1 \% & 87.60 & 90.50 & 87.80 & 19.20 & 8.85 & 28.00 & 3.55 \\ 
   \hline
    \multicolumn{8}{l}{$\theta=(1,1)$}\\
10 \% & 96.20 & 96.80 & 96.15 & 54.85 & 45.60 & 92.20 & 12.30 \\ 
  5 \% & 92.60 & 93.75 & 92.65 & 35.25 & 27.45 & 75.55 & 7.45 \\ 
  1 \% & 77.45 & 81.65 & 77.80 & 17.60 & 4.65 & 17.80 & 3.55 \\ 
   \hline
       \multicolumn{8}{l}{$\theta=(1.5,1)$}\\
10 \% & 65.35 & 68.30 & 65.25 & 44.65 & 17.25 & 54.85 & 9.45 \\ 
  5 \% & 51.70 & 56.20 & 52.10 & 29.35 & 7.90 & 30.55 & 6.45 \\ 
  1 \% & 30.75 & 35.90 & 30.50 & 14.70 & 0.90 & 1.50 & 3.60 \\ 
 \hline
       \multicolumn{8}{l}{$\theta=(3,1)$}\\
10 \% & 14.70 & 16.25 & 14.55 & 15.75 & 7.75 & 8.65 & 10.35 \\ 
  5 \% & 8.90 & 10.95 & 8.90 & 11.45 & 3.15 & 2.65 & 7.25 \\ 
  1 \% & 2.65 & 3.95 & 2.80 & 5.55 & 0.40 & 0.00 & 3.55 \\ 
\hline
\end{tabular}
}
\end{table}

\textbf{Adaptive length-scale}\\
\textbf{Censoring distribution $G(t)=1-e^{-2t}$}
\begin{table}[H]
\centering
\resizebox{1\columnwidth}{!}{
\begin{tabular}{rrrrarra}
\hline
  \multicolumn{8}{c}{Adaptive length-scale}\\
  \hline
  \multicolumn{8}{c}{Sample size $n=30$; $\gamma=2$}\\
  \hline
 & MW1 & MW2 & MW3 & Pearson & LR1 & LR2& WLR \\ 
  \hline
  \multicolumn{8}{l}{$\theta=(0.5,1)$}\\
10 \% & 99.95 & 99.95 & 99.95 & 86.00 & 98.35 & 99.95 &52.40\\ 
  5 \% & 99.90 & 99.90 & 99.90 & 62.10 & 93.30 & 99.65 &29.75\\ 
  1 \% & 99.55 & 99.55 & 99.45 & 18.85 & 62.05 & 84.00 &4.85\\ 
   \hline
     \multicolumn{8}{l}{$\theta=(0.7,1)$}\\
10 \% & 99.80 & 99.80 & 99.75 & 69.70 & 82.55 & 99.40 &28.25\\ 
  5 \% & 99.45 & 99.60 & 99.45 & 46.10 & 66.35 & 96.15 &12.75\\ 
  1 \% & 95.70 & 95.50 & 95.90 & 18.80 & 24.05 & 54.95 &2.80\\ 
\hline
        \multicolumn{8}{l}{$\theta=(0.9,1)$}\\
10 \% & 99.25 & 99.25 & 99.20 & 60.30 & 58.95 & 95.90 &17.55\\ 
  5 \% & 97.80 & 97.65 & 97.80 & 39.05 & 36.75 & 85.85 &7.55\\ 
  1 \% & 89.50 & 88.45 & 89.55 & 19.20 & 8.85 & 28.00 &2.65\\ 
   \hline
    \multicolumn{8}{l}{$\theta=(1,1)$}\\
  10 \% & 98.45 & 98.40 & 98.40 & 54.85 & 45.60 & 92.20 &15.15\\ 
  5 \% & 96.05 & 96.05 & 95.80 & 35.25 & 27.45 & 75.55 &7.40\\ 
  1 \% & 83.85 & 81.10 & 83.50 & 17.60 & 4.65 & 17.80 &2.55\\ 
\hline
       \multicolumn{8}{l}{$\theta=(1.5,1)$}\\
10 \% & 88.30 & 88.70 & 88.20 & 44.65 & 17.25 & 54.85 &7.90\\ 
  5 \% & 77.55 & 77.85 & 77.55 & 29.35 & 7.90 & 30.55 &3.60\\ 
  1 \% & 51.10 & 47.00 & 51.40 & 14.70 & 0.90 & 1.50 &1.40\\ 
\hline
       \multicolumn{8}{l}{$\theta=(3,1)$}\\
10 \% & 38.85 & 40.30 & 38.75 & 15.75 & 7.75 & 8.65 &6.30\\ 
  5 \% & 22.35 & 23.80 & 22.35 & 11.45 & 3.15 & 2.65 &3.85\\ 
  1 \% & 5.40 & 6.30 & 5.10 & 5.55 & 0.40 & 0.00 &2.25\\ 
\hline
\end{tabular}
}
\end{table}
\end{multicols}

\newpage
\subsubsection{Sample size $50$}
\begin{multicols}{2}
\textbf{Fixed length-scale 1}\\
\textbf{Censoring distribution $G(t)=1-e^{-1/2t}$}

\begin{table}[H]
\centering
\resizebox{1\columnwidth}{!}{
\begin{tabular}{rrrrbrra}
\hline
  \multicolumn{8}{c}{Fixed length-scale 1}\\
  \hline
  \multicolumn{8}{c}{Sample size $n=50$; $\gamma=1/2$}\\
  \hline
 & MW1 & MW2 & MW3 & Pearson & LR1 & LR2& WLR \\ 
  \hline
  \multicolumn{8}{l}{$\theta=(0.5,1)$}\\
   10 \% & 100.00 & 100.00 & 100.00 & 100.00 & 99.25 & 100.00 & 94.10 \\ 
  5 \% & 100.00 & 100.00 & 100.00 & 99.95 & 97.50 & 100.00 & 88.35 \\ 
  1 \% & 100.00 & 100.00 & 100.00 & 99.85 & 79.35 & 99.95 & 61.25 \\ 
\hline
     \multicolumn{8}{l}{$\theta=(0.7,1)$}\\
10 \% & 100.00 & 100.00 & 100.00 & 99.95 & 71.50 & 99.95 & 76.10 \\ 
  5 \% & 100.00 & 100.00 & 100.00 & 99.85 & 53.75 & 99.85 & 60.75 \\ 
  1 \% & 99.95 & 99.95 & 99.90 & 96.25 & 17.95 & 95.15 & 30.90 \\ 
   \hline
        \multicolumn{8}{l}{$\theta=(0.9,1)$}\\
10 \% & 99.95 & 99.90 & 99.95 & 99.55 & 36.85 & 98.15 & 57.10 \\ 
  5 \% & 99.80 & 99.80 & 99.80 & 98.70 & 22.25 & 91.90 & 42.00 \\ 
  1 \% & 96.60 & 97.40 & 96.85 & 87.45 & 4.50 & 59.05 & 18.20 \\ 
   \hline
    \multicolumn{8}{l}{$\theta=(1,1)$}\\
10 \% & 99.45 & 99.45 & 99.45 & 99.45 & 25.75 & 92.75 & 53.15 \\ 
  5 \% & 98.10 & 98.45 & 98.25 & 97.50 & 14.45 & 79.45 & 38.15 \\ 
  1 \% & 88.70 & 90.50 & 89.00 & 83.25 & 2.25 & 35.60 & 16.20 \\ 
\hline
       \multicolumn{8}{l}{$\theta=(1.5,1)$}\\
10 \% & 62.50 & 64.30 & 63.10 & 87.80 & 11.25 & 37.00 & 39.95 \\ 
  5 \% & 43.60 & 46.50 & 43.55 & 76.70 & 5.35 & 19.70 & 28.10 \\ 
  1 \% & 16.95 & 19.85 & 16.60 & 49.20 & 1.05 & 3.00 & 12.90 \\ 
 \hline
       \multicolumn{8}{l}{$\theta=(3,1)$}\\
10 \% & 12.65 & 13.75 & 12.45 & 22.20 & 8.65 & 7.50 & 18.30 \\ 
  5 \% & 6.85 & 7.80 & 6.90 & 14.20 & 4.75 & 3.25 & 12.80 \\ 
  1 \% & 1.65 & 2.90 & 2.05 & 5.75 & 1.00 & 0.30 & 5.60 \\ 
  \hline
\end{tabular}
}
\end{table}

\textbf{Adaptive length-scale}\\
\textbf{Censoring distribution $G(t)=1-e^{-1/2t}$}

\begin{table}[H]
\centering
\resizebox{1\columnwidth}{!}{
\begin{tabular}{rrrrbrra}
\hline
  \multicolumn{8}{c}{Adaptive length-scale}\\
  \hline
  \multicolumn{8}{c}{Sample size $n=50$; $\gamma=1/2$}\\
  \hline
 & MW1 & MW2 & MW3 & Pearson & LR1 & LR2& WLR \\ 
  \hline
  \multicolumn{8}{l}{$\theta=(0.5,1)$}\\
10 \% & 100.00 & 100.00 & 100.00 & 100.00 & 99.25 & 100.00 & 94.10 \\ 
  5 \% & 100.00 & 100.00 & 100.00 & 99.95 & 97.50 & 100.00 & 88.35 \\ 
  1 \% & 100.00 & 100.00 & 100.00 & 99.85 & 79.35 & 99.95 & 61.25 \\ 
   \hline
     \multicolumn{8}{l}{$\theta=(0.7,1)$}\\
10 \% & 100.00 & 100.00 & 100.00 & 99.95 & 71.50 & 99.95 & 76.10 \\ 
  5 \% & 100.00 & 100.00 & 100.00 & 99.85 & 53.75 & 99.85 & 60.75 \\ 
  1 \% & 99.95 & 100.00 & 99.90 & 96.25 & 17.95 & 95.15 & 30.90 \\ 
\hline
        \multicolumn{8}{l}{$\theta=(0.9,1)$}\\
10 \% & 100.00 & 100.00 & 100.00 & 99.55 & 36.85 & 98.15 & 57.10 \\ 
  5 \% & 99.95 & 99.95 & 99.95 & 98.70 & 22.25 & 91.90 & 42.00 \\ 
  1 \% & 99.55 & 99.75 & 99.50 & 87.45 & 4.50 & 59.05 & 18.20 \\ 
   \hline
    \multicolumn{8}{l}{$\theta=(1,1)$}\\
10 \% & 100.00 & 100.00 & 100.00 & 99.45 & 25.75 & 92.75 & 53.15 \\ 
  5 \% & 99.85 & 99.90 & 99.80 & 97.50 & 14.45 & 79.45 & 38.15 \\ 
  1 \% & 98.50 & 98.75 & 98.80 & 83.25 & 2.25 & 35.60 & 16.20 \\ 
\hline
       \multicolumn{8}{l}{$\theta=(1.5,1)$}\\
10 \% & 98.35 & 98.45 & 98.45 & 87.80 & 11.25 & 37.00 & 39.95 \\ 
  5 \% & 94.20 & 94.65 & 94.20 & 76.70 & 5.35 & 19.70 & 28.10 \\ 
  1 \% & 74.65 & 75.95 & 74.30 & 49.20 & 1.05 & 3.00 & 12.90 \\ 
\hline
       \multicolumn{8}{l}{$\theta=(3,1)$}\\
10 \% & 33.90 & 35.70 & 34.45 & 22.20 & 8.65 & 7.50 & 18.30 \\ 
  5 \% & 20.30 & 21.50 & 19.75 & 14.20 & 4.75 & 3.25 & 12.80 \\ 
  1 \% & 6.55 & 7.20 & 6.20 & 5.75 & 1.00 & 0.30 & 5.60 \\ 
\hline
\end{tabular}
}
\end{table}
\end{multicols}

\begin{multicols}{2}
\textbf{Fixed length-scale 1}\\
\textbf{Censoring distribution $G(t)=1-e^{-2t}$}

\begin{table}[H]
\centering
\resizebox{1\columnwidth}{!}{
\begin{tabular}{rrrrarra}
\hline
  \multicolumn{8}{c}{Fixed length-scale 1}\\
  \hline
  \multicolumn{8}{c}{Sample size $n=50$; $\gamma=2$}\\
  \hline
 & MW1 & MW2 & MW3 & Pearson & LR1 & LR2& WLR \\ 
  \hline
  \multicolumn{8}{l}{$\theta=(0.5,1)$}\\
   10 \% & 100.00 & 100.00 & 100.00 & 99.95 & 100.00 & 100.00 & 94.25 \\ 
  5 \% & 100.00 & 100.00 & 100.00 & 99.45 & 99.80 & 100.00 & 83.80 \\ 
  1 \% & 100.00 & 100.00 & 100.00 & 85.00 & 96.00 & 99.95 & 40.55 \\ 
\hline
     \multicolumn{8}{l}{$\theta=(0.7,1)$}\\
10 \% & 100.00 & 100.00 & 100.00 & 99.00 & 97.95 & 100.00 & 65.80 \\ 
  5 \% & 100.00 & 100.00 & 100.00 & 95.40 & 94.20 & 99.95 & 44.75 \\ 
  1 \% & 99.90 & 99.85 & 99.85 & 62.40 & 68.75 & 98.65 & 12.75 \\ 
\hline
        \multicolumn{8}{l}{$\theta=(0.9,1)$}\\
10 \% & 99.95 & 99.95 & 99.95 & 95.50 & 84.40 & 99.95 & 41.50 \\ 
  5 \% & 99.75 & 99.80 & 99.75 & 85.45 & 70.20 & 99.50 & 24.50 \\ 
  1 \% & 98.20 & 98.75 & 98.10 & 47.50 & 30.70 & 89.05 & 5.80 \\ 
   \hline
    \multicolumn{8}{l}{$\theta=(1,1)$}\\
10 \% & 99.80 & 99.80 & 99.80 & 92.80 & 72.80 & 99.80 & 30.20 \\ 
  5 \% & 99.30 & 99.30 & 99.25 & 81.90 & 54.30 & 98.55 & 17.55 \\ 
  1 \% & 95.70 & 96.40 & 95.35 & 41.95 & 17.60 & 79.20 & 5.55 \\ 
\hline
       \multicolumn{8}{l}{$\theta=(1.5,1)$}\\
10 \% & 85.00 & 86.45 & 85.10 & 75.35 & 29.40 & 87.35 & 14.25 \\ 
  5 \% & 75.35 & 77.30 & 75.00 & 58.75 & 15.95 & 69.15 & 8.60 \\ 
  1 \% & 51.30 & 55.40 & 51.35 & 29.30 & 3.70 & 22.80 & 3.20 \\ 
\hline
       \multicolumn{8}{l}{$\theta=(3,1)$}\\
10 \% & 20.05 & 21.20 & 20.30 & 21.55 & 10.70 & 17.65 & 9.25 \\ 
  5 \% & 13.15 & 14.60 & 12.95 & 14.90 & 4.35 & 8.05 & 5.60 \\ 
  1 \% & 4.35 & 5.50 & 4.55 & 6.75 & 0.65 & 0.70 & 3.15 \\ 
\hline
\end{tabular}
}
\end{table}

\textbf{Adaptive length-scale}\\
\textbf{Censoring distribution $G(t)=1-e^{-2t}$}

\begin{table}[H]
\centering
\resizebox{1\columnwidth}{!}{
\begin{tabular}{rrrrarra}
\hline
  \multicolumn{8}{c}{Adaptive length-scale}\\
  \hline
  \multicolumn{8}{c}{Sample size $n=50$; $\gamma=2$}\\
  \hline
 & MW1 & MW2 & MW3 & Pearson & LR1 & LR2& WLR \\ 
  \hline
  \multicolumn{8}{l}{$\theta=(0.5,1)$}\\
   10 \% & 100.00 & 100.00 & 100.00 & 99.95 & 100.00 & 100.00 & 94.25 \\ 
  5 \% & 100.00 & 100.00 & 100.00 & 99.45 & 99.80 & 100.00 & 83.80 \\ 
  1 \% & 99.95 & 99.95 & 99.95 & 85.00 & 96.00 & 99.95 & 40.55 \\ 
\hline
     \multicolumn{8}{l}{$\theta=(0.7,1)$}\\
10 \% & 100.00 & 100.00 & 100.00 & 99.00 & 97.95 & 100.00 & 65.80 \\ 
  5 \% & 100.00 & 100.00 & 99.95 & 95.40 & 94.20 & 99.95 & 44.75 \\ 
  1 \% & 99.95 & 99.90 & 99.95 & 62.40 & 68.75 & 98.65 & 12.75 \\ 
   \hline
        \multicolumn{8}{l}{$\theta=(0.9,1)$}\\
10 \% & 100.00 & 100.00 & 100.00 & 95.50 & 84.40 & 99.95 & 41.50 \\ 
  5 \% & 99.95 & 99.95 & 99.95 & 85.45 & 70.20 & 99.50 & 24.50 \\ 
  1 \% & 99.05 & 99.05 & 99.15 & 47.50 & 30.70 & 89.05 & 5.80 \\ 
   \hline
    \multicolumn{8}{l}{$\theta=(1,1)$}\\
10 \% & 99.95 & 99.95 & 99.95 & 92.80 & 72.80 & 99.80 & 30.20 \\ 
  5 \% & 99.90 & 99.90 & 99.95 & 81.90 & 54.30 & 98.55 & 17.55 \\ 
  1 \% & 98.50 & 98.35 & 98.65 & 41.95 & 17.60 & 79.20 & 5.55 \\ 
\hline
       \multicolumn{8}{l}{$\theta=(1.5,1)$}\\
10 \% & 98.85 & 98.80 & 98.80 & 75.35 & 29.40 & 87.35 & 14.25 \\ 
  5 \% & 96.35 & 96.50 & 96.20 & 58.75 & 15.95 & 69.15 & 8.60 \\ 
  1 \% & 83.25 & 81.05 & 83.10 & 29.30 & 3.70 & 22.80 & 3.20 \\ 
 \hline
       \multicolumn{8}{l}{$\theta=(3,1)$}\\
10 \% & 63.15 & 64.40 & 63.55 & 21.55 & 10.70 & 17.65 & 9.25 \\ 
  5 \% & 40.80 & 42.10 & 41.05 & 14.90 & 4.35 & 8.05 & 5.60 \\ 
  1 \% & 11.80 & 12.60 & 12.05 & 6.75 & 0.65 & 0.70 & 3.15 \\ 
  \hline
\end{tabular}
}
\end{table}
\end{multicols}

\newpage
\subsubsection{Sample size $100$}
\begin{multicols}{2}
\textbf{Fixed length-scale 1}\\
\textbf{Censoring distribution $G(t)=1-e^{-1/2t}$}

\begin{table}[H]
\centering
\resizebox{1\columnwidth}{!}{
\begin{tabular}{rrrrrrra}
\hline
  \multicolumn{8}{c}{Fixed length-scale 1}\\
  \hline
  \multicolumn{8}{c}{Sample size $n=100$; $\gamma=1/2$}\\
  \hline
 & MW1 & MW2 & MW3 & Pearson & LR1 & LR2& WLR \\ 
  \hline
  \multicolumn{8}{l}{$\theta=(0.5,1)$}\\
10 \% & 100.00 & 100.00 & 100.00 & 100.00 & 100.00 & 100.00 & 100.00 \\ 
  5 \% & 100.00 & 100.00 & 100.00 & 100.00 & 100.00 & 100.00 & 100.00 \\ 
  1 \% & 100.00 & 100.00 & 100.00 & 100.00 & 99.85 & 100.00 & 99.70 \\ 
   \hline
     \multicolumn{8}{l}{$\theta=(0.7,1)$}\\
10 \% & 100.00 & 100.00 & 100.00 & 100.00 & 97.00 & 100.00 & 99.50 \\ 
  5 \% & 100.00 & 100.00 & 100.00 & 100.00 & 92.20 & 100.00 & 98.10 \\ 
  1 \% & 100.00 & 100.00 & 100.00 & 100.00 & 66.75 & 100.00 & 88.95 \\ 
\hline
        \multicolumn{8}{l}{$\theta=(0.9,1)$}\\
10 \% & 100.00 & 100.00 & 100.00 & 100.00 & 66.40 & 100.00 & 96.40 \\ 
  5 \% & 100.00 & 100.00 & 100.00 & 100.00 & 49.25 & 100.00 & 91.00 \\ 
  1 \% & 100.00 & 100.00 & 100.00 & 100.00 & 19.30 & 99.65 & 69.65 \\ 
   \hline
    \multicolumn{8}{l}{$\theta=(1,1)$}\\
10 \% & 100.00 & 100.00 & 100.00 & 100.00 & 49.65 & 100.00 & 93.65 \\ 
  5 \% & 100.00 & 100.00 & 100.00 & 100.00 & 33.25 & 99.85 & 85.25 \\ 
  1 \% & 100.00 & 100.00 & 100.00 & 100.00 & 11.55 & 94.60 & 62.30 \\ 
\hline
       \multicolumn{8}{l}{$\theta=(1.5,1)$}\\
10 \% & 94.90 & 95.45 & 95.00 & 99.85 & 16.70 & 77.75 & 79.60 \\ 
  5 \% & 84.20 & 84.65 & 84.50 & 99.10 & 8.95 & 59.10 & 65.75 \\ 
  1 \% & 50.35 & 53.25 & 50.65 & 94.70 & 1.95 & 18.70 & 38.65 \\ 
\hline
       \multicolumn{8}{l}{$\theta=(3,1)$}\\
10 \% & 13.85 & 14.90 & 13.90 & 28.55 & 9.80 & 11.50 & 27.50 \\ 
  5 \% & 8.00 & 8.20 & 7.80 & 19.15 & 5.00 & 5.80 & 18.50 \\ 
  1 \% & 2.15 & 2.45 & 2.30 & 7.25 & 0.90 & 0.95 & 7.90 \\ 
\hline
\end{tabular}
}
\end{table}
\textbf{Adaptive length-scale}\\
\textbf{Censoring distribution $G(t)=1-e^{-1/2t}$}

\begin{table}[H]
\centering
\resizebox{1\columnwidth}{!}{
\begin{tabular}{rrrrrrra}
\hline
  \multicolumn{8}{c}{Adaptive length-scale}\\
  \hline
  \multicolumn{8}{c}{Sample size $n=100$; $\gamma=1/2$}\\
  \hline
 & MW1 & MW2 & MW3 & Pearson & LR1 & LR2& WLR \\ 
  \hline
  \multicolumn{8}{l}{$\theta=(0.5,1)$}\\
  10 \% & 100.00 & 100.00 & 100.00 & 100.00 & 100.00 & 100.00 & 100.00 \\ 
  5 \% & 100.00 & 100.00 & 100.00 & 100.00 & 100.00 & 100.00 & 100.00 \\ 
  1 \% & 100.00 & 100.00 & 100.00 & 100.00 & 99.85 & 100.00 & 99.70 \\ 
 \hline
     \multicolumn{8}{l}{$\theta=(0.7,1)$}\\
10 \% & 100.00 & 100.00 & 100.00 & 100.00 & 97.00 & 100.00 & 99.50 \\ 
  5 \% & 100.00 & 100.00 & 100.00 & 100.00 & 92.20 & 100.00 & 98.10 \\ 
  1 \% & 100.00 & 100.00 & 100.00 & 100.00 & 66.75 & 100.00 & 88.95 \\ 
   \hline
        \multicolumn{8}{l}{$\theta=(0.9,1)$}\\
10 \% & 100.00 & 100.00 & 100.00 & 100.00 & 66.40 & 100.00 & 96.40 \\ 
  5 \% & 100.00 & 100.00 & 100.00 & 100.00 & 49.25 & 100.00 & 91.00 \\ 
  1 \% & 100.00 & 100.00 & 100.00 & 100.00 & 19.30 & 99.65 & 69.65 \\ 
   \hline
    \multicolumn{8}{l}{$\theta=(1,1)$}\\
10 \% & 100.00 & 100.00 & 100.00 & 100.00 & 49.65 & 100.00 & 93.65 \\ 
  5 \% & 100.00 & 100.00 & 100.00 & 100.00 & 33.25 & 99.85 & 85.25 \\ 
  1 \% & 100.00 & 100.00 & 100.00 & 100.00 & 11.55 & 94.60 & 62.30 \\ 
\hline
       \multicolumn{8}{l}{$\theta=(1.5,1)$}\\
10 \% & 100.00 & 100.00 & 100.00 & 99.85 & 16.70 & 77.75 & 79.60 \\ 
  5 \% & 100.00 & 100.00 & 100.00 & 99.10 & 8.95 & 59.10 & 65.75 \\ 
  1 \% & 99.55 & 99.70 & 99.65 & 94.70 & 1.95 & 18.70 & 38.65 \\ 
 \hline
       \multicolumn{8}{l}{$\theta=(3,1)$}\\
10 \% & 67.00 & 67.65 & 67.05 & 28.55 & 9.80 & 11.50 & 27.50 \\ 
  5 \% & 43.55 & 44.65 & 43.45 & 19.15 & 5.00 & 5.80 & 18.50 \\ 
  1 \% & 14.20 & 15.20 & 14.55 & 7.25 & 0.90 & 0.95 & 7.90 \\ 
  \hline
\end{tabular}
}
\end{table}
\end{multicols}
\begin{multicols}{2}
\textbf{Fixed length-scale 1}\\
\textbf{Censoring distribution $G(t)=1-e^{-2t}$}

\begin{table}[H]
\centering
\resizebox{1\columnwidth}{!}{
\begin{tabular}{rrrrarrb}
\hline
  \multicolumn{8}{c}{Fixed length-scale 1}\\
  \hline
  \multicolumn{8}{c}{Sample size $n=100$; $\gamma=2$}\\
  \hline
 & MW1 & MW2 & MW3 & Pearson & LR1 & LR2& WLR \\ 
  \hline
  \multicolumn{8}{l}{$\theta=(0.5,1)$}\\
10 \% & 100.00 & 100.00 & 100.00 & 100.00 & 100.00 & 100.00 & 100.00 \\ 
  5 \% & 100.00 & 100.00 & 100.00 & 100.00 & 100.00 & 100.00 & 100.00 \\ 
  1 \% & 100.00 & 100.00 & 100.00 & 100.00 & 100.00 & 100.00 & 99.55 \\ 
   \hline
     \multicolumn{8}{l}{$\theta=(0.7,1)$}\\
10 \% & 100.00 & 100.00 & 100.00 & 100.00 & 100.00 & 100.00 & 99.20 \\ 
  5 \% & 100.00 & 100.00 & 100.00 & 100.00 & 99.95 & 100.00 & 97.50 \\ 
  1 \% & 100.00 & 100.00 & 100.00 & 100.00 & 99.35 & 100.00 & 83.95 \\ 
   \hline
        \multicolumn{8}{l}{$\theta=(0.9,1)$}\\
10 \% & 100.00 & 100.00 & 100.00 & 100.00 & 99.40 & 100.00 & 90.65 \\ 
  5 \% & 100.00 & 100.00 & 100.00 & 100.00 & 97.80 & 100.00 & 79.20 \\ 
  1 \% & 100.00 & 100.00 & 100.00 & 99.55 & 85.35 & 100.00 & 43.05 \\ 
   \hline
    \multicolumn{8}{l}{$\theta=(1,1)$}\\
10 \% & 100.00 & 100.00 & 100.00 & 100.00 & 97.25 & 100.00 & 79.95 \\ 
  5 \% & 100.00 & 100.00 & 100.00 & 99.95 & 92.70 & 100.00 & 63.15 \\ 
  1 \% & 100.00 & 100.00 & 100.00 & 98.55 & 67.90 & 100.00 & 29.10 \\ 
\hline
       \multicolumn{8}{l}{$\theta=(1.5,1)$}\\
10 \% & 99.35 & 99.30 & 99.30 & 99.00 & 54.70 & 99.90 & 31.10 \\ 
  5 \% & 97.35 & 97.40 & 97.25 & 96.65 & 37.60 & 99.20 & 17.50 \\ 
  1 \% & 88.00 & 88.50 & 87.75 & 82.20 & 12.10 & 87.85 & 4.15 \\ 
 \hline
       \multicolumn{8}{l}{$\theta=(3,1)$}\\
10 \% & 27.45 & 27.95 & 27.60 & 34.45 & 12.20 & 38.90 & 14.20 \\ 
  5 \% & 17.55 & 18.45 & 17.50 & 23.25 & 5.75 & 22.45 & 9.95 \\ 
  1 \% & 5.80 & 6.55 & 6.00 & 9.80 & 1.25 & 3.90 & 4.55 \\ 
  \hline
\end{tabular}
}
\end{table}

\textbf{Adaptive length-scale}\\
\textbf{Censoring distribution $G(t)=1-e^{-2t}$}

\begin{table}[H]
\centering
\resizebox{1\columnwidth}{!}{
\begin{tabular}{rrrrarrb}
\hline
  \multicolumn{8}{c}{Adaptive length-scale}\\
  \hline
  \multicolumn{8}{c}{Sample size $n=100$; $\gamma=2$}\\
  \hline
 & MW1 & MW2 & MW3 & Pearson & LR1 & LR2& WLR \\ 
  \hline
  \multicolumn{8}{l}{$\theta=(0.5,1)$}\\
   10 \% & 100.00 & 100.00 & 100.00 & 100.00 & 100.00 & 100.00 & 100.00 \\ 
  5 \% & 100.00 & 100.00 & 100.00 & 100.00 & 100.00 & 100.00 & 100.00 \\ 
  1 \% & 100.00 & 100.00 & 100.00 & 100.00 & 100.00 & 100.00 & 99.55 \\ 
\hline
     \multicolumn{8}{l}{$\theta=(0.7,1)$}\\
10 \% & 100.00 & 100.00 & 100.00 & 100.00 & 100.00 & 100.00 & 99.20 \\ 
  5 \% & 100.00 & 100.00 & 100.00 & 100.00 & 99.95 & 100.00 & 97.50 \\ 
  1 \% & 100.00 & 100.00 & 100.00 & 100.00 & 99.35 & 100.00 & 83.95 \\ 
\hline
        \multicolumn{8}{l}{$\theta=(0.9,1)$}\\
10 \% & 100.00 & 100.00 & 100.00 & 100.00 & 99.40 & 100.00 & 90.65 \\ 
  5 \% & 100.00 & 100.00 & 100.00 & 100.00 & 97.80 & 100.00 & 79.20 \\ 
  1 \% & 100.00 & 100.00 & 100.00 & 99.55 & 85.35 & 100.00 & 43.05 \\ 
   \hline
    \multicolumn{8}{l}{$\theta=(1,1)$}\\
10 \% & 100.00 & 100.00 & 100.00 & 100.00 & 97.25 & 100.00 & 80.90 \\ 
  5 \% & 100.00 & 100.00 & 100.00 & 99.95 & 92.70 & 100.00 & 65.30 \\ 
  1 \% & 100.00 & 100.00 & 100.00 & 98.55 & 67.90 & 100.00 & 27.95 \\ 
\hline
       \multicolumn{8}{l}{$\theta=(1.5,1)$}\\
10 \% & 100.00 & 100.00 & 100.00 & 99.00 & 54.70 & 99.90 & 35.35 \\ 
  5 \% & 100.00 & 100.00 & 100.00 & 96.65 & 37.60 & 99.20 & 21.40 \\ 
  1 \% & 99.85 & 99.80 & 99.80 & 82.20 & 12.10 & 87.85 & 5.30 \\ 
\hline
       \multicolumn{8}{l}{$\theta=(3,1)$}\\
10 \% & 95.20 & 95.10 & 95.15 & 34.45 & 12.20 & 38.90 & 11.80 \\ 
  5 \% & 86.70 & 86.80 & 86.90 & 23.25 & 5.75 & 22.45 & 7.20 \\ 
  1 \% & 49.30 & 48.45 & 49.65 & 9.80 & 1.25 & 3.90 & 2.85 \\ 
\hline
\end{tabular}
}
\end{table}
\end{multicols}

\newpage
\subsubsection{Sample size $200$}
\begin{multicols}{2}
\textbf{Fixed length-scale 1}\\
\textbf{Censoring distribution $G(t)=1-e^{-1/2t}$}

\begin{table}[H]
\centering
\resizebox{1\columnwidth}{!}{
\begin{tabular}{rrrrrrrr}
\hline
  \multicolumn{8}{c}{Fixed length-scale 1}\\
  \hline
  \multicolumn{8}{c}{Sample size $n=200$; $\gamma=1/2$}\\
  \hline
 & MW1 & MW2 & MW3 & Pearson & LR1 & LR2& WLR \\ 
  \hline
  \multicolumn{8}{l}{$\theta=(0.5,1)$}\\
  10 \% & 100.00 & 100.00 & 100.00 & 100.00 & 100.00 & 100.00 & 100.00 \\ 
  5 \% & 100.00 & 100.00 & 100.00 & 100.00 & 100.00 & 100.00 & 100.00 \\ 
  1 \% & 100.00 & 100.00 & 100.00 & 100.00 & 100.00 & 100.00 & 100.00 \\ 
 \hline
     \multicolumn{8}{l}{$\theta=(0.7,1)$}\\
10 \% & 100.00 & 100.00 & 100.00 & 100.00 & 100.00 & 100.00 & 100.00 \\ 
  5 \% & 100.00 & 100.00 & 100.00 & 100.00 & 99.90 & 100.00 & 100.00 \\ 
  1 \% & 100.00 & 100.00 & 100.00 & 100.00 & 98.80 & 100.00 & 100.00 \\ 
   \hline
        \multicolumn{8}{l}{$\theta=(0.9,1)$}\\
10 \% & 100.00 & 100.00 & 100.00 & 100.00 & 95.05 & 100.00 & 100.00 \\ 
  5 \% & 100.00 & 100.00 & 100.00 & 100.00 & 88.45 & 100.00 & 99.95 \\ 
  1 \% & 100.00 & 100.00 & 100.00 & 100.00 & 64.30 & 100.00 & 99.75 \\ 
   \hline
    \multicolumn{8}{l}{$\theta=(1,1)$}\\
10 \% & 100.00 & 100.00 & 100.00 & 100.00 & 83.00 & 100.00 & 100.00 \\ 
  5 \% & 100.00 & 100.00 & 100.00 & 100.00 & 71.00 & 100.00 & 99.75 \\ 
  1 \% & 100.00 & 100.00 & 100.00 & 100.00 & 39.95 & 100.00 & 99.15 \\ 
\hline
       \multicolumn{8}{l}{$\theta=(1.5,1)$}\\
10 \% & 99.95 & 100.00 & 100.00 & 100.00 & 28.50 & 98.10 & 99.30 \\ 
  5 \% & 99.70 & 99.80 & 99.80 & 100.00 & 17.70 & 94.40 & 98.65 \\ 
  1 \% & 93.80 & 93.65 & 93.45 & 100.00 & 5.20 & 73.50 & 93.40 \\ 
 \hline
       \multicolumn{8}{l}{$\theta=(3,1)$}\\
10 \% & 21.20 & 21.05 & 20.95 & 44.70 & 10.70 & 21.85 & 55.45 \\ 
  5 \% & 12.00 & 11.95 & 11.50 & 31.70 & 4.95 & 11.40 & 44.05 \\ 
  1 \% & 3.35 & 3.60 & 3.40 & 14.20 & 1.00 & 2.35 & 24.80 \\ 
  \hline
\end{tabular}
}
\end{table}

\textbf{Adaptive length-scale}\\
\textbf{Censoring distribution $G(t)=1-e^{-1/2t}$}

\begin{table}[H]
\centering
\resizebox{1\columnwidth}{!}{
\begin{tabular}{rrrrrrrr}
\hline
  \multicolumn{8}{c}{Adaptive length-scale}\\
  \hline
  \multicolumn{8}{c}{Sample size $n=200$; $\gamma=1/2$}\\
  \hline
 & MW1 & MW2 & MW3 & Pearson & LR1 & LR2& WLR \\ 
  \hline
  \multicolumn{8}{l}{$\theta=(0.5,1)$}\\
 10 \% & 100.00 & 100.00 & 100.00 & 100.00 & 100.00 & 100.00 & 100.00 \\ 
  5 \% & 100.00 & 100.00 & 100.00 & 100.00 & 100.00 & 100.00 & 100.00 \\ 
  1 \% & 100.00 & 100.00 & 100.00 & 100.00 & 100.00 & 100.00 & 100.00 \\ 
  \hline
     \multicolumn{8}{l}{$\theta=(0.7,1)$}\\
10 \% & 100.00 & 100.00 & 100.00 & 100.00 & 100.00 & 100.00 & 100.00 \\ 
  5 \% & 100.00 & 100.00 & 100.00 & 100.00 & 99.90 & 100.00 & 100.00 \\ 
  1 \% & 100.00 & 100.00 & 100.00 & 100.00 & 98.80 & 100.00 & 100.00 \\ 
\hline
        \multicolumn{8}{l}{$\theta=(0.9,1)$}\\
10 \% & 100.00 & 100.00 & 100.00 & 100.00 & 95.05 & 100.00 & 100.00 \\ 
  5 \% & 100.00 & 100.00 & 100.00 & 100.00 & 88.45 & 100.00 & 99.95 \\ 
  1 \% & 100.00 & 100.00 & 100.00 & 100.00 & 64.30 & 100.00 & 99.75 \\ 
   \hline
    \multicolumn{8}{l}{$\theta=(1,1)$}\\
10 \% & 100.00 & 100.00 & 100.00 & 100.00 & 83.00 & 100.00 & 100.00 \\ 
  5 \% & 100.00 & 100.00 & 100.00 & 100.00 & 71.00 & 100.00 & 99.95 \\ 
  1 \% & 100.00 & 100.00 & 100.00 & 100.00 & 39.95 & 100.00 & 99.30 \\ 
\hline
       \multicolumn{8}{l}{$\theta=(1.5,1)$}\\
10 \% & 100.00 & 100.00 & 100.00 & 100.00 & 28.50 & 98.10 & 99.25 \\ 
  5 \% & 100.00 & 100.00 & 100.00 & 100.00 & 17.70 & 94.40 & 98.05 \\ 
  1 \% & 100.00 & 100.00 & 100.00 & 100.00 & 5.20 & 73.50 & 90.55 \\ 
\hline
       \multicolumn{8}{l}{$\theta=(3,1)$}\\
10 \% & 98.20 & 98.35 & 98.30 & 44.70 & 10.70 & 21.85 & 44.35 \\ 
  5 \% & 89.75 & 89.95 & 89.85 & 31.70 & 4.95 & 11.40 & 31.45 \\ 
  1 \% & 48.90 & 48.00 & 47.85 & 14.20 & 1.00 & 2.35 & 12.85 \\ 
\hline
\end{tabular}
}
\end{table}
\end{multicols}

\begin{multicols}{2}
\textbf{Fixed length-scale 1}\\
\textbf{Censoring distribution $G(t)=1-e^{-2t}$}

\begin{table}[H]
\centering
\resizebox{1\columnwidth}{!}{
\begin{tabular}{rrrrarrr}
\hline
  \multicolumn{8}{c}{Fixed length-scale 1}\\
  \hline
  \multicolumn{8}{c}{Sample size $n=200$; $\gamma=2$}\\
  \hline
 & MW1 & MW2 & MW3 & Pearson & LR1 & LR2& WLR \\ 
  \hline
  \multicolumn{8}{l}{$\theta=(0.5,1)$}\\
  10 \% & 100.00 & 100.00 & 100.00 & 100.00 & 100.00 & 100.00 & 100.00 \\ 
  5 \% & 100.00 & 100.00 & 100.00 & 100.00 & 100.00 & 100.00 & 100.00 \\ 
  1 \% & 100.00 & 100.00 & 100.00 & 100.00 & 100.00 & 100.00 & 100.00 \\ 
 \hline
     \multicolumn{8}{l}{$\theta=(0.7,1)$}\\
10 \% & 100.00 & 100.00 & 100.00 & 100.00 & 100.00 & 100.00 & 100.00 \\ 
  5 \% & 100.00 & 100.00 & 100.00 & 100.00 & 100.00 & 100.00 & 100.00 \\ 
  1 \% & 100.00 & 100.00 & 100.00 & 100.00 & 100.00 & 100.00 & 100.00 \\ 
\hline
        \multicolumn{8}{l}{$\theta=(0.9,1)$}\\
10 \% & 100.00 & 100.00 & 100.00 & 100.00 & 100.00 & 100.00 & 99.95 \\ 
  5 \% & 100.00 & 100.00 & 100.00 & 100.00 & 100.00 & 100.00 & 99.75 \\ 
  1 \% & 100.00 & 100.00 & 100.00 & 100.00 & 99.90 & 100.00 & 98.25 \\ 
   \hline
    \multicolumn{8}{l}{$\theta=(1,1)$}\\
10 \% & 100.00 & 100.00 & 100.00 & 100.00 & 99.95 & 100.00 & 99.60 \\ 
  5 \% & 100.00 & 100.00 & 100.00 & 100.00 & 99.90 & 100.00 & 98.50 \\ 
  1 \% & 100.00 & 100.00 & 100.00 & 100.00 & 98.95 & 100.00 & 91.00 \\ 
\hline
       \multicolumn{8}{l}{$\theta=(1.5,1)$}\\
10 \% & 100.00 & 100.00 & 100.00 & 100.00 & 84.75 & 100.00 & 75.60 \\ 
  5 \% & 99.95 & 100.00 & 100.00 & 100.00 & 72.50 & 100.00 & 59.05 \\ 
  1 \% & 99.75 & 99.70 & 99.70 & 99.95 & 43.70 & 99.90 & 26.90 \\ 
\hline
       \multicolumn{8}{l}{$\theta=(3,1)$}\\
10 \% & 44.65 & 44.90 & 44.65 & 56.40 & 18.65 & 74.90 & 27.45 \\ 
  5 \% & 31.45 & 31.70 & 30.95 & 44.05 & 10.35 & 56.80 & 16.95 \\ 
  1 \% & 13.30 & 13.50 & 13.15 & 22.45 & 2.20 & 22.15 & 6.20 \\ 
\hline
\end{tabular}%
}
\end{table}
\end{multicols}
\newpage

\section{Weibull hazard functions}\label{sec:WeiHazardExperiments}
Let $\theta_w=(\theta_{w1},\theta_{w2})\in\R_+^2$.
We consider the hazard functions $\lambda_{\theta_w}=\theta_{w1}/\theta_{w2}(t/\theta_{w2})^{\theta_{w1}-1}$ for fixed $\theta_{w2}=1$ and $\theta_{w1}\in\{0.5,1,1.5,2,3\}$.  The censoring distribution is chosen to be $G(t)=1-e^{-\gamma t}$, where $\gamma$ is chosen i such way it produces $30\%$ and $50\%$ of censoring given a fixed alternative distribution. For this model, the null hypothesis is recovered when $\theta_w=(1,1)$ (shown in pink). 

\begin{tabular}{c}
   \includegraphics[width=0.9\textwidth]{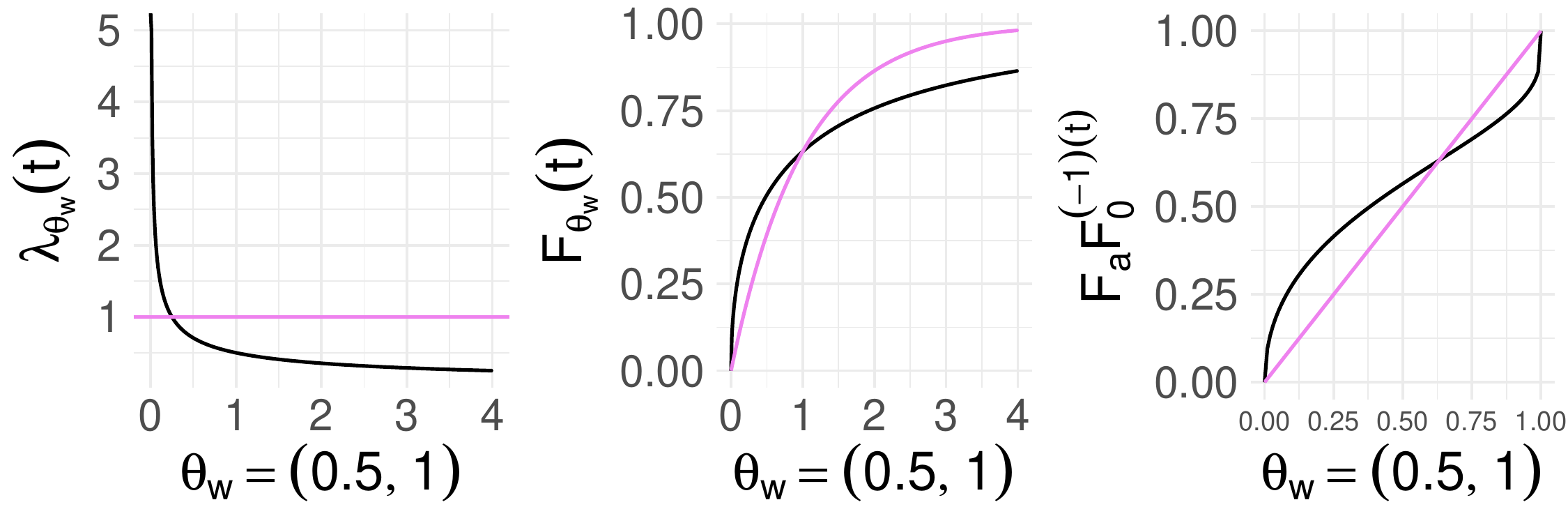}\\
        \includegraphics[width=0.9\textwidth]{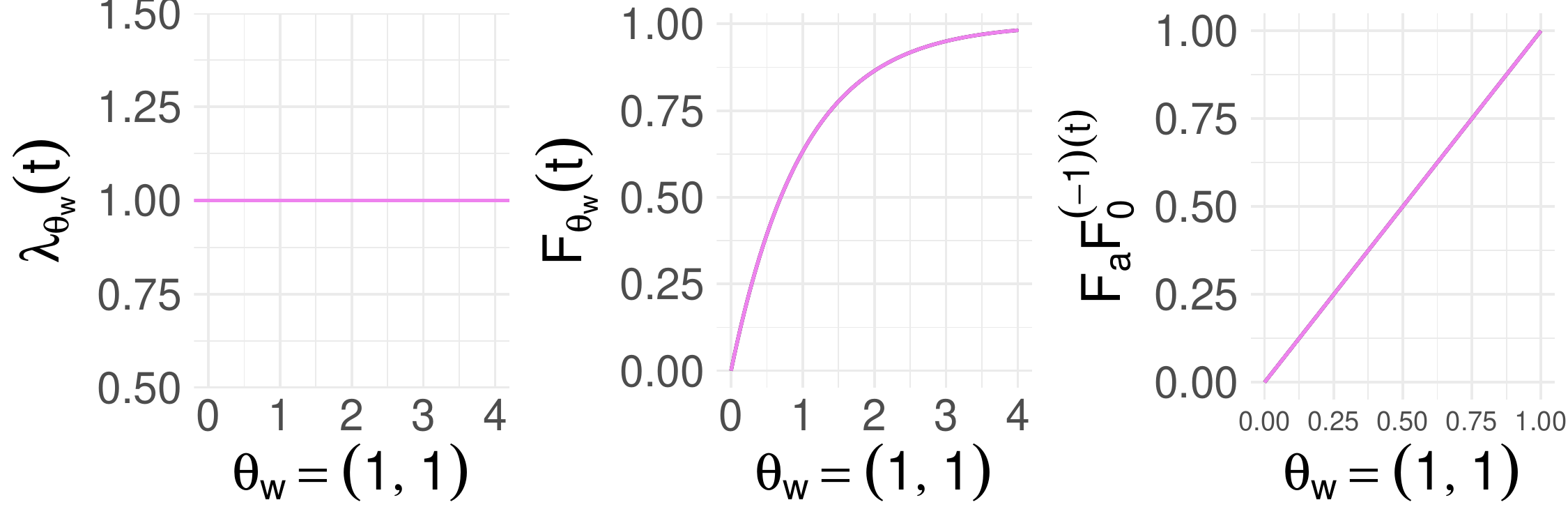}\\
        \includegraphics[width=0.9\textwidth]{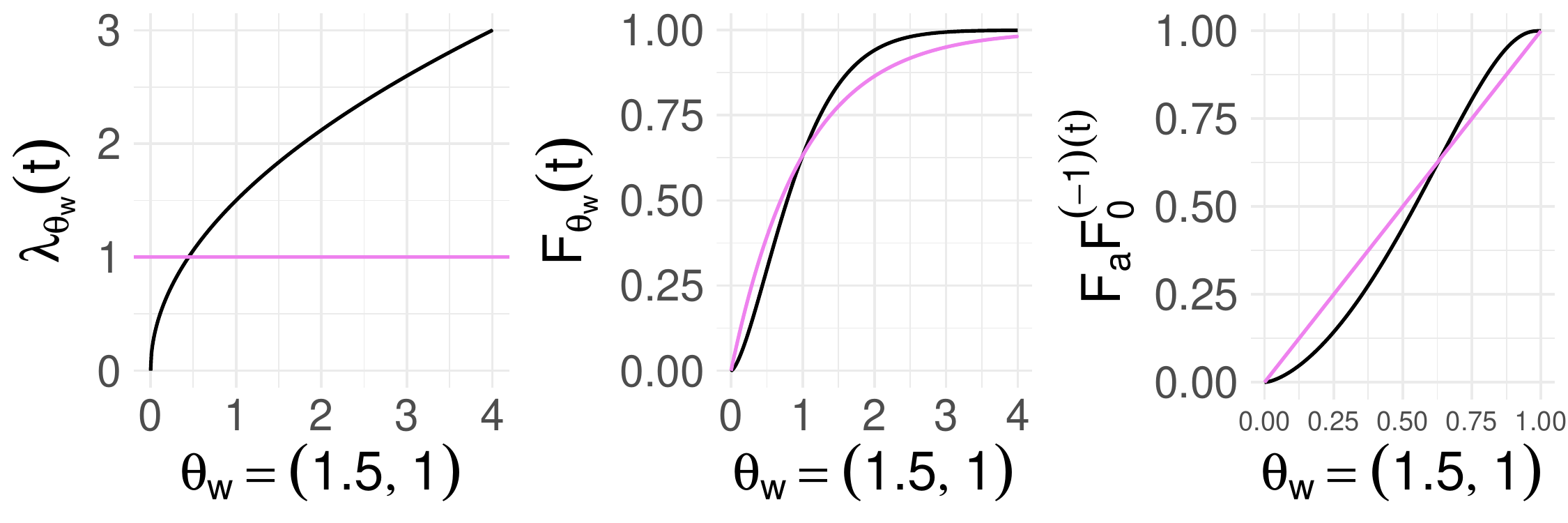}
\end{tabular}

\begin{tabular}{c}
        \includegraphics[width=0.9\textwidth]{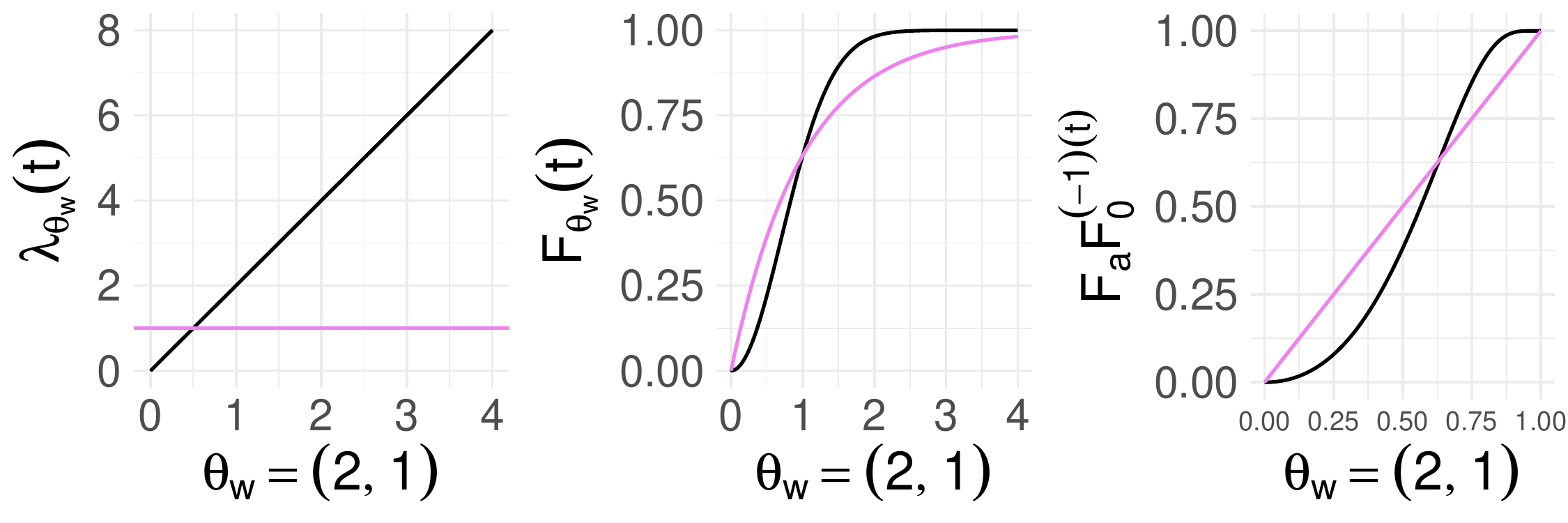}\\
   \includegraphics[width=0.9\textwidth]{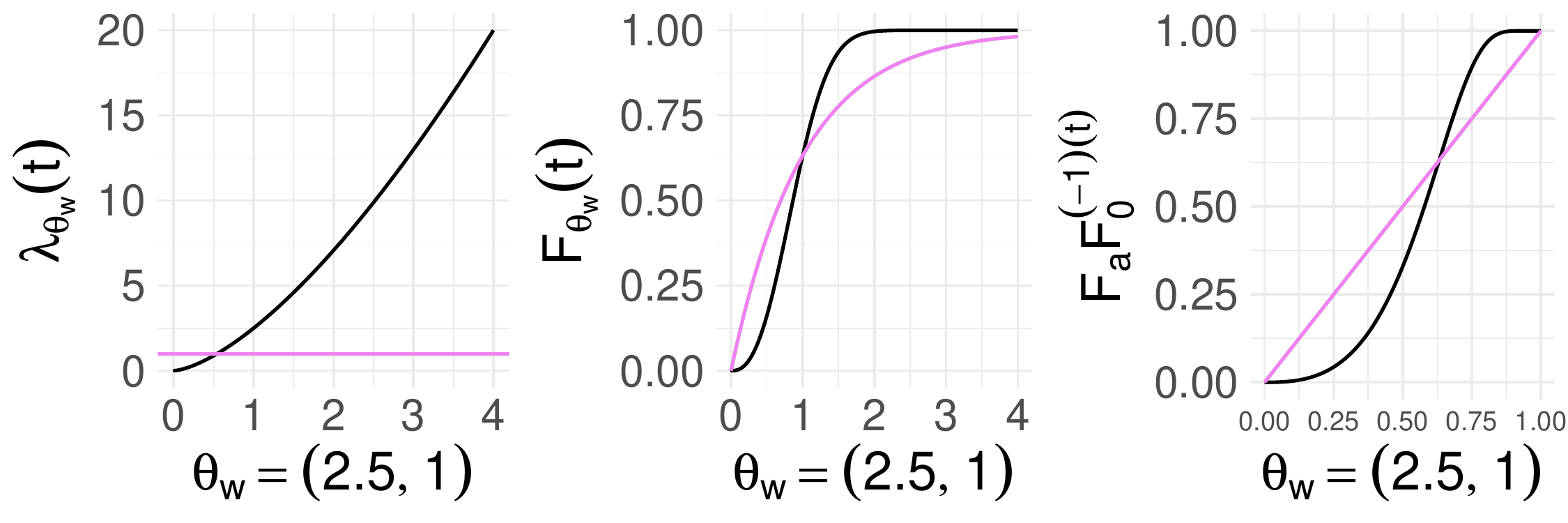}\\
        \includegraphics[width=0.9\textwidth]{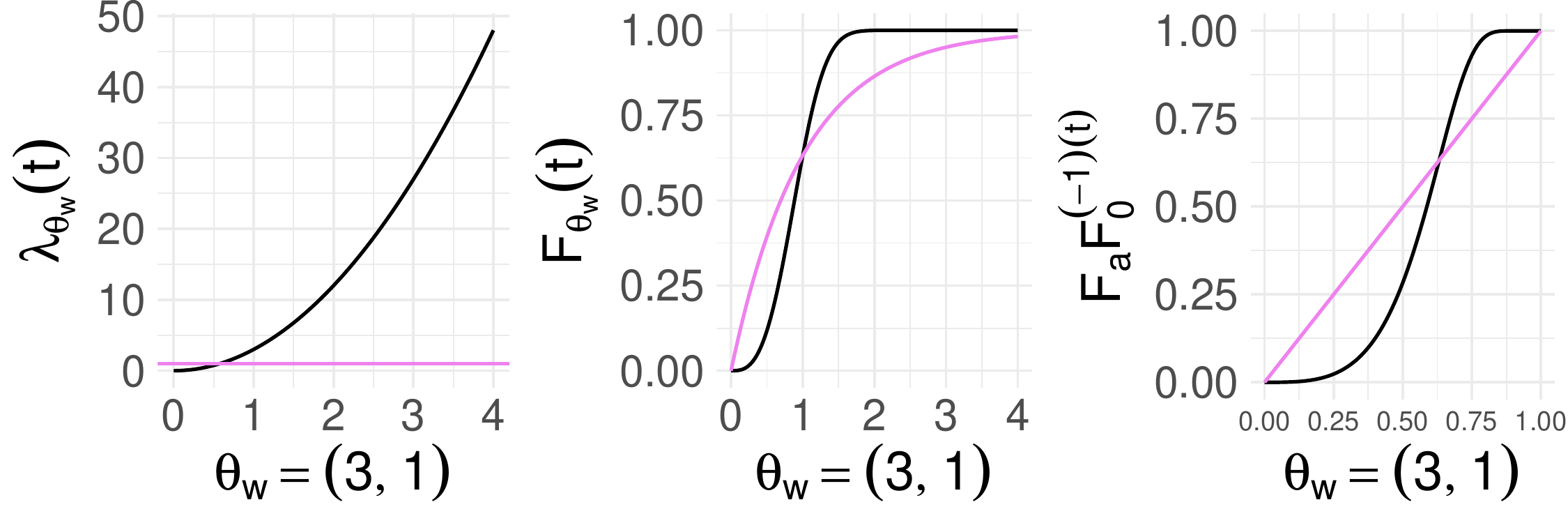}      
\end{tabular}
\newpage
\subsection{Type-I error}
The power tables include an estimation of the Type-I error for $\theta_w=(1,1)$ 
\subsection{Power tables}
In {\color{red}{red}} we observe tests that have an clear incorrect level and thus estimated power must be looked with caution. In{ \color{orange}{orange}} are tests that have a questionable incorrect level. Tables are based on $2000$ independent experiments. 

\subsubsection{Sample size $30$}
\begin{multicols}{2}
\textbf{Fixed length-scale 1}\\
\textbf{Censoring percentage $30\%$}
\begin{table}[H]
\centering
\resizebox{1\columnwidth}{!}{
\begin{tabular}{rrrrbrra}
\hline
\multicolumn{8}{c}{Weibull hazard functions}
\\\hline
\multicolumn{8}{c}{Sample size $n= 30 $; Censoring $ 30 \% $; Fixed length-scale 1}
\\\hline
& MW1 & MW2 & MW3 & Pearson & LR1 & LR2 &WLR\\ 
\hline
\multicolumn{8}{l}{$\theta_w=(0.5,1)$}\\
10 \%& 40.40 & 43.85 & 40.55 & 70.25 & 28.55 & 71.30 & 70.55 \\ 
5 \%& 26.65 & 31.35 & 26.45 & 60.15 & 20.45 & 65.15 & 63.70 \\ 
1\%& 7.75 & 12.75 & 8.05 & 42.05 & 9.60 & 51.80 & 52.55 \\ 
\hline
\multicolumn{8}{l}{$\theta_w=(1,1)$}\\
10 \%& 8.95 & 10.50 & 9.10 & 13.20 & 9.45 & 9.80 & 18.65 \\ 
5 \%& 4.65 & 6.05 & 4.80 & 8.85 & 5.10 & 5.20 & 13.30 \\ 
1\%& 1.20 & 2.05 & 1.45 & 3.60 & 1.60 & 1.45 & 6.40 \\ 
\hline
\multicolumn{8}{l}{$\theta_w=(1.5,1)$}\\
10 \%& 25.95 & 28.90 & 26.10 & 52.20 & 3.65 & 8.45 & 33.40 \\ 
5 \%& 16.75 & 19.30 & 16.65 & 40.80 & 1.55 & 3.45 & 24.10 \\ 
1\%& 5.30 & 7.80 & 5.35 & 21.10 & 0.05 & 0.25 & 11.55 \\ 
\hline
\multicolumn{8}{l}{$\theta_w=(2,1)$}\\
10 \%& 69.75 & 72.55 & 70.00 & 95.20 & 1.30 & 17.60 & 73.15 \\ 
5 \%& 53.45 & 57.85 & 53.55 & 89.65 & 0.45 & 5.95 & 63.10 \\ 
1\%& 25.65 & 30.15 & 25.45 & 72.95 & 0.00 & 0.20 & 42.00 \\ 
\hline
\multicolumn{8}{l}{$\theta_w=(2.5,1)$}\\
10 \%& 95.05 & 95.70 & 95.25 & 99.90 & 0.55 & 28.45 & 95.75 \\ 
5 \%& 87.60 & 89.90 & 87.75 & 99.80 & 0.00 & 9.40 & 92.80 \\ 
1\%& 63.20 & 69.10 & 63.75 & 98.20 & 0.00 & 0.30 & 80.35 \\ 
\hline
\multicolumn{8}{l}{$\theta_w=(3,1)$}\\
10 \%& 99.90 & 99.90 & 99.85 & 100.00 & 0.05 & 43.40 & 99.55 \\ 
5 \%& 98.90 & 99.25 & 98.75 & 100.00 & 0.00 & 16.30 & 99.00 \\ 
1\%& 89.60 & 91.40 & 89.50 & 99.90 & 0.00 & 0.45 & 96.20 \\ 
   \hline
\end{tabular}%
}
\end{table}
\end{multicols}
\newpage

\subsubsection{Sample size $50$}
\begin{multicols}{2}

\textbf{Fixed length-scale 1}\\
\textbf{Censoring percentage $30\%$}

\begin{table}[H]
\centering
\resizebox{1\columnwidth}{!}{
%
}
\end{table}
\end{multicols}
\newpage

\subsubsection{Sample size $100$}
\begin{multicols}{2}

\textbf{Fixed length-scale 1}\\
\textbf{Censoring percentage $30\%$}

\begin{table}[H]
\centering
\resizebox{1\columnwidth}{!}{
%
}
\end{table}
\end{multicols}

\newpage
\subsubsection{Sample size $200$}
\begin{multicols}{2}

\textbf{Fixed length-scale 1}\\
\textbf{Censoring percentage $30\%$}

\begin{table}[H]
\centering
\resizebox{1\columnwidth}{!}{
%
}
\end{table}
\end{multicols}

\end{document}